\documentclass[12pt]{article}
\newcommand{\showNote}{0}
\usepackage{varioref}
\usepackage{xr-hyper}
\usepackage[colorlinks]{hyperref}
\usepackage[letterpaper, hmargin=1.115in, vmargin=0.7in]{geometry}
\usepackage{booktabs, graphicx, tabularx}
\usepackage{amsthm, amsmath, amsfonts, amssymb}
\usepackage[capitalize]{cleveref}
\usepackage[nodisplayskipstretch]{setspace}
\usepackage[OT1]{fontenc}
\usepackage{authblk}
\usepackage{MnSymbol}
\usepackage{bbold}
\usepackage{xcolor}
\usepackage[normalem]{ulem}
\usepackage{enumitem, multicol}
\usepackage{array}
\usepackage{caption}
\usepackage{float}
\usepackage{multirow}
\usepackage{makecell}
\usepackage[lined,boxed,commentsnumbered,ruled,linesnumbered,algosection]{algorithm2e}
\usepackage{natbib}
\usepackage{comment}

\hypersetup{
    colorlinks = true,
    allcolors = black,
    linkcolor = blue,
    citecolor = blue,
    urlcolor = blue
}

\def\spacingset#1{\renewcommand{\baselinestretch}{#1}\small\normalsize}
\spacingset{1}

\numberwithin{equation}{section}
\allowdisplaybreaks
\theoremstyle{plain}
\newtheorem{theorem}{Theorem}
\newtheorem{proposition}{Proposition}

\newtheorem{lemma}{Lemma}

\theoremstyle{definition}

\newtheorem{example}{Example}
\newtheorem{remark}{Remark}
\newtheorem{assumption}{Assumption}

\newtheorem{innercustomas}{Assumption}
\newenvironment{customas}[1]
  {\renewcommand\theinnercustomas{#1}\innercustomas}
  {\endinnercustomas}

\setlist[enumerate]{label={(\alph*)}, nosep}
\setlist[itemize]{nosep}
\setcitestyle{authoryear,round}

\definecolor{myblue}{rgb}{0.1, 0.1, 0.7}
\hypersetup{
    hidelinks,
    linktocpage,
    colorlinks=true,
    linkcolor=myblue,
    citecolor=myblue,
    filecolor=myblue,
    urlcolor=myblue
}

\crefformat{section}{Section~#2#1#3}
\crefformat{subsection}{Section~#2#1#3}
\crefformat{subsubsection}{Section~#2#1#3}
\crefrangeformat{section}{Sections~#3#1#4 to~#5#2#6}
\crefmultiformat{section}{Sections~#2#1#3}{ and~#2#1#3}{, #2#1#3}{ and~#2#1#3}
\crefformat{equation}{Equation~#2#1#3}
\crefrangeformat{equation}{Equations~#3#1#4 to~#5#2#6}
\crefmultiformat{equation}{Equations~#2#1#3}{ and~#2#1#3}{, #2#1#3}{ and~#2#1#3}
\newcommand{\Supp}{Supplement}
\usepackage{marginnote}
\usepackage{todonotes}
\usepackage{setspace}
\ifnum\showNote=1
    \newcommand{\note}[1]{\reversemarginpar\marginnote{\tiny \hspace{.5cm}\color{red}\setstretch{.4}\begin{minipage}{0.15\textwidth}[\textsf{#1}]\end{minipage}}}
\else
    \newcommand{\note}[1]{}
\fi

\newcommand{\norm}[1]{\left\lVert #1 \right\rVert}
\newcommand{\sbk}[1]{\left( #1 \right)}
\newcommand{\mbk}[1]{\left[ #1 \right]}
\newcommand{\cbk}[1]{\left\{ #1 \right\}}
\newcommand{\abs}[1]{\left| #1 \right|}
\newcommand{\calE}{\mathcal{E}}

\newcommand{\T}{\intercal}
\newcommand{\argmin}{\mathop{\mathrm{argmin}}}

\newcommand{\tr}{\mathrm{tr}}

\newcommand{\Z}{\mathbb{Z}}
\newcommand{\R}{\mathbb{R}}

\newcommand{\pr}{\mathbb{P}}
\newcommand{\I}{\mathbb{1}}
\newcommand{\E}{\mathbb{E}}
\newcommand{\Var}{\mathrm{Var}}

\newcommand{\cin}[1]{\stackrel{\mathrm{#1}}{\to}}

\newcommand{\simIID}{\stackrel{\mathrm{iid}}{\sim}}
\newcommand{\Corr}{\mathrm{Corr}}

\newcommand{\Normal}{\mathrm{N}}

\newcommand{\GMM}{\mathsf{GMM}}
\newcommand{\OGMM}{\mathsf{OGMM}}
\newcommand{\SGMM}{\mathsf{SGMM}}
\newcommand{\GMWM}{\mathsf{GMWM}}

\newcommand{\TSLS}{\mathsf{2SLS}}
\newcommand{\LEQR}{\mathsf{LEQR}}

\newcommand{\SLIM}{\mathsf{SLIM}}
\newcommand{\SSGD}{\mathsf{SSGD}}
\newcommand{\CQER}{\mathsf{CQER}}

\newcommand{\unrestricted}{\mathrm{U}}
\newcommand{\full}{\mathrm{F}}

\newcommand{\RV}{\mathrm{RV}}
\newcommand{\RQ}{\mathrm{RQ}}

\newcommand{\paperTitle}{Online Generalized Method of Moments for Time Series}

\begin{document}
\title{\bf \paperTitle}
\newcommand*\samethanks[1][\value{footnote}]{\footnotemark[#1]}
\author[1]{Man Fung Leung\thanks{These authors contributed equally to this work.}}
\author[2]{Zhuohua Shen\samethanks}
\author[3]{Kin Wai Chan\thanks{Corresponding author: \href{kinwaichan@sta.cuhk.edu.hk}{\nolinkurl{kinwaichan@sta.cuhk.edu.hk}}}}
\author[4]{Xiaofeng Shao}
\affil[1]{Department of Statistics, University of Illinois at Urbana-Champaign}
\affil[2,3]{Department of Statistics and Data Science, The Chinese University of Hong Kong}
\affil[4]{Department of Statistics and Data Science, and Department of Economics, Washington University in St. Louis}
\date{}
\maketitle

\bigskip

\begin{abstract}
Online learning has gained popularity in recent years due to the urgent need to analyze large-scale streaming data, which can be collected in perpetuity and may be serially dependent.
This motivates us to develop the online generalized method of moments ($\OGMM$), an explicitly updated estimation and inference framework in the time series setting.
The $\OGMM$ inherits many properties of offline $\GMM$,
such as its broad applicability to many problems in econometrics and statistics,
natural accommodation for over-identification,
and achievement of semiparametric efficiency under temporal dependence.
As an online method, the key gain relative to offline $\GMM$ is the vast improvement in time complexity and memory requirement.

Building on the $\OGMM$ framework, we propose improved versions of online Sargan--Hansen and structural stability tests following recent work in econometrics and statistics.
Through Monte Carlo simulations, we observe encouraging finite-sample performance in online instrumental variables regression, online over-identifying restrictions test, online quantile regression, and online anomaly detection.
Interesting applications to stochastic volatility modeling and inertial sensor calibration are presented to demonstrate the effectiveness of our proposal.
\end{abstract}


\vfill
\newpage
\spacingset{1.675}

\section{Introduction} \label{sec:introduction}
\enlargethispage{\baselineskip}
The generalized method of moments ($\GMM$, \citealp{hansen1982gmm}) is a fundamental estimation and inference framework in econometrics and statistics.
It is widely used to analyze economic and financial data, and subsumes many popular statistical methods such as least squares, maximum likelihood, and instrumental variables regression \citep{hall2005gmm}.
However, despite being partly motivated by situations in which asymptotically efficient estimation (e.g., maximum likelihood) is computationally burdensome \citep{hansen1982gmm,hall2005gmm}, standard offline implementations of $\GMM$ may not scale well.
For large-scale streaming data, which can be collected in perpetuity and serially dependent, the classical (offline) $\GMM$ is too computationally expensive and may fail due to violation of memory constraints \citep{chen2023sgmm}.

With the emergence of streaming data collection techniques, it has become important to develop online estimation and inference procedures that can be computationally (and memory) efficient while preserving statistical properties of their offline counterparts.
In the context of $\GMM$, \citet{chen2023sgmm} adapted the stochastic approximation \citep{robbins1951sgd} to linear instrumental variables regression for independent data and provided plug-in and self-normalized inference procedures.
They also discussed learning rate selection, multi-epoch estimation, and nonlinear $\GMM$.
For nonlinear $\GMM$, \citet{chen2025slim} further developed a modular stochastic approximation framework using independent mini-batches of moments and their derivatives with first-order algorithm and second-order refinement. For online quantile regression,
\citet{chen2019rindicator} linearized estimating equations for smoothed quantile regression with some initial estimators, which correspond to $\GMM$ with exact identification. \citet{luo2020renew} developed renewable estimation and incremental inference for generalized linear models with independent and identically distributed (iid) data. Their proposal is also based on linearized estimating equations but uses implicitly updated estimators. Subsequently, \citet{luo2022qif} extended renewable estimation to longitudinal data analysis using the quadratic inference function \citep{qu2000qif}, and discussed anomaly detection based on a modified Sargan--Hansen test.
In a related work, \citet{luo2023qif} focused on streaming longitudinal data with a fixed set of participants.
Under a first-order autoregressive working correlation structure, they proposed a new decomposition for the quadratic inference function that allows online updates.
They also considered time-varying parameter estimation by introducing an exponential smoothing factor that dynamically adjusts the weights applied to historical data batches.

The overarching goal of this article is to develop the online generalized method of moments ($\OGMM$), a general and efficient estimation and inference framework for streaming time series, which seems lacking in the literature.
Using a lagged linearization (i.e., a linearization step with the $\OGMM$ estimator obtained at the previous step plugged in),
we derive an explicitly updated estimator in contrast to implicitly updated renewable estimators in the literature \citep{schifano2016renew,luo2020renew,luo2022qif,luo2023qif};
see \citet{toulis2017implicit} for the difference between explicit and implicit updates in the context of stochastic approximation.
The local linearization device itself is not new, and we use it for a batch-recursive implementation of $\GMM$ rather than as a new optimization principle.
To achieve semiparametric efficiency, we propose inverting the online long-run variance (LRV) estimator, such as that of \citet{rlrv}, for the optimal weighting matrix.
This allows us to perform efficient estimation and inference for the parameter in a fully online fashion.
Our major contributions include:

\begin{enumerate}
	\item We develop an explicit online $\GMM$ estimator in the time series setting. 
    The lagged linearization step for performing online Taylor expansion with cumulative online summary statistics makes the recursive $\GMM$ update explicit for batched time series data,
    unlike other existing online estimators that typically require implicit updates. Theoretically, under local assumptions, we show that it is consistent, asymptotically normal, and achieves the same asymptotic efficiency as the offline $\GMM$.
    Additionally, we investigate the theoretical and practical differences among various linearized updates, which appear largely unexplored in the literature;

	\item Our framework is autocorrelation-robust by handling the serial dependence in a nonparametric fashion. The semiparametric efficiency is thus achieved through online estimation of the optimal weighting matrix. Other than the user-chosen parameters involved in the online updates of the weighting matrix, our algorithm does not involve any tuning parameter such as the learning rate in stochastic approximation;

	\item We propose a general principle to extend offline $\GMM$ inference procedures with $\OGMM$.
    Our new versions of online Sargan--Hansen and structural stability tests
    exhibit substantially smaller size distortion in our experiments compared to those of \citep{chen2023sgmm}, and,
    unlike existing methods in the literature \citep{luo2022qif}, they do not require historical raw data;

	\item Owing to the broad framework of $\GMM$, we can include online one-step estimation, online least squares, online quantile regression and online instrumental variables regression as special cases
    commonly found in existing applications.
    Numerical comparison with several existing methods shows that the finite-sample performance of our $\OGMM$ is very competitive.
\end{enumerate}

The rest of this paper is organized as follows.
\cref{sec:review} reviews relevant literature and introduces some notation.
\cref{sec:method} formulates $\OGMM$ and shows that it covers many existing methods as special cases.
\cref{sec:theory} develops the asymptotic theory of $\OGMM$.
\cref{sec:simulations} compares $\OGMM$ with some existing methods in Monte Carlo experiments.
\cref{sec:applications} presents two applications in stochastic volatility modeling and inertial sensor calibration.
\cref{sec:discussion} discusses our findings and some future directions.
All proofs and some additional results are deferred to the online \Supp \ \citep{supp}.

\section{Literature review} \label{sec:review}

Consider time series data $\{ D_i = (x_{i,1}, x_{i,2}, \ldots, x_{i,n_i}) \}_{i=1}^b$ arriving in batches,
where $x_{i,j} \in \R^d$ is the $j$-th vector of observations in the $i$-th batch of data with sample size $n_i$, and
it is stationary and ergodic within and across batches; see \cref{sec:applications} for some real data examples.
Let $\theta^* \in \Theta \subset \R^p$ be a vector of unknown parameters which are to be estimated, and
$g(\theta, x)$ be a $q \times 1$ vector of functions that satisfies the population moment condition
$\E\{ g(\theta^*, x_{i,j}) \} = 0_q$ for all $i,j$ with $q \ge p$, where $0_q$ is a zero vector of length $q$.
To simplify the notation, write $G(\theta; D_i) = \sum_{j=1}^{n_i} g(\theta, x_{i,j})$ and $N_b = \sum_{i=1}^b n_i$.
The celebrated $\GMM$ estimator of $\theta^*$ is
\begin{equation} \label{eq:gmm}
	\tilde\theta_{b, \GMM}
	= \argmin_{\theta \in \Theta} \bar{g}_b(\theta)^\T \tilde{W} \bar{g}_b(\theta),
\end{equation}
where $\bar{g}_b(\theta) = N_b^{-1} \sum_{i=1}^b G(\theta; D_i)$ is the sample moment and $\tilde{W}\in\R^{q\times q}$ is a positive semi-definite weighting matrix.
Under suitable conditions (\citealt{hansen1982gmm}), it is well known that \eqref{eq:gmm} is asymptotically optimal when $\tilde{W} = \tilde\Sigma^{-1}$,
where $\tilde\Sigma\in\R^{q\times q}$ is a consistent estimator of the long-run covariance matrix
\begin{equation} \label{eq:lrv}
	\Sigma = \lim_{N_b \to \infty} N_b \Var\{ \bar{g}_b(\theta^*) \}.
\end{equation}
In this case, $\sqrt{N_b}(\tilde\theta_{b, \GMM} - \theta^*) \cin{d} \Normal(0, (V^\T \Sigma^{-1} V)^{-1})$, where $V = \E\{ \nabla g(\theta^*, x_{i,j}) \} \in\R^{q\times p}$ is the expectation of the gradient of $g(\theta^*, x_{i,j})$ with respect to $\theta$, and $\cin{d}$ denotes convergence in distribution.
However, as $\Sigma$ depends on the unknown $\theta^*$, substitutes of $\theta^*$ are necessary for efficient estimation.
Classical variants of $\GMM$ such as two-step, iterated and continuously updating differ in the way to substitute $\theta^*$; see \citet{hall2005gmm} and the references therein.

There is a rich literature on the estimation of \eqref{eq:lrv} because the LRV is an important quantity in time series analysis.
However, classical nonparametric estimators utilizing the overlapping batch means \citep{meketon1984obm} or kernels/lag windows \citep{parzen1957kernel,andrews1991kernel} cannot be updated online.
\citet{wu2009recursive} and \citet{rtavc,rtacm} proposed different online alternatives based on subsample selection rules, which replace the constant batch size of a batch means estimator with a sequence of batch sizes.
\citet{rlrv} developed a different approach by decomposing kernels, which led to estimators with higher statistical and computational efficiency.

To update $\GMM$ and obtain online estimates, the stochastic approximation \citep{robbins1951sgd} is a natural candidate.
\citet{chen2023sgmm} appears to be the first to extend the stochastic approximation to the $\GMM$ setting, and they developed the stochastic $\GMM$ ($\SGMM$) for independent data.
They also proposed an online Sargan--Hansen test statistic and state its asymptotic distribution under the null hypothesis.
However, in finite samples, their online Sargan--Hansen test appears to suffer from considerable size distortion.
In addition, since $\SGMM$ is based on explicit stochastic approximation, its performance can also be sensitive to the choice of initial learning rate \citep{toulis2017implicit}.
We will revisit each of these issues and show numerical evidence in \cref{sec:simulations}.

Another candidate for updating $\GMM$ is linearization.
Linearization of estimating equations is a classical technique in econometrics, both for computing roots and iterative approximations of estimating equations \citep{robinson1988stochastic} and for establishing large-sample theory through local expansions and stochastic equicontinuity \citep{andrews1994asymptotics}.
In $\GMM$, linearization of the sample moments underlies standard
asymptotic normality theory \citep{hansen1982gmm,pakes1989simulation,newey1994asymptotic},
and also appears algorithmically through both Gauss--Newton procedures and Newton--Raphson one-step refinements \citep{robinson1988stochastic,andrews1997stopping,forneron2025convexity}.
For online $\GMM$ update, the same idea suggests that
one stores a finite number of statistics after first-order Taylor approximations so that the estimating equations only depend on the stored statistics and the latest batch of data.
In \citet{chen2019rindicator}, some historical observations are removed from the summary statistics and the updates are explicitly based on some initial estimators.
The finite-sample statistical efficiency is affected because data are not fully utilized; see \cref{sec:simulations}.
In \citet{luo2020renew}, observations are not removed but the updates are implicit.
The actual computational cost can vary according to the objective function as the time to meet the stopping criteria may vary.
\citet{luo2022qif,luo2023qif} extended \citet{luo2020renew} to longitudinal data analysis using quadratic inference function,
which allows for over-identification due to the inverse working correlation matrix being approximated by a linear combination of some basis matrices.
Nevertheless, the anomaly detection statistic in \citet{luo2022qif} requires storing a reference data batch and solving an optimization problem for every new data batch.
A fully online detector remains to be developed.
Table \ref{tab:properties} highlights some key differences between our proposal and some existing ones.

\begin{table}[!t]
	\centering
    \resizebox{\textwidth}{!}{
    \begin{tabular}{lccccc}
			\toprule[2pt]
			Method & Data assumption & Over-identification & Update & Structural stability \\
			\midrule[1pt]
			\citet{chen2019rindicator} & independent & not allowed & explicit & not considered \\
			$\SGMM$ \citep{chen2023sgmm} & independent & allowed & explicit & not considered \\
			\citet{luo2020renew} & iid & not allowed & implicit & not considered \\
			\citet{luo2022qif} & independent & under working corr. & implicit & offline type I \\
			\citet{luo2023qif} & first-order autoregressive & under working corr. & implicit & not applicable \\
			$\OGMM$ (proposal) & stationary \& ergodic & allowed & explicit & online types I \& O \\
			\bottomrule[2pt]
	\end{tabular}}
	\caption{Summary of properties of different online estimation and inference frameworks.
	I, identifying (caused by a change point in parameter value);
	O, over-identifying (caused by a model misspecification);
	see \citet[Section 3.2]{luo2023qif} considered time-varying parameters, so structural stability testing does not apply.}
	\label{tab:properties}
\end{table}

Finally, we introduce some notation.
Denote the gradient operation with respect to $\theta$ by $\nabla$.
For any vector $a$ and matrix $A$, denote the $l^2$-vector norm and Frobenius norm by $\norm{a}_2 = \sqrt{a^\T a}$ and $\norm{A}_F = \sqrt{\tr(A^\T A)}$, respectively.
Denote the identity matrix of size $q$ by $I_q$.
We reserve the tilde, circumflex and check for offline, $\OGMM$ and other online methods (e.g., $\tilde\theta_{b, \GMM}$, $\hat\theta_{b, \OGMM}$ and $\check\theta_{k, \SGMM}$), respectively.
We may also suppress the subscript if the context is clear.

\section{Methodology} \label{sec:method}

\subsection{Explicit online estimation} \label{sec:online-moment}

Let $\hat\theta_1$ be an initial estimate based on $D_1$, $\hat\theta_b$ be the $\OGMM$ estimate updated using $D_b$ (which we are going to derive for $b = 2, 3, \ldots$), and $\hat{W}$ be a weighting matrix.
For $b = 1, 2, \ldots$, define
\begin{align*}
	V_b
	&= \frac{1}{N_b} \sum_{i=1}^b \nabla G(\hat\theta_i; D_i)
	= \frac{1}{N_b} \left\{ N_{b-1} V_{b-1} +\nabla G(\hat\theta_b; D_b) \right\} \quad \text{and} \\
	U_b
	&= \frac{1}{N_b} \sum_{i=1}^b \left\{ G(\hat\theta_i; D_i) +\nabla G(\hat\theta_i; D_i) (\hat\theta_b -\hat\theta_i) \right\} \\
	&= \frac{1}{N_b} \left\{ N_{b-1} U_{b-1} +N_{b-1} V_{b-1} (\hat\theta_b -\hat\theta_{b-1}) +G(\hat\theta_b; D_b) \right\},
\end{align*}
which are estimators of $V = \E\{ \nabla g(\theta^*, x_{i,j}) \}$ and $U = \E\{ g(\theta^*, x_{i,j}) \}$, respectively.
When $b = 2$, linearizing the first-order optimality condition of \eqref{eq:gmm} leads to
\begin{equation} \label{eq:ogmm-foc}
	V_2^\T \hat{W} U_2
	= \frac{1}{N_2^2} \left\{ N_1 V_1 +\nabla G(\hat\theta_2; D_2) \right\}^\T \hat{W} \left\{ N_1 U_1 +N_1 V_1 (\hat\theta_2 -\hat\theta_1) +G(\hat\theta_2; D_2) \right\}
	= 0_p.
\end{equation}
Solving \eqref{eq:ogmm-foc} for $\hat\theta_2$ explicitly is not possible because $U_2$ and $V_2$ depend on $\hat\theta_2$.
In light of it, we propose to estimate $V$ and $U$ by
\begin{align*}
	\hat{V}_b
	&= \frac{1}{N_b} \left\{ N_{b-1} V_{b-1} +\nabla G(\hat\theta_{b-1}; D_b) \right\}
	\quad \text{and} \quad
	\hat{U}_b
	= \frac{1}{N_b} \left\{ N_{b-1} U_{b-1} +G(\hat\theta_{b-1}; D_b) \right\},
\end{align*}
respectively.
Essentially, we replace $\hat\theta_b$ by $\hat\theta_{b-1}$ in $V_b$ and $U_b$ to obtain $\hat{V}_b$ and $\hat{U}_b$, respectively.
By the lagged $\hat\theta_1$ for linearization and replacing $N_1 V_1 (\hat\theta_2 -\hat\theta_1)$ with $N_2 \hat{V}_2 (\hat\theta_2 -\hat\theta_1)$, we have
\begin{equation*}
\hat{V}_2^\T \hat{W} \{ \hat{U}_2 +\hat{V}_2 (\hat\theta_2 -\hat\theta_1) \} = 0. \tag{\ref*{eq:ogmm-foc}$'$}
\end{equation*}
The solution $\hat\theta_2 = \hat\theta_1 -(\hat{V}_2^\T \hat{W} \hat{V}_2)^{-1} \hat{V}_2^\T \hat{W} \hat{U}_2$ is explicit and does not involve any tuning parameter.
These are also true for $b > 2$ with
\begin{equation} \label{eq:ogmm-linear}
	\hat\theta_{b, \OGMM}
	= \hat\theta_{b-1, \OGMM} -(\hat{V}_b^\T \hat{W} \hat{V}_b)^{-1} \hat{V}_b^\T \hat{W} \hat{U}_b
\end{equation}
and $\hat{W}=\hat{W}(\{D_i\}_{i=1}^{b-1})$;
see Algorithm \ref{algo:ogmm} for the recursive steps involved in the computation of $\hat\theta_{b, \OGMM}$, where $b$ is the running batch index, which also tells the total number of currently available batches.
Note that Algorithm~\ref{algo:ogmm} is an online recursion rather than an optimizer with fixed sample size and an internal convergence stopping rule.
So, when the observed stream contains $b$ batches, the reported estimator is $\hat\theta_{b}$.
We remark that it is possible to use a finite and
fixed number of lags of
observations to redefine $x_{i,j}$ so that
$g(\theta,x_{i,j})$ can be used to
define moments across lags;
see, e.g., \cref{sec:nonlinear,sec:applications}.

\begin{algorithm}[!t]
	\caption{Online Generalized Method of Moments} \label{algo:ogmm}
	\SetAlgoVlined
	\DontPrintSemicolon
	\SetNlSty{texttt}{[}{]}
	\small
	\textbf{initialization}: \;
	Set $b = 1$ and $N_1 = n_1$ \;
	Compute $\hat\theta_1$ using $D_1$ and any reasonable method (e.g., $\GMM$) \;
	Initialize an online LRV estimator $\hat\Sigma_1$ using $\hat\theta_1$ and set $\hat{W} = \hat\Sigma_1^{-1}$ \;
	Set $U_1 = \bar{g}_1(\hat\theta_1)$ and $V_1 = \nabla \bar{g}_1(\hat\theta_1)$ \;
	\Begin{
		Set $b = b + 1$ \;
		Receive $D_b$ and set $N_b = N_{b-1} + n_b$ \;
		Set $\hat{U}_b = N_b^{-1} \{ N_{b-1} U_{b-1} +G(\hat\theta_{b-1}; D_b) \}$ \;
		Set $\hat{V}_b = N_b^{-1} \{ N_{b-1} V_{b-1} +\nabla G(\hat\theta_{b-1}; D_b) \}$ \;
		Compute $\hat\theta_b = \hat\theta_{b-1} -(\hat{V}_b^\T \hat{W} \hat{V}_b)^{-1} \hat{V}_b^\T \hat{W} \hat{U}_b$ \;
		(Optional) update $\hat\Sigma_{b-1}$ using $\hat\theta_b$ and set $\hat{W} = \hat\Sigma_b^{-1}$ \;
		Set $U_b = N_b^{-1} \{ N_{b-1} U_{b-1} +N_{b-1} V_{b-1} (\hat\theta_b -\hat\theta_{b-1}) +G(\hat\theta_b; D_b) \}$ \;
		Set $V_b = N_b^{-1} \{ N_{b-1} V_{b-1} +\nabla G(\hat\theta_b; D_b) \}$ \;
	}
	\Return the updated $\OGMM$ estimator $\hat\theta_{b, \OGMM} \equiv \hat\theta_b$, and the statistics for the next update: $b$, $N_b$, $U_b$, $V_b$ \;
\end{algorithm}

\subsection{Principle-driven online inference} \label{sec:online-inference}

Within the offline $\GMM$ framework, different inference procedures have already been proposed to address a wide range of problems.
We propose ``online utilization'' as a common principle to extend them, i.e., we replace offline statistics with online ones that utilize the cumulative data.
Here, the online use of all the data is important because there can be a loss of statistical efficiency if the online statistics are only based on a fixed sample.
Therefore, we also consider updating an online nonparametric LRV estimator $\hat\Sigma_b$ in Algorithm \ref{algo:ogmm}; see Remark \ref{rmk:lrv}.
By using its inverse as the weighting matrix, the asymptotic variance of $\hat\theta_b$ is minimized.
Some online inference examples using $\hat\Sigma_b$ and stored statistics in Algorithm \ref{algo:ogmm} are given below.

\begin{example}[Over-identifying restrictions testing and confidence region estimation] \label{eg:overident}
	Let $\chi_{p, \alpha}^2$ be the $\alpha$-quantile of a chi-squared distribution with $p$ degrees of freedom.
	The offline Sargan--Hansen test statistic is $N_b \bar{g}_b(\tilde\theta_{b, \GMM})^\T \tilde\Sigma^{-1} \bar{g}_b(\tilde\theta_{b, \GMM})$, where $\tilde\Sigma$ is an offline estimator of $\Sigma$.
	Since $U_b$ estimates $\E\{g(\theta^*, x_{i,j})\}$ and utilizes $D_b$ compared with $\hat{U}_b$, a natural online test statistic is $T = N_b U_b^\T \hat\Sigma_b^{-1} U_b$.
	We reject $H_0: \E\{g(\theta^*, x_{i,j})\} = 0$ at $100\alpha\%$ nominal level if $T > \chi_{q-p, 1-\alpha}^2$.
	Similarly, an online $100(1-\alpha)\%$ confidence region for $\theta^*$ is $\{\theta\in\mathbb{R}^p : N_b (\hat\theta_{b, \OGMM}-\theta)^\T (V_b^\T \hat\Sigma_b^{-1} V_b) (\hat\theta_{b, \OGMM}-\theta) \le \chi_{p, 1-\alpha}^2 \}$; see Theorem \ref{thm:normality-par} for details.
\end{example}

		\begin{example}[Structural stability testing] \label{eg:stability}
		Suppose we are interested in testing $H_0: \E\{g(\theta^*, x_{i,j})\} = \E\{g(\theta^*, x_{b,h})\} = 0_q$ against $H_1: \E\{g(\theta^*, x_{i,j})\} = 0_q$ but $\E\{g(\theta^*, x_{b,h})\} \ne 0_q$ ($i = 1, \ldots, b-1$; $j=1,\ldots,n_i$; $h=1,\ldots,n_b$).
		In the context of $\GMM$, this is structural stability testing with a known break point.
		Let $\Sigma_i(\theta)$ and $P_i(\theta)$ be the subsample analogues of $\Sigma(\theta) = \lim_{N_b \to \infty} N_b \Var\{ \bar{g}_b(\theta) \}$ and $P(\theta) = \Sigma(\theta)^{-1/2} V(\theta) \{V(\theta)^\T \Sigma(\theta)^{-1} V(\theta)\}^{-1} V(\theta)^\T \Sigma(\theta)^{-1/2}$ for $D_i$, respectively, where $V(\theta) = \E\{ \nabla g(\theta, x_{i,j}) \}$.
		\citet{hall1999stability} decomposed $H_0$ into
		\begin{align*}
			H_{0,\text{I}}
			&: P_i(\theta^*) \Sigma_i(\theta^*)^{-1/2} \E\{g(\theta^*, x_{i,j})\}
			= P_b(\theta^*) \Sigma_b(\theta^*)^{-1/2} \E\{g(\theta^*, x_{b,h})\}
			= 0_q, \quad \text{and} \\
			H_{0,\text{O}}
			&: \{I_q -P_i(\theta^*)\} \Sigma_i(\theta^*)^{-1/2} \E\{g(\theta^*, x_{i,j})\} = \{I_q -P_b(\theta^*)\} \Sigma_b(\theta^*)^{-1/2} \E\{g(\theta^*, x_{b,h})\} = 0_q.
		\end{align*}
		If $H_0$ is rejected and $\E\{g(\theta^*, x_{b,h})\} \ne 0_q$, there are two possibilities.
		(Type I) There is some unique $\theta' \in \Theta \setminus \{\theta^*\}$ that satisfies $\E\{g(\theta', x_{b,h})\} = 0_q$, which means that the instability is caused by a change point in parameter value. In this case, only $H_{0,\text{I}}$ is violated.
		(Type O) There is no $\theta \in \Theta$ that satisfies $\E\{g(\theta, x_{b,h})\} = 0_q$, which means that the instability is caused by a more fundamental misspecification. In this case, $H_{0,\text{O}}$ and likely $H_{0,\text{I}}$ are both violated.
		To distinguish between these two cases, based on \citet{luo2022qif} and \citet{hall1999stability}, we propose
	\begin{align}
		T_\full
		&= N_{b-1} U_{b-1,\full}^\T \hat\Sigma_{b-1}^{-1} U_{b-1,\full}
		+n_b^{-1} G(\hat\theta_{b,\full}; D_b)^\T \hat\Sigma_{b-1}^{-1} G(\hat\theta_{b,\full}; D_b), \label{eq:stable-full} \\
		T_\unrestricted
		&= N_{b-1} U_{b-1,\unrestricted}^\T \hat\Sigma_{b-1}^{-1} U_{b-1,\unrestricted}
		+n_b^{-1} G(\tilde\theta_{b,\unrestricted}; D_b)^\T \tilde\Sigma_{b,\unrestricted}^{-1} G(\tilde\theta_{b,\unrestricted}; D_b), \label{eq:stable-unrestricted}
	\end{align}
	where $\hat\theta_{b,\full}$ is the (full-sample) $\OGMM$ estimator updated from $\hat\theta_{b-1}$ with $D_b$, $\hat\Sigma_{b-1}$ is the estimated LRV before receiving $D_b$ (assuming the serial dependence structure does not change over time), $\tilde\theta_{b,\unrestricted}$ and $\tilde\Sigma_{b,\unrestricted}$ are (unrestricted) $\GMM$ and LRV estimators computed with $D_b$ only, $U_{b-1,\full} = U_{b-1} +V_{b-1} (\hat\theta_{b,\full} -\hat\theta_{b-1})$, and $U_{b-1,\unrestricted} = U_{b-1}$ in Algorithm \ref{algo:ogmm}.
	The use of $\hat\theta_{b,\full}$ is inspired by a common practice in structural stability testing; see \citet[pages~173--174]{hall2005gmm} for additional references and discussion.
	It follows that \eqref{eq:stable-full} does not require solving another optimization problem, and the computational complexities of \eqref{eq:stable-full} and \eqref{eq:stable-unrestricted} depend on $D_b$ only, which is more efficient than \citet{luo2022qif}.
	Under the null hypothesis, \eqref{eq:stable-full} and \eqref{eq:stable-unrestricted} are asymptotically independent and converge in distribution to $\chi_{2q-p}^2$ and $\chi_{2(q-p)}^2$, respectively.
	Therefore, given that \eqref{eq:stable-full} is rejected, we may believe that there is a change point in parameter value if \eqref{eq:stable-unrestricted} is not rejected, and the model is misspecified if \eqref{eq:stable-unrestricted} is rejected.
	The unknown break point case seems more challenging and is left for future research.
\end{example}

\begin{remark}[LRV estimation] \label{rmk:lrv}
	As we mentioned in \cref{sec:review}, there is a literature on online nonparametric LRV estimation.
	Here we consider the latest proposal in \citet{rlrv} and leave all the methodological details to the \Supp\ to save space.
\end{remark}

\subsection{Connections and differences} \label{sec:implementation}

We have shown that estimation and inference using $\OGMM$ can be fully online.
Table~\ref{tab:complexity} gives a direct comparison per newly arrived batch, which separates the point estimation update from the additional cost of weighting matrix or LRV updates.
More precisely, $\SGMM$ processes a newly arrived batch through many stochastic recursions, whereas $\OGMM$ aggregates the batch and then performs one explicit $\GMM$ update.
The recursive implementation of $\SGMM$ uses the Sherman--Morrison--Woodbury formula for the inverse update, hence it does not need to repeatedly store and invert a full large matrix; see equation (5) and Remark~11 in \citet{chen2023sgmm}.
Similarly, since the Hessian $\hat H_b=\hat V_b^\top\hat W_b\hat V_b$ can be formed from the current summaries, the current $\OGMM$ does not store the full historical Hessian as well, and the dominant storage is the $q\times q$ weighting or LRV matrix; see Remark \ref{rmk:sgd} for discussions on the memory and computational complexity.
Therefore, the main comparison concerns the number of stochastic updates and the maintained running matrices as shown in Table~\ref{tab:complexity}.

To conclude this subsection, we give some examples that discuss $\OGMM$'s connections to and differences from other methods.

\begin{table}[!t]
\centering
\small
\caption{
Computational comparison per newly arrived batch.
Here $p$ is the parameter dimension, $q$ is the moment dimension, and $n_b$ is the size of the new batch.
We assume evaluating $g$ and $\nabla g$ costs $O(q)$ and $O(pq)$ per observation, respectively.
For $\GMM$, $O(N_bq)$ is the cost of one evaluation of $\bar{g}_b(\theta)$, not including forming the Jacobian, matrix inversion, LRV estimation, etc.
For $\SGMM$, $c_{\mathrm{SA}}(p,q)$ is the cost of one stochastic update with a generic order $c_{\mathrm{SA}}(p,q)=O(pq+p^2+q^2)$.
Typically, it needs $O(pq+q^2)$ space complexity for storing the running matrices.
For $\OGMM$ with online weighting using the LRV estimator in \citet{rlrv},
which is denoted as $\OGMM^{\star}$,
$c_{\mathrm{LRV}}(n_b,q)=O(n_bq^2+q^3)$
is the additional weighting/LRV update cost,
and $s_{\mathrm{LRV}}(n_b,q)$ is the space complexity of the LRV routine, which can be as low as $O(n_bq+q^2)$.
For $\OGMM$ and $\OGMM^{\star}$, the time cost $O(p^3)$ is needed for
forming and inverting the $p\times p$ matrix $(\hat{V}^\top \hat{W} \hat{V}_b)^{-1}$.
}
\label{tab:complexity}
\begin{tabular}{lccc}
\toprule[2pt]
Method & Update unit & {Time complexity} & {Space complexity}\\
\midrule[1pt]
{$\GMM$} & {whole sample} & {at least $O(N_b q)$} & {$O(N_b q+q^2)$}\\
{$\SGMM$} & {small mini-batch} & {$O\{n_b c_{\mathrm{SA}}(p,q)\}$} & {$O(pq+q^2)$}\\
{$\OGMM$} & {one batch} & {$O(n_b p q+p q^2+p^3)$} & {$O(n_b q+q^2)$}\\
{$\OGMM^{\star}$} & {one batch} & {$O\{n_b p q + p q^2 + p^3 +c_{\mathrm{LRV}}(n_b,q)\}$} & {$O\{n_b q + q^2+s_{\mathrm{LRV}}(n_b,q)\}$}\\
\bottomrule[2pt]
\end{tabular}
\end{table}

\begin{example}[One-step estimation] \label{eg:lecam}
	Suppose $p=q$, $\hat{W} = I_q$ and $g(\theta, x)$ is the score function.
	Let $\hat\theta_1$ be a $\sqrt{N_1}$-consistent estimator and $D_2$ be empty.
	Then, $\hat\theta_2 = \hat\theta_1 -(V_1^\T I_q V_1)^{-1} V_1^\T I_q \bar{g}_1(\hat\theta_1) = \hat\theta_1 -V_1^{-1} \bar{g}_1(\hat\theta_1)$ performs \citeauthor{lecam1956asymptotic}'s \citeyearpar{lecam1956asymptotic} one-step estimation provided that $V_1$ is a consistent estimator for the Fisher information matrix.
	In this regard, $\OGMM$ can be viewed as a one-step estimator for $\GMM$ with online data.
\end{example}

\begin{example}[Online least squares] \label{eg:ls}
	Let $Y_i = (y_{i,1}, \ldots, y_{i,n_i})^\T$ and $X_i = (x_{i,1}, \ldots, x_{i,n_i})^\T$ be the $i$-th batch of responses and predictors, respectively.
	The ordinary least squares is equivalent to $\GMM$ with $\hat{W} = I_q$ and $g(\theta, \{x, y\}) = x (y -x^\T \theta)$.
	We can verify that $\nabla g(\theta, \{x, y\}) = -x x^\T$, $\hat{V}_b = -N_b^{-1} \sum_{i=1}^b X_i^\T X_i$ and $\hat{U}_b = N_b^{-1} \sum_{i=1}^b (X_i^\T Y_i -X_i^\T X_i \hat\theta_{b-1})$.
	Therefore, $\hat\theta_b = (\sum_{i=1}^b X_i^\T X_i)^{-1} (\sum_{i=1}^b X_i^\T Y_i)$ exactly recovers the ordinary least squares.
\end{example}

\begin{example}[Online quantile regression] \label{eg:quant}
	The smoothed quantile regression \citep{decastro2019iv} is equivalent to $\GMM$ with $\hat{W} = I_q$ and $g(\theta_\tau, \{x, y, h\}) = x [H\{ (y -x^\T \theta_\tau)/h \} +\tau -1]$, where $\tau$ is a quantile level of interest, and $H(\cdot)$ is a smooth approximation of the indicator function $\I_{\cdot > 0}$.
	\citet{chen2019rindicator} extended it to an online setting with memory constraint also by linearization and explicit updates.
	Indeed, we can verify that the form of our $\OGMM$ estimator is close to their proposed linear estimator for quantile regression ($\LEQR$).
	{
    However, while \citet{chen2019rindicator} treated the smoothing bandwidth explicitly throughout the theory, we directly apply the smoothed objective function into our proposal without developing a parallel bandwidth-explicit theory.
    It is because our goal is to illustrate how the $\OGMM$ can be applied to smoothed-moment condition problems, and the explicit bandwidth or smoothing procedure and its theoretical arguments under $\OGMM$ are left for future study.
    Besides,
    there are two methodological differences.
    First, our lagged linearization maintains cumulative online summary statistics $\hat{V}_b$ and $\hat{U}_b$, utilizing the memory to refresh $\hat\theta_{b-1}$ frequently after each incoming batch,
    whereas $\LEQR$ updates the initial estimate used in their summary statistics infrequently only at the beginning of a new interval; see Algorithm 2 in \citet{chen2019rindicator}.}
	Second, our procedures obey the ``online utilization'' principle {as the summary statistics combine all the past information}, whereas $\LEQR$ {mainly combines information from the previous and current intervals
    when their initial estimate is updated.
    Another setting difference is that $\LEQR$ is developed under the i.i.d.~assumption, while $\OGMM$ is designed under a more general time series setting, requiring a more sophisticated LRV estimation instead of marginal variance estimation.}
    {Section~\ref{sec:online-quantile-regression} reports the resulting finite-sample behavior.}
\end{example}

\begin{remark}\label{rmk:sgd}
    The current $\OGMM$ implementation forms the Hessian $\hat H_b=\hat V_b^\top\hat W_b\hat V_b$ for statistically efficient overidentified GMM estimation.
    To further reduce the memory and computational complexity, one can also consider stochastic gradient descent (SGD)-type updates with a diagonal or identity preconditioner; see \citet{chen2025slim} for example.
    This is a promising variant and we leave it for future studies.
\end{remark}

\section{Theory} \label{sec:theory}

To compare with $\GMM$ on the same theoretical basis, we develop the asymptotic theory of $\OGMM$ based on the assumptions in \citet{hall2005gmm}.
Two different cases that correspond to the scenario (S2) in \citet{luo2022qif} are considered:
(a) $\min(n_1, \ldots, n_b) \to \infty$;
(b) $n_1 \to \infty$ and no restrictions on $n_2, n_3, \ldots$.
We allow $b$ to be any finite positive integer but require $n_1 \to \infty$ to ensure the initial estimate is consistent, which {is a realistic assumption and} implicitly assumed in \citet{luo2020renew} and many subsequent works{; see also \citet{chen2019rindicator} for initial consistency requirement in quantile regression.}
{The conditions give the consistency theory in Theorem~\ref{thm:consistency} below: it requires the initial estimate $\hat{\theta}_1$ to be in a neighborhood where the Taylor approximation is valid.
While both cases require a divergent $n_1$ for consistency, we also develop a separate case with fixed-batch and $b\to\infty$ to prove the consistency that requires global stability conditions instead of asymptotically consistent initialization in Theorem~\ref{thm:tail_entry} as a complement.
}

\subsection{Consistency} \label{sec:consistency}
Concatenate the data batches
$\{x_{1,1}, \ldots, x_{1,n_1}, x_{2,1}, \ldots, x_{2,n_2}, \ldots\}$
to be indexed by a single subscript $\{x_1, x_2, \ldots\}$,
where $x_t = x_{i, j}$ and $t = j+\sum_{b=1}^{i-1} n_b$.

\begin{assumption}[Strict stationarity] \label{asum:stationarity}
	The time series data
    $\{x_1, x_2, \ldots\}$
    form a strictly stationary process with sample space $\mathcal{X} \subset \R^d$.
\end{assumption}

\begin{assumption}[Regularity conditions for $g(\theta, x)$] \label{asum:moment-regularity}
	The function $g(\theta, x)$ satisfies:
	\begin{enumerate}
		\item it is continuous on $\Theta$ for each $x \in \mathcal{X}$;
		\item $\E\{g(\theta, x_{i,j})\}$ exists and is finite for every $\theta \in \Theta$;
		\item $\E\{g(\theta, x_{i,j})\}$ is continuous on $\Theta$.
	\end{enumerate}
\end{assumption}

\begin{assumption}[Population moment condition] \label{asum:moment-population}
	The data $\{x_{i,j}\}$ and the parameter $\theta^*$ satisfy the population moment condition $\E\{g(\theta^*, x_{i,j})\} = 0_q$.
\end{assumption}

\begin{assumption}[Global identification] \label{asum:moment-global}
	$\E\{g(\theta', x_{i,j})\} \ne 0_q$ for all $\theta' \in \Theta \setminus \{\theta^*\}$.
\end{assumption}

\begin{assumption}[Regularity conditions for $\nabla g(\theta, x)$] \label{asum:gradient-regularity}
	\hfill
	\begin{enumerate}
		\item $\nabla g(\theta, x)$ exists and is continuous on $\Theta$ for each $x \in \mathcal{X}$;
		\item $\theta^*$ is an interior point of $\Theta$;
		\item $\E\{\nabla g(\theta^*, x_{i,j})\}$ exists and is finite.
	\end{enumerate}
\end{assumption}

\begin{assumption}[Weighting matrix] \label{asum:weight}
	$\hat{W}$ is a symmetric positive semi-definite matrix that converges in probability to a positive definite constant matrix $W$.
\end{assumption}

\begin{assumption}[Ergodicity] \label{asum:ergodicity}
	The random process $\{x_{i,j}\}$ is ergodic.
\end{assumption}

\begin{assumption}[Compactness] \label{asum:compactness}
	$\Theta$ is a compact set.
\end{assumption}

\begin{assumption}[Domination of $g(\theta, x)$] \label{asum:moment-dominance}
	$\E\{\sup_{\theta \in \Theta} \norm{g(\theta, x_{i,j})}_2 \} < \infty$.
\end{assumption}

\begin{assumption}[Short-range dependence]\label{asum:lrv}
Let $x_t = h(\mathcal{F}_t)$
for some measurable function $h$,
where
$\mathcal{F}_t = \{\ldots, \varepsilon_{t-1}, \varepsilon_t\}$ and
$\{\varepsilon_t\}_{t\in\mathbb{Z}}$ are iid innovations.
It satisfies that
(a) $\E\{ g(\theta^*, x_{t}) g(\theta^*, x_{t})^\T \}$ exists and is finite,
(b) $\Sigma$ defined in \eqref{eq:lrv} exists and is a finite positive definite matrix, and
(c) $\sum_{t=0}^{\infty}\E\{\|\mathcal{P}_0 g(\theta^*, x_t)\|_2^2 \}<\infty$,
where
$\mathcal{P}_0 \cdot= \E(\cdot\mid \mathcal{F}_0)-\E(\cdot\mid \mathcal{F}_{-1})$.
\end{assumption}

\begin{assumption}[Local Lipschitz continuity of $\nabla g(\theta, x)$] \label{asum:gradient-lipschitz}
	Let $B_\epsilon = \{\theta \in \Theta: \norm{\theta-\theta^*}_2 \le \epsilon\}$.
	\begin{enumerate}
		\item For some $\epsilon > 0$, there exists $L > 0$ such that $\norm{ \E\{\nabla g(\theta, x_{i,j})\} -\E\{\nabla g(\theta^*, x_{i,j})\} }_F \le L \norm{\theta-\theta^*}_2$ for all $\theta \in B_\epsilon$.
		\item For some $\epsilon > 0$, there exists $L > 0$ such that $\norm{ \nabla g(\theta, x) -\nabla g(\theta^*, x) }_F \le L \norm{\theta-\theta^*}_2$ for all $\theta \in B_\epsilon$ and $x \in \mathcal{X}$.
	\end{enumerate}
\end{assumption}

\begin{assumption}[Local uniform convergence to $\E\{\nabla g(\theta, x_{i,j})\}$] \label{asum:gradient-uniform}
	\hfill
	\begin{enumerate}
		\item For some $\epsilon > 0$ and $i=1, \ldots, b$, as
		$n_i \to \infty$,
		\begin{enumerate}[label=(\roman*)]
			\item $\sup_{\theta \in B_\epsilon} \norm{n_i^{-1} \nabla G(\theta; D_i) -\E\{\nabla g(\theta, x_{i,j})\} }_F = o_p(1)$, which is implied by
			\item $\sup_{\theta \in B_\epsilon} \norm{n_i^{-1} \nabla G(\theta; D_i) -\E\{\nabla g(\theta, x_{i,j})\} }_F = O_p(n_i^{-1/2})$.
		\end{enumerate}
		\item For some $\epsilon > 0$, as $N_b \to \infty$,
		\begin{enumerate}[label=(\roman*)]
			\item $\norm{N_b^{-1} \sum_{k=1}^b \nabla G(\theta^*; D_k) -\E\{\nabla g(\theta^*, x_{i,j})\}}_F = o_p(1)$, which is implied by
			\item $\sup_{\theta \in B_\epsilon} \norm{N_b^{-1} \sum_{k=1}^b \nabla G(\theta; D_k) -\E\{\nabla g(\theta, x_{i,j})\}}_F = o_p(1)$.
		\end{enumerate}
	\end{enumerate}
\end{assumption}

Assumptions \ref{asum:stationarity}--\ref{asum:moment-dominance} are identical to Assumptions 3.1--3.5 and 3.7--3.10 in \citet{hall2005gmm}.
Assumption \ref{asum:lrv} asserts short-range dependence of
$\{g(\theta^*, x_{t}) \}_{t\in\mathbb{Z}}$ via the weak stability condition \citep{wu2005}.
It implies the central limit theorem (CLT)
$\sqrt{N_b}\bar{g}_b(\theta^*)\cin{d}\Normal(0,\Sigma)$;
see, e.g., Theorem~3 of \cite{wu2011}.
Assumption \ref{asum:lrv} can be replaced by other conditions that
ensure the CLT holds, e.g., strong mixing \citep{Rosenblatt56} and
near-epoch dependence \citep{ibragimov1962}.
Assumptions \ref{asum:gradient-lipschitz} and \ref{asum:gradient-uniform} are required to hold for some neighborhood $B_\epsilon$ of $\theta^*$ only so both are local assumptions.
Although Assumptions \ref{asum:gradient-lipschitz}(a) and (b) (local Lipschitz continuity) are stronger than Assumption 3.12 (continuity) in \citet{hall2005gmm}, they are standard in the literature to bound the Taylor approximation error.
Compared with Condition 3 in \citet{luo2020renew}, Assumption \ref{asum:gradient-lipschitz}(b) is slightly weaker because only local smoothness is required,
and Assumption \ref{asum:gradient-lipschitz}(a) is implied by Assumption \ref{asum:gradient-lipschitz}(b).
Note that Assumption \ref{asum:gradient-uniform}(a)(ii) is stronger than Assumption \ref{asum:gradient-uniform}(b)(ii),
where the latter is identical to Assumption 3.13 in \citet{hall2005gmm}.

\begin{theorem}[Consistency of $\hat{\theta}_b$] \label{thm:consistency}
	Suppose Assumptions \ref{asum:stationarity}--\ref{asum:moment-dominance} hold.
	\begin{enumerate}
		\item Under Assumptions \ref{asum:gradient-lipschitz}(a) and \ref{asum:gradient-uniform}(a)(i),
		$\hat\theta_b \cin{p} \theta^*$ as
		$\min(n_1, \ldots, n_b) \to \infty$.
		\item Under Assumptions \ref{asum:gradient-lipschitz}(b) and \ref{asum:gradient-uniform}(b)(i),
		$\hat\theta_b \cin{p} \theta^*$ as
		$n_1 \to \infty$.
	\end{enumerate}
\end{theorem}

When $b = 1$, Theorem \ref{thm:consistency} is exactly Theorem 3.1 in \citet{hall2005gmm}, which states the consistency of $\hat\theta_1$ if it is computed using $\GMM$.
{Then, for $b \ge 2$, Assumptions \ref{asum:gradient-lipschitz} and \ref{asum:gradient-uniform}(i) ensure the Taylor approximation error, containing both $O_p(n_i^{1/2}\|\hat\theta_i-\theta^*\|_2)$ term reflecting sample Jacobian fluctuation and quadratic $O_p(n_i\|\hat\theta_i-\theta^*\|_2^2)$ term reflecting curvature under case (a) and only quadratic term under case (b), is negligible so that $\hat\theta_b$ remains consistent.
Thus the complexity of nonlinearity is controlled locally by the Jacobian Lipschitz behavior: more pronounced nonlinearity leads to larger remainder constants.}

{The two regimes (a)--(b) correspond to different smoothness requirements.
Case (a) uses a weaker local smoothness condition in Assumption~\ref{asum:gradient-lipschitz}(a) but a stronger batch-wise uniform convergence condition in Assumption~\ref{asum:gradient-uniform}(a), and therefore requires $\min(n_1,\ldots,n_b)\to\infty$.
Case (b) uses stronger local Lipschitz control of the sample Jacobian in Assumption~\ref{asum:gradient-lipschitz}(b) and a weaker uniform convergence of sample average in Assumption~\ref{asum:gradient-uniform}(b), which removes the need for the later update batch sizes to diverge for consistency.
}

{
The online-specific technical difficulty of not using a fixed $n_1<\infty$ in
Theorem \ref{thm:consistency}
is that the recursion depends on the running averages evaluated at all past iterations,
whose error cannot be analyzed as a standard one-shot perturbation and their stability
should be guaranteed by at least $n_1\to\infty$; see the proof of Lemma~\ref{lem:convergence-uv}.
Such difficulty complicates the derivation of a non-asymptotic error bound as well, which is left for future work.
}
{
However, under stronger global stability conditions,
we can also prove the consistency even when the batch size is fixed at $n_0 = n_1 = \cdots =n_b <\infty$
while the number of batches $b\to\infty$.
This result is stated in the following theorem, whose detailed assumptions and discussions are deferred to \cref{sec:consistency_fix_batch} of the \Supp.
}

{
\begin{theorem}[Consistency of $\hat{\theta}_b$ with finite $n_b$]\label{thm:tail_entry}
Let $n_i\equiv n_0 <\infty$ for all $i=1,\ldots,b$ with $b\to\infty$.
Suppose Assumptions~\ref{asum:stationarity}--\ref{asum:moment-dominance} hold.
Suppose further Assumptions~\ref{as:weight-new}, \ref{as:cvx}, \ref{as:Lip_nabla_global},
\ref{as:full_rank_V}, \ref{as:global-contraction}, and \ref{as:regularProblem} in the \Supp~hold.
Then $\hat{\theta}_b \cin{p} \theta^*$ and $\hat V_b \cin{p} V=\E\{\nabla g(\theta^*,x_{i,j})\}$, and $\hat U_b \cin{p} U = \E\{g(\theta^*,x_{i,j})\}=0_q$ as $b\to \infty$.
\end{theorem}
Theorem~\ref{thm:tail_entry} is a complementary result for the consistency of $\hat{\theta}_b$.
It does not require $n_0 \rightarrow \infty$ or a consistent initialization $\hat{\theta}_1$, but at the price of requiring that the problem considered be more regular in the sense of those additional assumptions, which are still sensible and plausible.
In particular, we discuss these assumptions in Example~\ref{eg:linear-models} of the Supplement in the case of linear moment models, which cover the least squares method and instrumental variables.
In that example, the global Lipschitz continuity of $\nabla g(\theta;x)$ in Assumption~\ref{as:Lip_nabla_global}
and the contraction assumptions in \ref{as:global-contraction} hold directly,
while reducing other conditions to more intuitive ones.
}

\subsection{Asymptotic normality} \label{sec:asymptotic-normality}

\begin{theorem}[Asymptotic normality of $\hat\theta_b$] \label{thm:normality-par}
	Recall that $V = \E\{\nabla g(\theta^*, x_{i,j})\}$.
	Let $H = (V^\T W V)^{-1} V^\T W$.
Suppose Assumptions \ref{asum:stationarity}--\ref{asum:lrv} hold and $n_i = o(N_{i-1}^2)$ for $i = 2, \ldots, b$.
	\begin{enumerate}
		\item Under Assumptions \ref{asum:gradient-lipschitz}(a) and \ref{asum:gradient-uniform}(a)(ii),
		\[
		\sqrt{N_b} (\hat\theta_b -\theta^*) \cin{d} \Normal(0, H \Sigma H^\T)
		\]
		as
		$\min(n_1, \ldots, n_b) \to \infty$.
		\item Under Assumptions \ref{asum:gradient-lipschitz}(b) and \ref{asum:gradient-uniform}(b)(ii),
		$\sqrt{N_b} (\hat\theta_b -\theta^*) \cin{d} \Normal(0, H \Sigma H^\T)$ as
		$n_1 \to \infty$.
	\end{enumerate}
\end{theorem}

When $b = 1$, Theorem \ref{thm:normality-par} is exactly Theorem 3.2 in \citet{hall2005gmm}, which states the asymptotic normality of $\hat\theta_1$ if it is computed using $\GMM$.
The asymptotic variance is minimized at $(V^\T \Sigma^{-1} V)^{-1}$ when $W = \Sigma^{-1}$.
{When $b=2$, we require $n_2=o(n_1^2)$ by $N_1=n_1$, which is the rate needed in the proof to make the linearization remainder asymptotically vanish.}
{For $b \geq 2$, the condition $n_i=o(N_{i-1}^2)$ guarantees the $o_p(1)$ accumulated Taylor remainder terms, allowing }$\hat\theta_b$ to achieve the same asymptotic efficiency regardless of whether the initial estimator is $\GMM$.
As briefly discussed in Example \ref{eg:lecam}, any $\sqrt{N_1}$-consistent $\hat\theta_1$ will be sufficient for $\hat\theta_2$ to achieve the same asymptotic efficiency as $\GMM$.
Therefore, users have some flexibility in choosing $\hat\theta_1$ in Algorithm \ref{algo:ogmm}.
{We remark that the finite-$n_b$ result in Theorem~\ref{thm:tail_entry} guarantees only consistency but not asymptotic normality,
however the Monte Carlo experiments in \cref{sec:simulations} still show promising inference performance under the finite-$n_b$ regime.}

Compared with implicitly updated estimators in the literature, $\OGMM$ also achieves the same asymptotic efficiency as its offline counterpart in Theorem \ref{thm:normality-par}.
The condition $n_i = o(N_{i-1}^2)$ for $i = 2, \ldots, b$ is not new for explicit estimators since the procedure proposed in \citet{chen2019rindicator} also relies on this condition.
Nevertheless, we emphasize that it is mild as it is satisfied when the batch size is constant.

\begin{remark}[Batch selection]\label{rmk:batch}
In practice, the initial batch can be formed by combining  early arriving blocks until the initial estimator and the initial weighting matrix are numerically stable (e.g., $n_1\geq 100$).
In this sense, it is often possible to justify the diverging $n_1\to\infty$ assumption when an initial burn-in period is available.
Following initialization, one may use a fixed, moderate update batch size, or the natural incoming batch size in streaming applications.
Smaller update batches improve the local approximation by keeping adjacent estimates close, while larger batches reduce update overhead but make the approximation coarser.
We recommend using natural batch sizes and sub-batching
unusually large incoming blocks.
\end{remark}

\begin{remark}[Verification of Assumptions \ref{asum:gradient-lipschitz}(b) and \ref{asum:gradient-uniform}(b)(ii)] \label{rmk:gradient}
	In practice, case (b) is usually of broader interest because there are no restrictions on $n_2,n_3,\ldots$.
	Therefore, we give some examples where we can verify Assumptions \ref{asum:gradient-lipschitz}(b) and \ref{asum:gradient-uniform}(b)(ii).
	\begin{enumerate}
		\item For the ordinary least squares in Example \ref{eg:ls}, Assumption \ref{asum:gradient-lipschitz}(b) holds because $\nabla g(\theta, \{x,y\}) = -x x^\T$ does not depend on $\theta$.
		By the law of large numbers, Assumption \ref{asum:gradient-uniform}(b)(ii) also holds if $\E\{ \norm{x x^\T}_F \} < \infty$.
		\item The instrumental variables estimation is equivalent to $\GMM$ with $g(\theta, \{x, y, z\}) = z (y -x^\T \theta)$, where $z$ is a vector of instrumental variables.
		Similar to the ordinary least squares, Assumption \ref{asum:gradient-lipschitz}(b) holds because $\nabla g(\theta, \{x, y, z\}) = -z x^\T$ does not depend on $\theta$, and
		Assumption \ref{asum:gradient-uniform}(b)(ii) holds if $\E\{ \norm{z x^\T}_F \} < \infty$.
		\item The maximum likelihood estimation is equivalent to $\GMM$ with $g(\theta, x) = \nabla \ell(\theta, x)$, where $\ell(\theta, x)$ is the log-likelihood function.
		If $\ell(\theta, x)$ is twice continuously differentiable for all $\theta \in \Theta$ (\citealt[Condition 3]{luo2020renew}), Assumption \ref{asum:gradient-lipschitz}(b) holds by the mean value theorem.
		If $\Theta$ is compact and $\E\{\sup_{\theta \in \Theta} \norm{\nabla^2 \ell(\theta, x)}_F \} < \infty$, then Assumption \ref{asum:gradient-uniform}(b)(ii) also holds by Lemma 2.4 in \citet{newey1994asymptotic}.
	\end{enumerate}
\end{remark}

\subsection{Over-identifying restrictions test} \label{sec:over-identifying-restrictions-test}

\begin{theorem} \label{thm:normality-moment}
	If the conditions in Theorem \ref{thm:normality-par}(a) or (b) hold, then
	\[
	\hat{W}^{1/2} \sqrt{N_b} U_b \cin{d} \Normal(0, R W^{1/2} \Sigma (W^{1/2})^\T R^\T)
	\]
	as (a) $\min(n_1, \ldots, n_b) \to \infty$; or (b) $n_1 \to \infty$,
	where $R = I_q -W^{1/2} V (V^\T W V)^{-1} V^\T (W^{1/2})^\T$.
	If $\hat\Sigma$ is positive definite and converges in probability to $\Sigma$,
    {where $\Sigma$ satisfies Assumption \ref{asum:lrv},}
	then $N_b U_b^\T \hat\Sigma^{-1} U_b \cin{d} \chi_{q-p}^2$ as (a) $\min(n_1, \ldots, n_b) \to \infty$; or (b) $n_1 \to \infty$.
\end{theorem}

When $b = 1$, Theorem \ref{thm:normality-moment} is equivalent to Theorems 3.3 and 5.1 in \citet{hall2005gmm} if $\hat\theta_1$ is computed using $\GMM$ because $U_1 \equiv \bar{g}_1(\hat\theta_1)$.
For $b \ge 2$, the online sample moment and online Sargan--Hansen test statistic differ from their offline version as more than one $\hat\theta_i$'s are plugged in.
Therefore, we have an additional first-order correction term $N_b^{-1} \sum_{i=1}^b \nabla G(\hat\theta_i; D_i) (\hat\theta_b -\hat\theta_i)$ compared with \citet{chen2023sgmm}.
We remark that the condition on $\hat\Sigma$ is satisfied by the online LRV estimator under standard conditions,
{and again, the finite-$n_b$ inference like Theorem~\ref{thm:tail_entry} is not covered in this theory, but we still apply it empirically in Section~\ref{sec:simulations}.}

\section{Simulation results} \label{sec:simulations}

Throughout Sections~\ref{sec:simulations} and \ref{sec:applications},
we simplify the notation by using a single subscript to index variables
(i.e., $x_t = x_{i,j}$);
see Section~\ref{sec:consistency} for the precise mapping.
Unless otherwise stated, we use the R-package \texttt{momentfit} 0.5 for offline two-step $\GMM$ and $\TSLS$, \texttt{sandwich} 3.1.1 for the Bartlett kernel estimator with a bandwidth selected by fitting a first-order autoregressive model \citep{andrews1991kernel} and \texttt{prewhite=FALSE}, \texttt{rlaser} 0.1.1 for \citeauthor{rlrv}'s \citeyearpar{rlrv} recursive LRV estimator with \texttt{pilot} equals to $n_1$ and \texttt{warm=FALSE}, and our R-package \texttt{ogmm} 0.0.1 for $\OGMM$, $\SGMM$ and $\LEQR$.
Non-default arguments will be discussed whenever they are modified.
All simulations are performed on
mac OS Tahoe 26.4.1 with Apple M5 Pro Chip, 24GB Memory, running R version 4.5.3 across 8 parallel threads.
The simulations are repeated $520$ times (i.e., 65 replicates per thread). A $5\%$ nominal level is used where applicable.

\subsection{Online instrumental variables regression} \label{sec:online-instrumental-variables-regression}

Before we investigate the finite-sample performance of $\OGMM$, we describe $\SGMM$ \citep{chen2023sgmm}, another framework for computing $\GMM$ with online updates.
Let $\eta_k = \eta_0 k^{-a}$ be the learning rate at the $k$-th iteration for some $\eta_0 > 0$ and $a \in (1/2, 1)$.
The general $\SGMM$ algorithm considers
$
\check\theta_{k+1}
= \check\theta_k -\eta_{k+1} (\check{V}_k^\T \check{W}_k \check{V}_k)^{-1} \check{V}_k^\T \check{W}_k g(\check\theta_k, x_{k+1}),
$
where $\check{V}_k$ is an online average that estimate $V = \E\{ \nabla g(\theta^*, x_{i,j}) \}$, and
$\check{W}_k$ is the inverse of the second sample moment updated using the Sherman--Morrison--Woodbury formula.
By updating the average $\bar\theta_N = N^{-1} \sum_{k=1}^N \check\theta_k$,
inference on $\theta^*$ can be performed using a plug-in approach.
To conduct the Sargan--Hansen test, \citet{chen2023sgmm} proposed the test statistic
\begin{equation} \label{eq:sgmm-overident}
	N \check{g}_N^\T \check{W}_N \check{g}_N,
	\quad \text{where} \quad
	\check{g}_N = \frac{1}{N} \sum_{k=1}^N g(\bar\theta_k, x_k).
\end{equation}
For the initial learning rate, consider using an initialization data batch $D_0$ to compute $\check\theta_0$, $\check{V}_0$ and $\check{W}_0$.
Then, \citet{chen2023sgmm} proposed to use $\eta_0 = 1/\Psi_0(\kappa)$, where
\begin{equation} \label{eq:sgmm-lr}
	\Psi_0(\kappa)
	= \mathrm{quantile}_\kappa \left\{ p^{-1} \norm{(\check{V}_0^\T \check{W}_0 \check{V}_0)^{-1} \check{V}_0^\T \check{W}_0 \nabla g(\check\theta_0, x)}_S: x \in D_0 \right\},
\end{equation}
$\kappa$ is a predetermined quantile level and $\norm{\cdot}_S$ is the spectral norm in \eqref{eq:sgmm-lr}.

Now, consider the following linear instrumental variables regression models:
\begin{enumerate}[series=enum-model,label=(\Roman*),ref=\Roman*]
	\item \label{enum:iv-het} Independent: let
	$
	y_k = x_k^\T \theta^* +5 \exp(z_{k,q}) (\nu_k +\varepsilon_k),
	$
	where $\nu_k, \varepsilon_k \simIID \Normal(0,1)$,
	$z_k \simIID \Normal_q(0, \mathcal{V})$ with $\mathcal{V}_{i,j} = \rho^{|i-j|}$,
	$x_{k,1} = 0.1 \sum_{j=2}^p x_{k,j} +0.5 \sum_{j=p}^q z_{k,j} +\nu_k$, and
	$x_{k,j} = z_{k,j-1}$ for $j = 2, \ldots, p$.
	We exactly replicate a simulation model in \citet{chen2023sgmm} by setting $\theta^* = (1, \ldots, 1)^\T$, $p = 5$, $q = 20$, and $\rho = 0.5$.
	\item \label{enum:iv-arma} Dependent: let
	$
	y_k = \theta_1^* y_{k-1} +\theta_2^* y_{k-2} +\phi_1^* \varepsilon_{k-1} +\phi_2^* \varepsilon_{k-2} +\varepsilon_k,
	$
	where $\varepsilon_k \simIID \Normal(0,1)$.
	We exactly replicate a simulation model in \citet{pierre2010gmm} by setting $\theta^* = (1.4, -0.6)^\T$ and $\phi^* = (0.6, -0.3)^\T$.
	The selected instruments are $z_k = (y_{k-3}, \ldots, y_{k-6})^\T$.
\end{enumerate}
We are interested in estimating $\theta^*$ at both $N_b = 2000b$ for $b = 1, \ldots, 100$ {and $N_b=50b$ for $b=1, \ldots, 2000$} using
\begin{enumerate}
	\item $\TSLS$: the two-stage least squares as the offline benchmark and initial estimator for model \ref{enum:iv-het}.
	The Bartlett kernel estimator is used for inference.
	\item $\GMM$: the two-step $\GMM$ as the offline benchmark and initial estimator for model \ref{enum:iv-arma}.
	The Bartlett kernel estimator is used for inference and the weighting matrix.
	We also use $\tilde\theta_{b-1}$ as the initial value when we optimize for $\tilde\theta_b$, and $\tilde\theta_0 = (0, 0)^\T$.
	\item $\SGMM$: for model \ref{enum:iv-het}, we consider different initial learning rates selected at $\kappa \in \{0.5, 0.7, 0.9\}$.
	For model \ref{enum:iv-arma}, for the case $N_b=2000b$, we only consider $\kappa=0.5$ but further try \citeauthor{rlrv}'s \citeyearpar{rlrv} recursive LRV estimator for both inference and the weighting matrix; and for the case $N_b=50b$, we also consider $\kappa=0.3,0.4$ for comparison.
	For both models, we set $a = 0.501$ in the learning rate $\eta_k$.
	\item $\OGMM$: both \citeauthor{welford1962recursive}'s \citeyearpar{welford1962recursive} and \citeauthor{rlrv}'s \citeyearpar{rlrv} estimators are considered for inference and the weighting matrix.
	We further compare explicit and implicit updates with the stopping criteria in \citet{luo2022qif}.
		Specifically, we stop when the number of Newton--Raphson iterations reaches $50$ or
		\[
		\{ (\hat{V}_b^{(r)})^\T \hat{W} \hat{U}_b^{(r)} \}^\T \{ (\hat{V}_b^{(r)})^\T \hat{W} \hat{V}_b^{(r)} \}^{-1} (\hat{V}_b^{(r)})^\T \hat{W} \hat{U}_b^{(r)} < 10^{-6},
		\]
		where
		\begin{align*}
		\hat{V}_b^{(r)} &= N_b^{-1} \{ N_{b-1} V_{b-1} +\nabla G(\hat\theta_b^{(r-1)}; D_b) \}, \\
		\hat{U}_b^{(r)} &= N_b^{-1} \{ N_{b-1} U_{b-1} +N_{b-1} V_{b-1} (\hat\theta_b^{(r-1)} -\hat\theta_{b-1}) +G(\hat\theta_b^{(r-1)}; D_b) \}, \quad \text{and} \\
		\hat\theta_b^{(r)} &= \hat\theta_b^{(r-1)} -\{ (\hat{V}_b^{(r)})^\T \hat{W} \hat{V}_b^{(r)} \}^{-1} (\hat{V}_b^{(r)})^\T \hat{W} \hat{U}_b^{(r)}.
		\end{align*}
\end{enumerate}

\begin{figure}[!t]
	\centering
	\includegraphics[width=\textwidth]{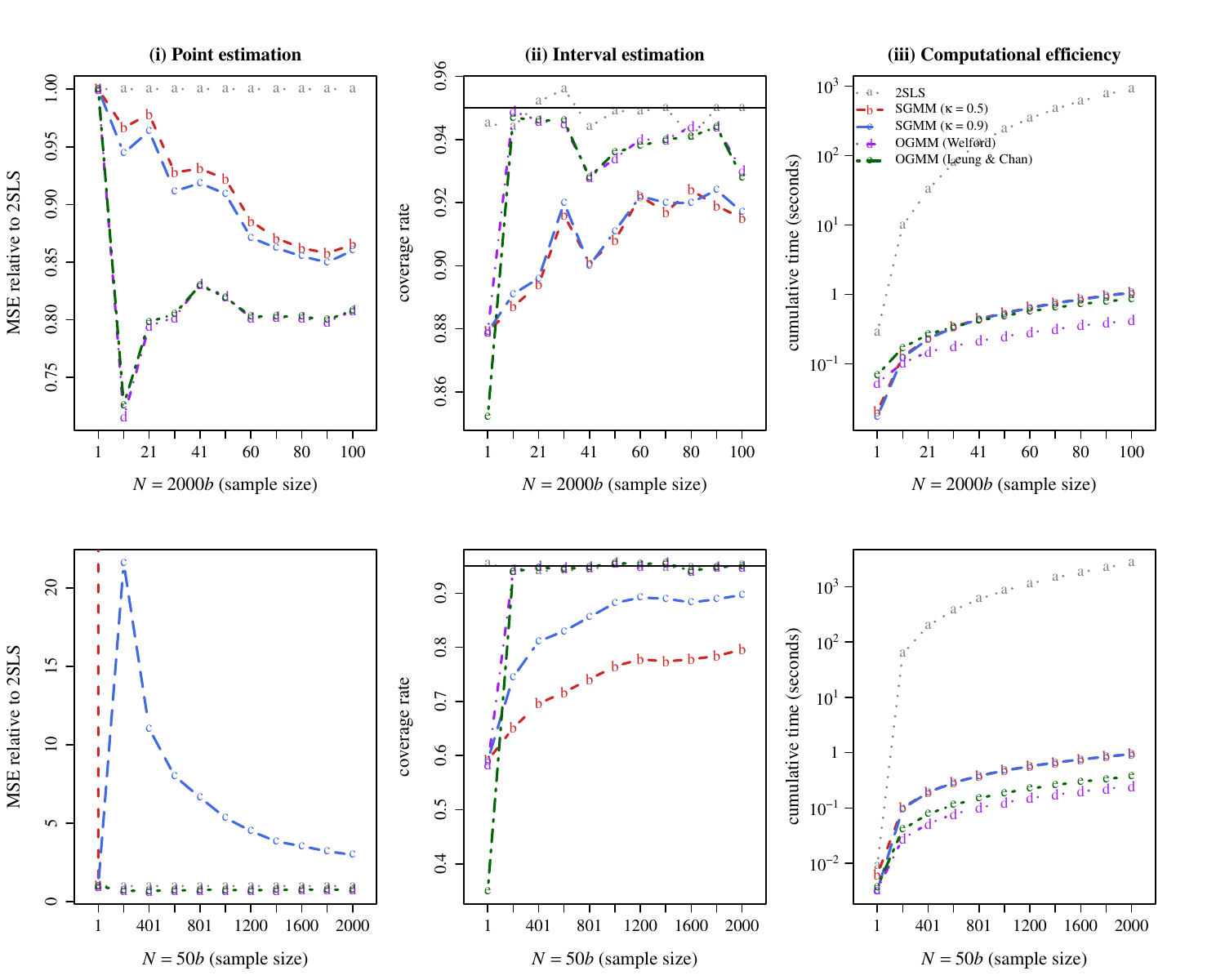}
	\caption{Online instrumental variables estimation of $\theta^*_1$ in the independent model \ref{enum:iv-het}.
    The upper panel uses a large batch size $N_b=2000b$ for $b=1,\ldots,100$,
    and the lower panel sets a small batch size with large batch number $N_b=50b$ for $b=1,\ldots,2000$, where the values of $\SGMM \ (\kappa=0.5)$ are too large to display, which are ($22.25\times 10^5$, $9.77\times 10^5$, $1.61\times 10^6$) for $b=(201,1000,2000)$,
    respectively.}
	\label{fig:iv-het}
\end{figure}

\begin{figure}[!t]
	\centering
	\includegraphics[width=\textwidth]{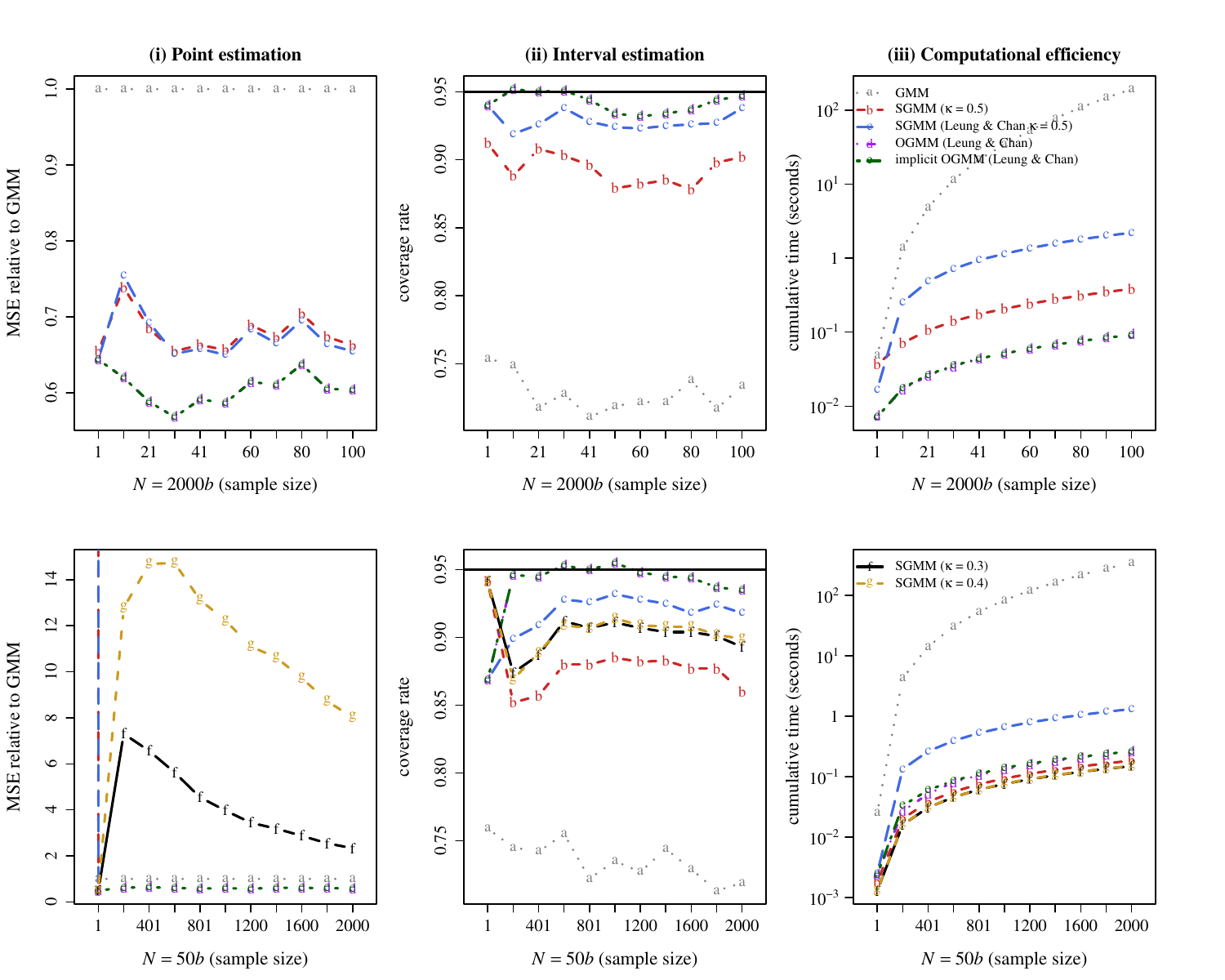}
	\caption{Online instrumental variables estimation of $\theta^*_1$ in the dependent model \ref{enum:iv-arma}.
    The upper panel uses a large batch size $N_b=2000b$ for $b=1,\ldots,100$,
    and the lower panel sets a small batch size with large batch number $N_b=50b$ for $b=1,\ldots,2000$, where the values of $\SGMM$ ($\kappa=0.5$) and $\SGMM$ (Leung \& Chan, $\kappa=0.5$) are too large to display, which are ($2.25\times 10^6$, $3.92\times 10^6$) and ($7.33\times 10^5$,
    $3.32\times 10^5$) for $b=(1000,2000)$, respectively.
    In the lower panel, we add $\SGMM$ with $\kappa=0.3$ and $\kappa=0.4$ for comparison.}
	\label{fig:iv-arma}
\end{figure}

In Figures \ref{fig:iv-het} \& \ref{fig:iv-arma}, we compare mean squared errors (MSEs) of the point estimates, coverage rates of the  $95\%$ confidence intervals for $\theta_1^*$, and computational time in seconds.
Figure \ref{fig:iv-het} confirms that suboptimal $\sqrt{N_1}$-consistent $\hat\theta_1$ (such as $\TSLS$) is still sufficient for $\OGMM$ to be asymptotically optimal.
Despite the fact that data are independent under model \ref{enum:iv-het}, using a recursive LRV estimator for autocorrelation robustness does not affect the performance of $\OGMM$ much.
On the other hand, when the batch size is large (figure upper panel), $\SGMM$ is not sensitive to the choice of $\kappa$ here.
Its MSEs relative to $\TSLS$ and coverage rates are similar to those reported in \citet{chen2023sgmm}.
The gap between $\SGMM$ and $\GMM$ (thus $\OGMM$) estimates was similarly observed in the first empirical example in \citet{chen2023sgmm}.
Such gap disappeared when $\SGMM$ was run over multiple epochs, and so they recommended opting for a multi-epoch algorithm ``when the dataset size allows for the storage of all observations'' \citep{chen2023sgmm}.
We do not consider multi-epoch estimation as our time series is streaming.
{
However, with small batch size (figure lower panel), although $\kappa=0.9$ outperforms $\kappa=0.5$ in terms of MSE, $\SGMM$ obtains extremely inflated error, while our $\OGMM$ still gives a consistent $\hat{\theta}_b$, showing the small initialization sample size will not affect the long-term convergence behavior.
{
We note that
although under a fixed batch size, the $\OGMM$ asymptotic normality is not established theoretically,
it continues to perform well empirically,
as illustrated in the plot.}
}

Figure \ref{fig:iv-arma} verifies the asymptotic efficiency of $\OGMM$ relative to $\GMM$.
Since data are dependent under model \ref{enum:iv-arma}, recursive LRV estimation improves the coverage rates of $\SGMM$, which was originally designed for independent data.
However, $\OGMM$ still outperforms $\SGMM$ when the same LRV estimator is used.
It does not gain efficiency through implicit updates though because the moment is linear in $\theta$; see Remark \ref{rmk:gradient}.
{
Computationally, the figure shows that $\OGMM$ and $\SGMM$ greatly reduce the computational cost, where $\SGMM$ requires approximately the same times across with different $\kappa$ and the usage of \citeauthor{rlrv}'s \citeyearpar{rlrv} will increase the time  due to the inversion of near-singular matrices.
$\OGMM$ is faster than $\SGMM$ due to its one batch-level update after aggregation compared to $\SGMM$'s many stochastic recursions within the same incoming batch, aligned with Table~\ref{tab:complexity}.
}
{We provide additional simulations with larger dimension $p$ in \cref{sec:additional-iv-p} of the \Supp, where the $\OGMM$ continues to work well.}

\subsection{Nonlinear model}\label{sec:nonlinear}
To further examine a nonlinear and over-identified $\GMM$ problem, consider the moving average (MA) model of order $1$:
$x_i = \mu + \epsilon_i + \psi \epsilon_{i-1}$
for $i=1,\ldots,n$ where $\epsilon_i \simIID \Normal(0,\sigma^2)$ and $|\psi^*| < 1$ \citep{harris1999gmm}, where $\psi^*$ is the true value.
The unknown parameters are $\theta=(\mu,\psi,\sigma^2)$.
We use the moment condition $g(\theta, x)=(
x_t - \mu,
x_t^2 - \mu^2 - \sigma^2 (1+\psi^2),
x_tx_{t-1} - \mu^2 - \sigma^2\psi,
x_tx_{t-2} - \mu^2)^\top$.
Set the true $\theta^* = (1,0.6,1)$, $N_b=500b$ for $b=1,\ldots,400$.
The candidates are $\GMM$ and $\OGMM$ in \cref{sec:online-instrumental-variables-regression} under the same settings.
As an additional competitor, we include $\SLIM$ (Stochastic Learning and Inference in overidentified Models), which is a stochastic approximation framework for nonlinear $\GMM$ proposed by \citet{chen2025slim}.
We emphasize that $\SLIM$ is restarted and executed using
the full sample
whenever a new data batch arrives; hence, it is an offline approach.
We consider the following variants:
\begin{enumerate}
	\item $\SLIM$ (\textsf{1RS}): warm start with mini-batch size $B_{\mathrm{ws}}=256$, number of epochs $E_{\mathrm{ws}}=25$, $\{\gamma_e=\gamma_{0,\mathrm{ws}}e^{-a_{\mathrm{ws}}}\}_{e=1}^{E_{\mathrm{ws}}}$, $\gamma_{0,\mathrm{ws}}=0.005$, $a_{\mathrm{ws}}=0.501$, sample size $\min(N_b, 20000)$, and initial value $\theta_{0,\mathrm{ws}}^*=(0,0,1)$;
    first-order stage with number of iteration $N=6000$, mini-batch size $B_G=256$ and $B_g=32$, and learning rate $\gamma_t=\gamma_0 t^{-a}$ where $a=0.501$ and $\gamma_0$ is computed by \citet[Eq.~(24)]{chen2025slim} with $s_0=20$, $B_{\mathrm{main}}=\min(B_G,B_g)=32$ and $|\mathbb{T}|=40$ calibration batches;
    for interval estimation, it uses the single-coordinate random-scaling (\textsf{RS}) procedure in \citet[Section 4]{chen2025slim} with critical value $\mathrm{cv}_{0.975}=6.747$ given by \citet[Table I]{abadir1997two};
	\item $\SLIM$ (\textsf{2RS}): same warm start, first-order stage, and \textsf{RS} approach as $\SLIM$ (\textsf{1RS}), followed by a second-order stage with $T=12000$, $B_{G,0}=B_G=256$, $B_{G,t}=\lceil B_{G,0}+\log(t-N)\cdot\mathbb{1}(t>N)\rceil$, $M_{\mathrm{MB}}=4000$ for mini-batch weight estimation $W_{\mathrm{MB}}$, $\Phi$ averaged over $20$ batches of size $256$, and second-stage step size $\gamma_t = 10^{-4} t^{-0.7}$;
	\item $\SLIM$ (\textsf{2PI}): same warm start, but first-order stage with $N=10000$ and $B_g=4$; second-order stage with $T=20000$, $B_g=4$, $B_{G,0}=256$, full-sample weight estimation and $\Phi_n$ computation, and plug-in (\textsf{PI}) covariance estimation and inference in \citet[Section 5]{chen2025slim} with standard normal critical value.
\end{enumerate}
The results in Figure~\ref{fig:51(II-MA)1_MA1} suggest that $\OGMM$ with \citeauthor{rlrv}'s \citeyearpar{rlrv} estimator and $\SLIM$ are competitive for point estimation of all the three parameters.
For the inference, $\OGMM$ with \citeauthor{rlrv}'s \citeyearpar{rlrv} estimator maintains a near-nominal coverage rate, but not for other $\OGMM$ variants yet, since \citeauthor{welford1962recursive}'s \citeyearpar{welford1962recursive} estimator does not account for serial dependence.
Among the three $\SLIM$ variants, (\textsf{2PI}) clearly dominates (\textsf{1RS}) and (\textsf{2RS}) in coverage, but its empirical coverage rates are still below the nominal $95\%$ level mainly caused by under-estimated standard errors rather than biased point estimates.
It is because $\SLIM$ is tailored to iid data setting and may require a nontrivial modification to accommodate time series.

{
\begin{figure}[!t]
	\centering
	\includegraphics[width=\textwidth]{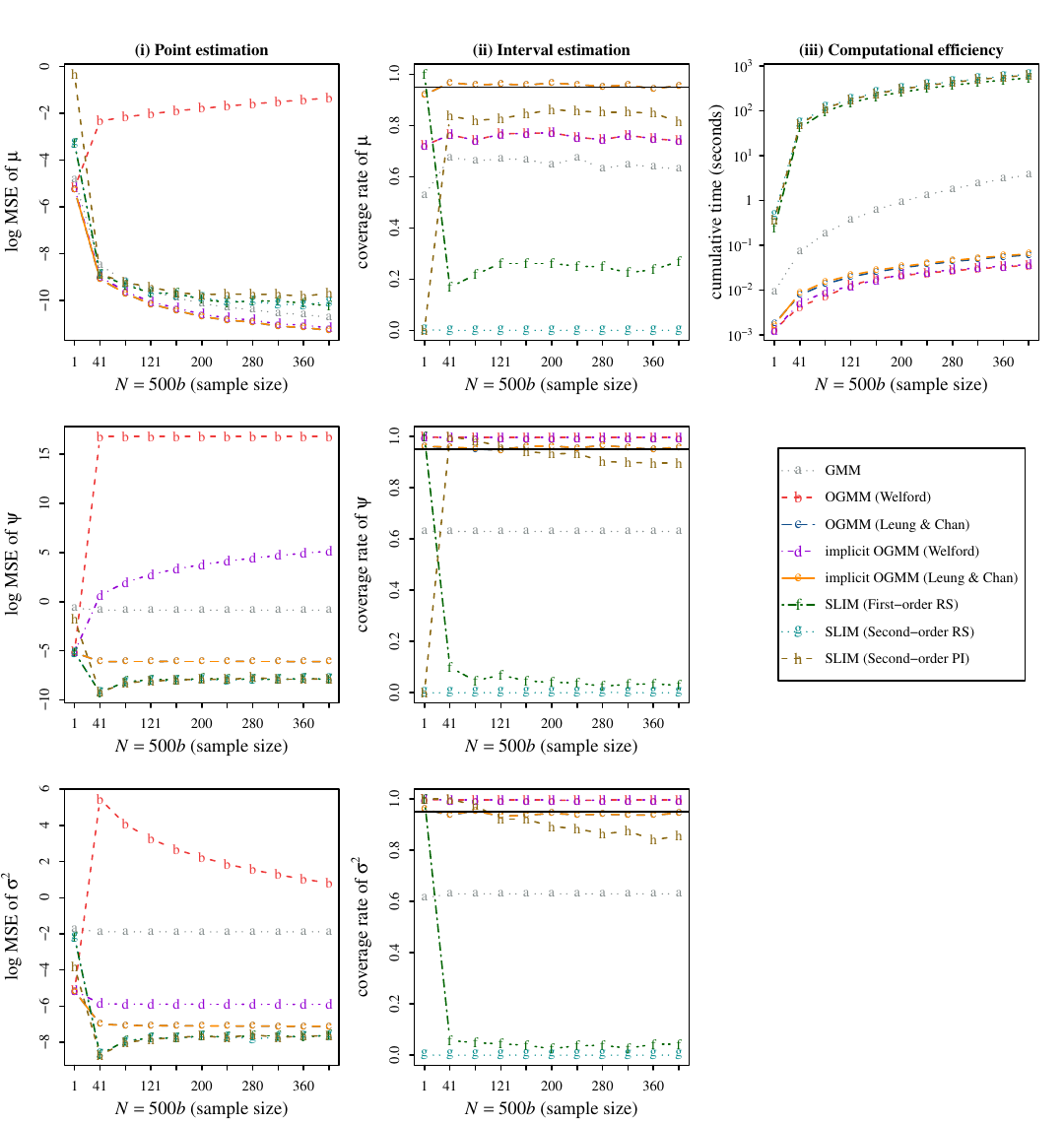}
	\caption{Parameter estimation of $\theta^*$ in the dependent and nonlinear MA(1) model.
	$N_b=500b$ for $b=1,\ldots,400$.
    For (i) and (ii), $\OGMM$ (Leung \& Chan) and implicit $\OGMM$ (Leung \& Chan) coincide with each other.
    }
	\label{fig:51(II-MA)1_MA1}
\end{figure}
}

\subsection{Online Sargan--Hansen test} \label{sec:online-sargan-hansen-test}

In \cref{sec:over-identifying-restrictions-test}, we observe that there is a first-order correction term in our online Sargan--Hansen test statistic because we are plugging in a sequence of estimates rather than a single estimate in the offline setting.
On the other hand, there is no first-order correction term in \eqref{eq:sgmm-overident}.
To investigate the influence of this term, consider the following models:

\begin{enumerate}[resume*=enum-model]
	\item \label{enum:overident-ind} Independent: let
	$y_k = \theta_1^* x_k +\theta_2^* z_{k,1} +\varepsilon_k$ and
	$x_k = z_{k,1} +z_{k,2} +\nu_k$, where
	$(z_{k,1}, z_{k,2}, \nu_k, \varepsilon_k)^\T \simIID \Normal_4(0, \mathcal{V})$ with $\mathcal{V}_{i,j} = \I_{i=j} +0.5 \I_{(i,j) \in \{(1,2),(2,1),(3,4),(4,3)\}}$.
	This is the simulation model in \citet{hall2000overident}, and we consider $\theta_1^* = 1$ and $\theta_2^* \in \{0, 0.05, \ldots, 0.5\}$.
	\item \label{enum:overident-dep} Dependent: same as model \ref{enum:overident-ind} but $(z_{k,1}, z_{k,2}, \nu_k, \varepsilon_k)^\T$ follows a vector autoregressive model of order $1$ with mean $0_4$, coefficient matrix $0.5 I_4$ and Gaussian noise covariance matrix $\mathcal{V}$ in model \ref{enum:overident-ind}.
\end{enumerate}
We estimate $\theta_1^*$ at $N_b = (100) 2^b$ for $b = 1, 2, \ldots, 5$ under the assumption that $\E\{ z_k(y_k -\theta_1^* x_k)\} = 0$.
Clearly, misspecification happens when $\theta_2^* \ne 0$.
The method settings are same as those for model \ref{enum:iv-arma}.
We continue to use two-step $\GMM$ as the offline benchmark and initial estimator, and
$\SGMM$ with $\kappa=0.5$ and $a=0.501$.

\begin{figure}[!t]
	\centering
	\includegraphics[width=\textwidth]{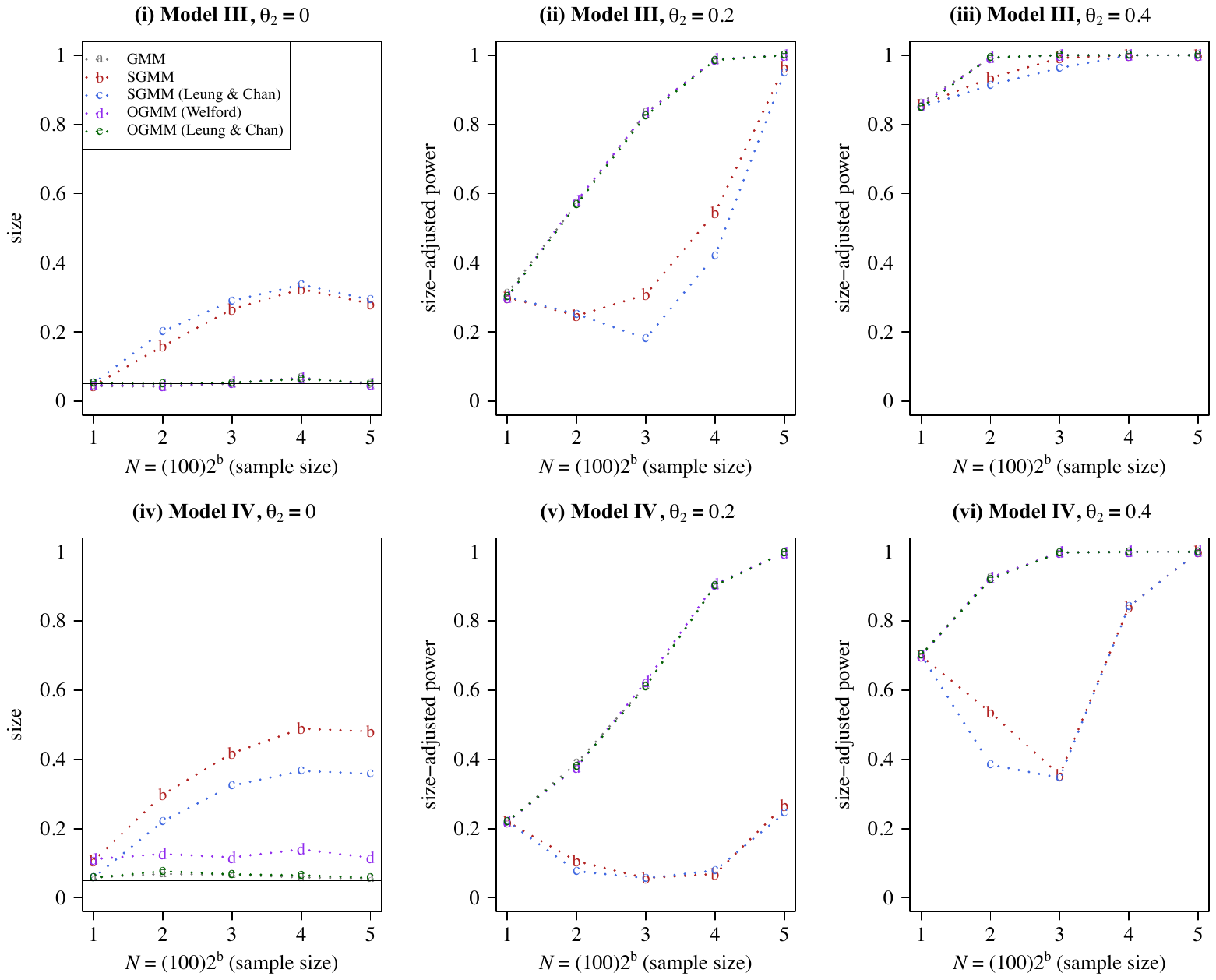}
	\caption{Online Sargan--Hansen test at $5\%$ nominal level under the independent model \ref{enum:overident-ind} (upper panel) and dependent model \ref{enum:overident-dep} (lower panel).
	}
	\label{fig:overident}
\end{figure}

Figure \ref{fig:overident} reports the empirical size and size-adjusted power of different Sargan--Hansen tests for $\theta_2^* \in \{0, 0.2, 0.4\}$.
Overall, the performance of $\OGMM$ is very encouraging.
Although the batch size $n_b = N_{b-1}$ is increasing, the statistical efficiency of $\OGMM$ remains comparable to that of $\GMM$, which corroborates with our theory.
Recursive LRV estimation does not cause size distortion under model \ref{enum:overident-ind} and allows autocorrelation-robust online testing under model \ref{enum:overident-dep}.
On the other hand, $\SGMM$ suffers from considerable size distortion and exhibits non-monotonic size-adjusted power (with respect to the sample size) that seems counter-intuitive and cannot be explained by the variance estimator.
This confirms the validity and necessity of the first-order correction term in Theorem \ref{thm:normality-moment} when a sequence of estimates is plugged in.
The results for other values of $\theta_2^*$ are similar and deferred to the \Supp.

\subsection{Online quantile regression} \label{sec:online-quantile-regression}

In Example \ref{eg:quant}, we reviewed that $\LEQR$ \citep{chen2019rindicator} is also based on linearization and explicit updates; however it differs methodologically from $\OGMM$ in two key aspects.
Specifically, $\OGMM$ relies on our proposed lagged linearization and follows ``online utilization'' principle, as detailed in Example \ref{eg:quant}.
To illustrate the benefits of these methodological advancements, we now consider the following models:
\begin{enumerate}[resume*=enum-model]
	\item \label{enum:quant-ind} Independent: let
	$
	y_k = x_k^\T \theta^* +\varepsilon_k,
	$
	where $x_{k,1} = 1$,
	$x_{k,2}, \ldots, x_{k,p}$ follow a uniform distribution on $[0,1]$ with $\Corr(x_{k,i}, x_{k,j}) = 0.5^{|i-j|}$, and
	$\varepsilon_k \simIID \Normal(0,1)$.
	We set $\theta^* = (1, \ldots, 1)^\T$ and $p=10$, which replicates a simulation model in \citet{chen2019rindicator}.
	\item \label{enum:quant-dep} Dependent: same as model \ref{enum:quant-ind} but $(\varepsilon_k, x_{k,2}, \ldots, x_{k,p})^\T$ follows a vector autoregressive model of order $1$ with mean $0_p$, coefficient matrix $0.5I_p$ and Gaussian noise covariance matrix $\mathcal{V}_{i,j} = 0.5^{|i-j|} \I_{i \ne 1} \I_{j \ne 1} +\I_{i = 1} \I_{j = 1}$ for $1 \le i,j \le p$.
\end{enumerate}
We are interested in estimating $\theta_\tau^*$ at $N_b = 2000b$ for $b = 1, \ldots, 250$.
Recall that $\theta_\tau^*$ can be computed by shifting $\varepsilon_k$ such that $\pr(\varepsilon_k \le 0 \mid x_k) = \tau$ \citep{chen2019rindicator}.
For $\tau = 0.1$, the intercept $\theta_{\tau 1}^*$ in models \ref{enum:quant-ind} and \ref{enum:quant-dep} are approximately $-0.28$ and $-0.48$, respectively.
The method settings are close to those for model \ref{enum:iv-het}.
The differences are that we use the default modified Barrodale and Roberts algorithm in the R-package \texttt{quantreg} \citep{koenker2005quant} for initial estimation,
$\LEQR$ as the online benchmark, and
$\SGMM$ with $\kappa \in \{0.97, 0.98, 0.99\}$ and $a=0.501$.
For the smoothing in Example \ref{eg:quant}, we follow \citet{chen2019rindicator} to use $H(x) = \{1/2 +(15/16)(x -2x^3/3 +x^5/5)\} \I_{|x| < 1} +\I_{x \ge 1}$ and $h_b = \sqrt{p/N_{b-1}}$, where $b$ is the current number of intervals ($\LEQR$) or batches ($\SGMM$/$\OGMM$).
{In addition,
we also consider the stochastic subgradient descent algorithm $\SSGD$  by \citet{lee2025fast}, where new iteration follows from a new sample under the online setting.
}

\begin{figure}[!t]
	\centering
	\includegraphics[width=\textwidth]{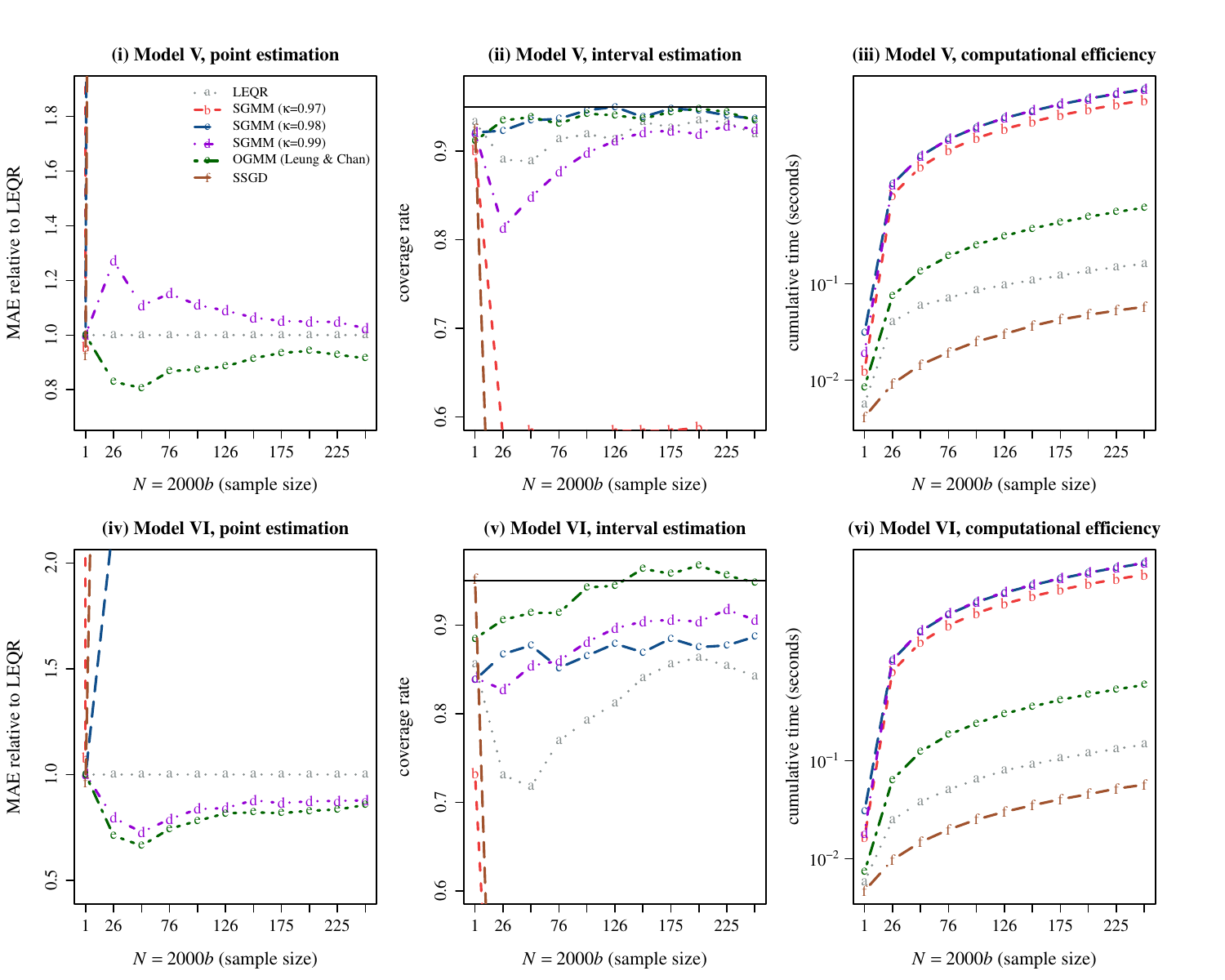}
	\caption{Online quantile regression under the independent model \ref{enum:quant-ind} (upper panel) and dependent model \ref{enum:quant-dep} (lower panel) for $J^\top \theta_\tau^*$,
    where the values of $\SGMM$ ($\kappa=0.97$), $\SGMM$ ($\kappa=0.98$) and $\SSGD$ are too large to display, which are $9.15\times 10^3$, $147.65$ and $45.00$ under model \ref{enum:quant-ind} and $1.24\times 10^7$, $6.80$ and $17.92$ under model \ref{enum:quant-dep}, respectively, for $b=101$.
    }
	\label{fig:reg-quant}
\end{figure}

In Figure \ref{fig:reg-quant}, we compare the median absolute errors (MAEs) of point estimates, the computational time in seconds, and the coverage rates of $95\%$ confidence intervals for the summed coefficient $J^\T \theta_\tau^*$,
where $J = (1, \ldots, 1)^\T$ and $\tau=0.1$.
$\SGMM$ with $\kappa \in \{0.97, 0.98\}$ and $\SSGD$ fail to
give reasonable estimates in some replications.
Figure~\ref{fig:reg-quant} illustrates that $\OGMM$'s methodological differences from $\LEQR$ are beneficial, {although the smoothing bandwidth is not treated in the $\OGMM$ theory}.
The computational cost of $\OGMM$ is slightly higher due to recursive LRV estimation, but it enables autocorrelation-robust inference in model \ref{enum:quant-dep}.
On the other hand, {$\SGMM$ is more sensitive to the choice of $\kappa$ and less stable than $\OGMM$ here} because many entries of $\nabla g(\check\theta_0, x)$ in \eqref{eq:sgmm-lr} are zero or very small.
It is also computationally inefficient because the general $\SGMM$ algorithm does not include a Sherman--Morrison--Woodbury formula for $(\check{V}_k^\top \check{W}_k \check{V}_k)^{-1}$.
{
The \Supp~adds the results for estimating and inferring the first entry of $\theta_\tau^*$ with an additional method by \citet{he2023smoothed} in \cref{sec:additional-quant} of the \Supp.}

\subsection{Online anomaly detection} \label{sec:online-anomaly-detection}

Recall from Example \ref{eg:stability} that there are two types of abnormality or structural instability caused by a change point in parameter value and a misspecified model, respectively.
Now, we extend the experiments in \cref{sec:online-sargan-hansen-test} by modifying the models there:

\begin{enumerate}[resume*=enum-model]
	\item \label{enum:stable-ind} Independent: let
	$y_k = (\theta_1^*+\theta_2^* \I_{k \in \Lambda_\text{I}}) x_k +\theta_2^* \I_{k \in \Lambda_\text{O}} z_{k,1} +\varepsilon_k$ and
	$x_k = z_{k,1} +z_{k,2} +\nu_k$, where
	$(z_{k,1}, z_{k,2}, \nu_k, \varepsilon_k)^\T \simIID \Normal_4(0, \mathcal{V})$ with $\mathcal{V}_{i,j} = \I_{i=j} +0.5 \I_{(i,j) \in \{(1,2),(2,1),(3,4),(4,3)\}}$,
	$\Lambda_\text{I} = \{2001, \ldots, 2500, 6001, \ldots, 6500 \}$ and
	$\Lambda_\text{O} = \{4001, \ldots, 4500, 8001, \ldots, 8500 \}$.
	Same as model \ref{enum:overident-ind}, we consider $\theta_1^* = 1$ and $\theta_2^* \in \{0, 0.05, \ldots, 0.5\}$.
	\item \label{enum:stable-dep} Dependent: same as model \ref{enum:stable-ind} but $(z_{k,1}, z_{k,2}, \nu_k, \varepsilon_k)^\T$ follows the vector autoregressive model in model \ref{enum:overident-dep}.
\end{enumerate}

We estimate $\theta_1^*$ at $N_b = 500 b$ for $b = 1, 2, \ldots, 20$ under the assumption that $\E\{ z_k(y_k -\theta_1^* x_k)\} = 0$.
There are abnormalities caused by a change point in parameter value and a misspecified model for $b \in \{5,13\}$ and $b \in \{9,17\}$, respectively.
To detect them, consider

\begin{enumerate}
	\item \citeauthor{luo2022qif}: the detection statistic in \citet{luo2022qif} with $\TSLS$ as the initial estimator:
	\begin{align*}
		T_\mathrm R &= n_1^{-1} G(\tilde\theta_{b,\mathrm R};D_1)^\top \tilde C_1^{-1} G(\tilde\theta_{b,\mathrm R};D_1) + n_b^{-1} G(\tilde\theta_{b,\mathrm R};D_b)^\top \tilde C_b^{-1} G(\tilde\theta_{b,\mathrm R};D_b), \quad \text{where} \\
		\tilde{\theta}_{b,\mathrm{R}} &= \argmin_{\theta\in\Theta} \left\{
			n_1^{-1} G(\theta;D_1)^\top \tilde C_1(\theta)^{-1} G(\theta;D_1) +
			n_b^{-1} G(\theta;D_b)^\top \tilde C_b(\theta)^{-1} G(\theta;D_b)
		\right\}
	\end{align*}is a restricted estimator, and $\tilde C_i(\theta)$ is the sample variance that depends on $D_i$ for $i\in\{1,b\}$.
	\item $\OGMM$ I: the online detection statistic in \eqref{eq:stable-full} with $\TSLS$ as the initial estimator.
	The $\OGMM$ summary statistics are based on $D_1$ only (initial) or updated with data batches where abnormality is not detected (cumulative).
	\citeauthor{rlrv}'s \citeyearpar{rlrv} estimator is used for inference and the weighting matrix.
	\item $\OGMM$ O: same as $\OGMM$ I but \eqref{eq:stable-unrestricted} is used.
	We expect that $\OGMM$ O should be sensitive to model misspecification only.
\end{enumerate}

\begin{figure}[!t]
	\centering
	\includegraphics[width=\textwidth]{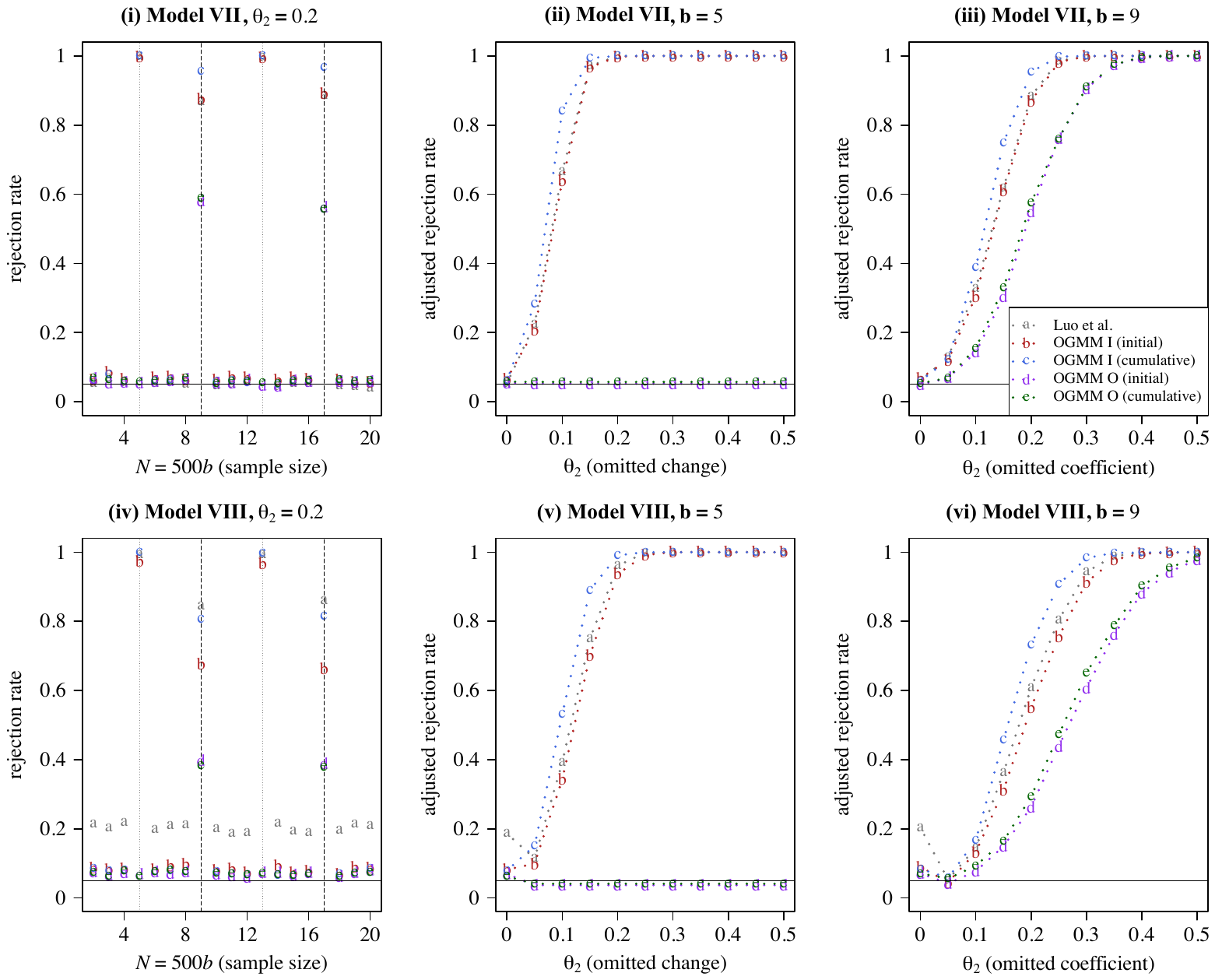}
	\caption{Online anomaly detection at $5\%$ nominal level under the independent model \ref{enum:stable-ind} (upper panel) and dependent model \ref{enum:stable-dep} (lower panel).
	The dotted and dashed lines in plots (i) and (iv) are the locations of abnormality caused by a change point in parameter value and a misspecified model, respectively.
	}
	\label{fig:stable}
\end{figure}

We compare the empirical rejection rates and size-adjusted rejection rates of different proposals.
Figure \ref{fig:stable} reports the results for $\theta_2^* = 0.2$ and for $b \in \{5,9\}$.
It confirms that $\OGMM$ is able to preserve the statistical properties when we replace offline statistics with online counterparts in our framework.
For example, $\OGMM$ O only rejects when $b=9$ but not when $b=5$, which helps identify the two possible causes of structural instability.
We can also use more normal data to improve the sensitivity of $\OGMM$ I without worrying about the computational cost.
Under model \ref{enum:stable-ind}, there are $38\%$ (initial) and $22.6\%$ (cumulative) reductions in averaged time compared with \citet{luo2022qif}.
Under model \ref{enum:stable-dep}, the reductions change to $26.7\%$ (initial) and $10\%$ (cumulative) due to LRV estimation, but it allows autocorrelation-robust online testing.
The results for other values of $\theta_2^*$ or $b$ are similar and deferred to the \Supp.

\section{Applications} \label{sec:applications}

In \cref{sec:sv}, we consider the rolling approach of stochastic volatility modeling \citep{pascalau2023sv} and find that our $\OGMM$ can be a handy addition to a forecaster's toolbox.
In \cref{sec:sensor}, we show that $\OGMM$ can be extended to the frequency domain and offers some alternative insights into the inertial sensor calibration problem.
{
The former application represents a conventional use of $\GMM$ in economics.
Although the latter application originates from engineering, it provides a complementary perspective on how the proposed $\OGMM$ can handle a broader class of nonlinear moment conditions beyond traditional econometric applications.
}

\subsection{Stochastic volatility modeling} \label{sec:sv}

Forecasting is a central topic in econometrics and statistics.
With wider availability of high frequency data, practitioners are interested in updating models quickly to generate real-time prediction.
However, traditional stochastic volatility models are usually trained with daily aggregated data.
In light of it, \citet{pascalau2023sv} proposed a rolling approach to incorporate intraday information.
We consider different online estimation methods via $\OGMM$ here.
Our data set consists of $7,976,478$ one-minute bar euro dollar rate (EUR/USD) from 2000 to 2023 (UTC-5), which are freely available on \href{https://www.histdata.com/}{HistData.com}.
We exclude 1,681 observations corresponding to weekends (UTC+2) and
remove two exchange-rate observations that are anomalous due to
negativity and extreme jumps;
see Section~\ref{sec:implementation-sv} for details.
Then we
compute the close-to-close log 5-minute return $r_{t,h,m}$ for $m = 1,\ldots, 12$ in the $h$-th hour on day $t$.
Define the rolling daily realized variance and realized quarticity
\begin{align*}
	\RV_{t,h}
	= \sum_{i=h-23}^{h} \sum_{j=1}^{12} r_{t,i,j}^2 \quad \text{and} \quad
	\RQ_{t,h}
	= \frac{24 \times 12}{3} \sum_{i=h-23}^{h} \sum_{j=1}^{12} r_{t,i,j}^4,
\end{align*}
where $r_{t,i-\ell,j} = r_{t-1,i+24-\ell,j}$ if $i-\ell \in (-24, 0]$, and $r_{t,i+\ell,j} = r_{t+1,i-24+\ell,j}$ if $i+\ell \in (24, 48]$.
We forecast the one-day-ahead daily realized variance $\RV_{t+1,24}$ using the following approaches with all available observations to date.

\begin{enumerate}
	\item Traditional: the quarticity-adjusted heterogeneous autoregressive model
	\begin{equation} \label{eq:harq}
		\RV_{t,24} = \theta_0^* +(\theta_1^* +\theta_2^* \sqrt{\RQ_{t-1,24}}) \RV_{t-1,24} +\theta_3^* \sum_{\ell=1}^5 \frac{\RV_{t-\ell,24}}{5} +\theta_4^* \sum_{\ell=1}^{21} \frac{\RV_{t-\ell,24}}{21} +\varepsilon_{t,24}
	\end{equation}
	in \citet{bollerslev2016sv} is estimated using ordinary least squares and updated daily.
	Following \citet{pascalau2023sv}, the initial sample size is $1000$.
	\item Rolling: the model in \eqref{eq:harq} is modified to
	\begin{equation} \label{eq:harq-roll}
		\RV_{t,h} = \theta_0^* +(\theta_1^* +\theta_2^* \sqrt{\RQ_{t-1,h}}) \RV_{t-1,h} +\theta_3^* \sum_{\ell=1}^5 \frac{\RV_{t-\ell,h}}{5} +\theta_4^* \sum_{\ell=1}^{21} \frac{\RV_{t-\ell,h}}{21} +\varepsilon_{t,h}
	\end{equation}
	using the rolling approach \citep{pascalau2023sv}.
	Therefore, the initial sample size becomes $1000 \times 24$.
	Other settings are same as Traditional and we focus on forecasting $\RV_{t+1,24}$ at day $t$'s end.
	\item $\OGMM$: consider \eqref{eq:harq-roll} with different moment conditions.
	(LS) least squares; see Example \ref{eg:ls}.
	(QR) quantile regression with $\tau = 0.5$; see \cref{sec:online-quantile-regression}.
	(IV) instrumental variables regression with $1$, $\sum_{\ell=1}^5 \RV_{t-\ell,h}/5$, $\sum_{\ell=1}^{21} \RV_{t-\ell,h}/21$ and $\{ \sqrt{\RQ_{t-i,h}} \}_{i=k}^{21}$ as instruments, and $\TSLS$ as the initial estimator; see \cref{sec:online-instrumental-variables-regression}.
	As \citet{bollerslev2016sv} suggested that the measurement error in $\RV_{t-1,h}$ is most serious, we try to use $\{ \sqrt{\RQ_{t-i,h}} \}_{i=k}^{21}$ as instruments for terms that involve $\RV_{t-1,h}$.
	The weighting matrix is inverted from \citeauthor{welford1962recursive}'s \citeyearpar{welford1962recursive} estimator for (LS) and (QR), and \citeauthor{rlrv}'s \citeyearpar{rlrv} estimator for (IV).
	Other settings are same as Rolling.
	\item $\SGMM$: same as $\OGMM$.
	The learning rate is chosen by \eqref{eq:sgmm-lr} with $\kappa=0.5$ and $a=0.501$.
\end{enumerate}

\begin{table}[!t]
	\centering
	\scalebox{1}{
		\begin{tabular}{lcccccc}
			\toprule[2pt]
			& \multicolumn{3}{c}{MSE relative to Traditional} & \multicolumn{3}{c}{Time relative to Traditional} \\
			\cline{2-4} \cline{5-7}
			& Rolling 	& $\OGMM$ 	& $\SGMM$  & Rolling 	& $\OGMM$ 	& $\SGMM$  \\
			\midrule[0.5pt]
			LS 		 & 0.997	& 0.997	& 1.088	& 19.221 	& 0.118	& 0.368 \\
			QR 		 & --		& 0.994	& 1.076	& --		& 0.251	& 0.661 \\
			IV ($k=1$) & --		& 0.990 	& 1.015 & --		& 0.289 & 0.791 \\
			IV ($k=2$) & --		& 0.990 	& 1.019 & --		& 0.270  & 0.789 \\
			IV ($k=3$) & --		& 0.980 	& 1.013 & --		& 0.275 & 0.780  \\
			IV ($k=4$) & --		& 0.970 	& 1.008 & --		& 0.250  & 0.793 \\
			IV ($k=5$) & --		& 0.960 	& 1.004 & --		& 0.249 & 0.766 \\
			\bottomrule[2pt]
		\end{tabular}
	}
	\caption{One-day-ahead forecasting performances of \eqref{eq:harq-roll} estimated using different methods.
	The MSEs and computation times are relative to the traditional method in \eqref{eq:harq}.}
	\label{tab:sv}
\end{table}

Table \ref{tab:sv} reports the MSE and computation time relative to the traditional approach to update models and generate predictions.
It confirms that intraday information can be used to deliver more accurate forecasts, which is consistent with the findings in \citet{pascalau2023sv}.
Furthermore, $\OGMM$ LS vastly improves the computational efficiency of the rolling approach by reducing its computing time by approximately 160 times.
Meanwhile, $\OGMM$ QR and IV open the door to statistical refinements by further reducing the relative MSE from 0.997 to 0.994 and then to 0.96 (with $k=5$).
$\SGMM$ is unable to achieve them as it is affected by the choice of learning rate and its computational cost is approximately 2.5 to 3 times that of $\OGMM$.
In practice, $\OGMM$ allows us to conduct online estimation and inference easily given any moment condition.
Users can maintain several sets of $\OGMM$ summary statistics and ensemble forecasts in real time.
The efficiency and flexibility of $\OGMM$ are thus valuable in the analysis of large-scale streaming data.

\subsection{Inertial sensor calibration} \label{sec:sensor}

Inertial sensors are very common in modern vehicles and mobile robots.
However, they may be corrupted by stochastic errors of complex structure, e.g.,
\begin{equation} \label{eq:sensor}
	y_t = \sum_{i=1}^2 z_t^{(i)} +u_t^{(3)},
\end{equation}
where $\{ y_t \}$ are the observed error signals, $z_t^{(i)} = \rho_i z_{t-1}^{(i)} +u_t^{(i)}$ follow latent first-order autoregressive processes, and $u_1^{(i)},u_2^{(i)},\ldots \simIID \Normal(0, \sigma_i^2)$ and are independent across $i$ \citep{stebler2014gmwm}.
To estimate $\theta^* = (\rho_1, \sigma_1^2, \rho_2, \sigma_2^2, \sigma_3^2)^\T$,
we can apply the generalized method of wavelet moments ($\GMWM$, \citealp{guerrier2013gmwm}).
Specifically, let $w_{t,j}$ be the maximal overlap discrete wavelet transform of $y_t$ at scale $\tau_j = 2^{j-1}$ using the Haar filter.
Then, the $j$-th element of the moment function is
\begin{align*}
	w_{t,j}^2 -\sum_{i=1}^2 \frac{(0.5\tau_j -3\rho_i -0.5\tau_j \rho_i^2 +4 \rho_i^{1 +0.5\tau_j} -\rho_i^{1+\tau_j}) \sigma_i^2}{0.5 \tau_j^2 (1-\rho_i)^2 (1-\rho_i^2)} -\frac{\sigma_3^2}{\tau_j}, \quad \text{for} \quad
	j = 1, \ldots, q;
\end{align*}
see the R-package \href{https://github.com/SMAC-Group/gmwm}{\texttt{gmwm}} and Table 3 in \citet{guerrier2013gmwm}.
When $q > 5$, the Sargan--Hansen test can be used to assess the goodness of fit of \eqref{eq:sensor}.

Now, consider the y-axis gyroscope data set available through \texttt{imar.gyro} in the R-package \href{https://github.com/SMAC-Group/imudata}{\texttt{imudata}}.
$5,760,000$ error signals were collected at 400 Hz in 4 hours.
We divide them into $b=48$ equal batches (each of size 120,000) and fit \eqref{eq:sensor} using $q = 10$ and

\begin{enumerate}
	\item $\GMWM$: the two-step $\GMWM$ implemented in the R-package \href{https://github.com/SMAC-Group/gmwm}{\texttt{gmwm}} with full data.
	The default first-step weighting matrix is diagonal, which may not be optimal.
	Therefore, parametric bootstraps of size $100$ are used for the second-step weighting matrix in estimation and the asymptotic variance in inference separately.
	\item $\OGMM$: $\OGMM$ with normal data batches.
	Using the two-step $\GMWM$ as the initial estimator, we update $\OGMM$ summary statistics if \eqref{eq:stable-full} and \eqref{eq:stable-unrestricted} are not rejected; see the cumulative strategy in \cref{sec:online-anomaly-detection}.
\end{enumerate}

\begin{table}[!t]
	\centering
	\scalebox{1}{
		\begin{tabularx}{0.89\textwidth}{Xcc}
			\toprule[2pt]
			& $\GMWM$ & $\OGMM$ \\
			\midrule[1pt]
			$\rho_1 \times 10^{-1}$ & 9.592 (9.592, 9.592) & 9.910 (9.910, 9.910)  \\
			$\sigma_1^2 \times 10^{-12}$ & 23.217 (17.620, 28.814) & 4.191 (3.174, 5.209) \\
			$\rho_2 \times 10^{-2}$ & 7.529 (7.529, 7.529) & 7.733 (7.733, 7.733) \\
			$\sigma_2^2 \times 10^{-7}$ & 3.445 (3.409, 3.482) & 3.350 (3.281, 3.419) \\
			$\sigma_3^2 \times 10^{-8}$ & 3.234 (2.864, 3.605) & 4.177 (3.499, 4.855) \\
			Estimation time (seconds) & 298.46 & 11.19 \\
			Total time (seconds) & 594.33 & 175.59 \\
			Sargan--Hansen test ($p$-value) & 0.000585 & 0.107 \\
			\bottomrule[2pt]
		\end{tabularx}
	}
	\caption{Estimated parameters with associated 95\% confidence intervals for \eqref{eq:sensor}.}
	\label{tab:sensor}
\end{table}

Table \ref{tab:sensor} reports the estimation results.
For $\GMWM$, the Sargan–Hansen test suggests that there are abnormal data batches, or \eqref{eq:sensor} is a poorly fitted model.
To investigate the latter possibility, we also estimate (\ref*{eq:sensor}$-$) $y_t = z_t^{(1)} +u_t^{(2)}$ and (\ref*{eq:sensor}$+$) $y_t = \sum_{i=1}^3 z_t^{(i)} +u_t^{(4)}$, which are considered in \citet{guerrier2013gmwm} and \citet{stebler2014gmwm}.
The Sargan--Hansen test $p$-values are (\ref*{eq:sensor}$-$) $1.021 \times 10^{-14}$ and (\ref*{eq:sensor}$+$) $8.822 \times 10^{-10}$, which indicate worse fit.
In contrast, applying $\GMWM$ and $\OGMM$ to the first data batch leads to $p$-values of $0.494$ and $0.173$, respectively.
This evidence supports structural stability testing in $\OGMM$.
Using \eqref{eq:stable-full} and \eqref{eq:stable-unrestricted}, we find that $17$ batches of data (number 10, 11, 13, 14, 15, 16, 17, 19, 21, 24, 25, 29, 36, 39, 41, 45, 46) are normal, and most abnormalities are due to model misspecification because \eqref{eq:stable-unrestricted} is rejected for $28$ batches of data.
The resulting $\OGMM$ estimate of $\sigma_1^2$ is much smaller than $\GMWM$, while estimates of other parameters are closer.
The $\OGMM$ $p$-value $0.107$ suggests better screened-sample fit than $\GMWM$, although it is not adjusted for possible post-screening issues.
In terms of total time (estimation and testing), \eqref{eq:stable-unrestricted} requires solving an optimization problem for each data batch.
Nevertheless, $\OGMM$ offers online estimates and remains faster than $\GMWM$ because it does not perform bootstrap for optimal estimation and inference.
Overall, this application illustrates that the $\OGMM$ framework is very general and offers an alternative approach to estimation and inference for time series models formulated in wavelet domains.

\section{Discussion} \label{sec:discussion}

Time series data are inherently serially dependent and sequential.
Motivated by the need to analyze large streaming time series data, we develop an online version of the classical $\GMM$,
which is robust to weak serial dependence, preserves the asymptotic efficiency of offline $\GMM$, and vastly reduces the time complexity and memory requirements.
Our methodological development covers $\OGMM$ point estimation, confidence region construction, over-identifying restrictions testing and anomaly detection with rigorous theoretical justification.
Our numerical simulation and data illustration convincingly demonstrate the favorable performance of $\OGMM$-based estimation and inference as compared to several existing procedures in terms of computational time and statistical errors.

To conclude, we mention a few potential research topics.
Throughout this paper the dimension of the parameter is assumed fixed.
It would be interesting to develop an online version of $\GMM$ that can accommodate high dimensionality of the parameter; see \citet{belloni2018hd} for regularized $\GMM$ in the offline setting.
Additionally, $\GMM$ has been extended to handle the estimation of a finite dimensional parameter when a nuisance function is present; see \cite{ackerberg2014semi} and \citet{chen2015semi}.
Developing an online version of two-step $\GMM$ to deal with this semiparametric setting would be interesting.
Finally, it is of interest to allow for time-varying parameter and/or local stationarity of streaming time series, as stationarity and time-invariant parameter might be too strong to be satisfied for a very long streaming time series.
One possible direction is to introduce an exponential smoothing factor as in \citet{luo2023qif} to the $\OGMM$ summary statistics, and develop a new online LRV estimator for efficient estimation and inference.
We leave them for future investigation.

\section*{Funding sources}

Chan was supported in part by grants
GRF-14306421 and GRF-14307922
provided by the Research Grants Council of HKSAR.
Shao would like to acknowledge partial financial support from National Science Foundation DMS-2210002 and DMS-2412833.

\appendix
\setcounter{equation}{0}
\setcounter{algocf}{0}
\setcounter{table}{0}
\setcounter{figure}{0}
\renewcommand{\theequation}{\thesection.\arabic{equation}}
\renewcommand{\thealgocf}{\thesection.\arabic{algocf}}
\renewcommand{\thetable}{\thesection.\arabic{table}}
\renewcommand{\thefigure}{\thesection.\arabic{figure}}
\renewcommand{\thetheorem}{\thesection.\arabic{theorem}}
\renewcommand{\thelemma}{\thesection.\arabic{lemma}}
\renewcommand{\thecorollary}{\thesection.\arabic{corollary}}
\renewcommand{\theproposition}{\thesection.\arabic{proposition}}
\renewcommand{\theexample}{\thesection.\arabic{example}}
\renewcommand{\theassumption}{\thesection.\arabic{assumption}}
\renewcommand{\theremark}{\thesection.\arabic{remark}}

\newpage 
\section*{Supplement}
In the supplement, we discuss different strategies to estimate the long-run variance online in Section~\ref{sec:long-run-variance}. 
We also collect the proofs of main theorems in Section~\ref{sec:proof-of-theorems}. 
Section~\ref{sec:consistency_fix_batch} gives a fixed batch consistency theorem for $\OGMM$ under stronger global assumptions with proofs in Section~\ref{sec:proof-of-theorems-fix}.  
Auxiliary lemmas and their proofs are in Section~\ref{sec:proof-of-lemmas}.
Sections~\ref{sec:additional-results} and \ref{sec:implementation-details} gather additional simulation results and implementation details.

\section{Online long-run variance estimation} \label{sec:long-run-variance}

Computationally, there are at least three approaches to estimate the long-run variance.
(1) We compute $\hat\Sigma_1$ based on $D_1$ and fix $\hat W = \hat\Sigma_1^{-1}$.
Many classical estimators are available and the computational cost is minimized. Nevertheless, this approach may be inefficient since it does not utilize the sequentially arrived data.
(2) We assume $g(\theta^*, x_{i,j})$ follows some parametric model (e.g., a vector autoregressive moving average model; see \citet[Section~3.5.2]{hall2005gmm}) and update such model recursively. Then, we can plug in the online parameter estimates to obtain $\hat\Sigma_b$ via the long-run variance estimator based on the parametric model and set $\hat W = \hat\Sigma_b^{-1}$. This is similar to using a working correlation matrix in \citet{luo2022qif,luo2023qif}, which leads to a loss of statistical efficiency in case of model misspecification.
(3) We update $\hat W$ using an online nonparametric estimator $\hat\Sigma_b(K)$, where $K$ is a kernel function that satisfies some conditions. This approach yields both statistical efficiency and autocorrelation robustness. While it involves some smoothing parameter, there are automatic optimal selectors in the literature \citep{rtacm,rlrv} so that users only need to provide problem-specific information, e.g., the strength of serial dependence and the memory constraint.

In general, we recommend the third approach. If users have a strict time budget or are certain that the data are independent, they can use a fixed weighting matrix or invert \citeauthor{welford1962recursive}'s \citeyearpar{welford1962recursive} sample variance estimator, which are implemented in our R-package \texttt{ogmm}. However, the default choice is to invert the long-run variance estimator in \citet{rlrv}, whose implementation allows batch updates and includes a positive definiteness adjustment. To the best of our knowledge, the proposed estimator in \citet{rlrv} is also the most efficient in the online nonparametric long-run variance estimation literature. Specifically, it takes the form
\begin{equation}\label{eq:supp-lrv-window-form}
\hat\Sigma_b(K)
= \frac{1}{N_b} \sum_{i=1}^{N_b} \sum_{j=1}^{N_b} K_{N_b}(i, j)(X_i - \bar X_{N_b})(X_j - \bar X_{N_b})^\T,
\end{equation}
{where $X_i=g(\hat\theta_{\beta(i)},x_i)$ is the evaluated moment function indexed by a single subscript with the concatenated
observation batches $x_i$ (see Section~\ref{sec:consistency} for the precise mapping),
the batch index $\beta(i)$ satisfying $N_{\beta(i)-1}<i\leq N_{\beta(i)}$ for $i=1,\ldots,N_b$, and
$\bar X_{N_b} = N_b^{-1} \sum_{i=1}^{N_b} X_i \equiv N_b^{-1} \sum_{j=1}^b G(\hat\theta_j; D_j)$.}
Their kernel takes the form
\begin{align}
	K_n(i, j) &= \left(1 - \frac{|i-j|^\lambda}{t_n^\lambda}\right) \I_{|i-j| \le s'_{\max(i,j)}}, \label{eq:supp-kernel}
\end{align}where $s_n=\min(\lfloor \Psi n^\psi \rfloor, n-1)$; $t_n = \min(\lceil \Xi n^\xi \rceil, n)$;
\begin{align}
	s'_n &= \begin{cases}
		s'_{n-1}+1, & \text{if } s_{n-1} \le s'_{n-1}+1 < \phi s_{n-1},\\
		s_n, & \text{if } s'_{n-1}+1 \ge \phi s_{n-1},
	\end{cases} \label{eq:supp-kernel-s}
\end{align}
$\Psi,\Xi\in\R^+$; $\psi,\xi\in(0,1)$; $\lambda\in\Z^+$; and $\phi\in[1,\infty)$.
The tapering part $\left(1 - |i-j|^\lambda/t_n^\lambda\right)$ mainly controls the statistical efficiency as it assigns weight to $(X_i-\bar X_{N_b})(X_j-\bar X_{N_b})^\T$ for different time lag $|i-j|$ \citep{rlrv}.
In this part, the kernel order $\lambda$ determines the convergence rate of $\hat\Sigma_b(K)$, while the tapering parameter $t_n$ governs the bias--variance trade-off.
The subsampling part $\I_{|i-j| \le s'_{\max(i,j)}}$ mainly controls the computational efficiency as it determines the size and regularity of every subsample \citep{rlrv}.
In this part, the memory parameter $\phi$ in \eqref{eq:supp-kernel-s} governs the space complexity of $\hat\Sigma_b(K)$.
{Using the kernel in \eqref{eq:supp-kernel}, \citet{rlrv} showed that $\hat\Sigma_b(K)$ can be updated in $O(n_bq^2+q^3)$-time and stored in $O(n_bq+q^2)$-space when $\phi \ge 2$.}
In their R-package \texttt{rlaser}, the optimal parameters selector automatically handles $\Psi$, $\Xi$, {and $\psi=\xi=1/(1+2\lambda)$} so users only need to choose $\lambda$ and $\phi$.
Following \citet{rlrv}, we recommend $\lambda=1$ ($\lambda=3$) if the serial dependence is strong (weak), and $\phi=1$ ($\phi=2$) if the memory is abundant (scarce).

{
Now we derive the consistency for the LRV estimator \eqref{eq:supp-lrv-window-form} using kernel \eqref{eq:supp-kernel}.
\begin{proposition}[Consistency of LRV estimator]
\label{prop:plugin-lrv}
Let $x_i$ denote the concatenated observations, $\beta(i)$ the batch index, and $X_i^*=g(\theta^*,x_i)$ for $i=1,\ldots,N_b$.
Suppose the kernel \eqref{eq:supp-kernel} is used and Assumptions~\ref{asum:stationarity}--\ref{asum:lrv} hold, and in particular, $\hat{W}=\hat{W}(\{D_i\}_{i=1}^{b-1})\cin{p} W$ for Assumption~\ref{asum:weight}.
Further assume the following conditions:
\begin{enumerate}
	\item For any $c\in\R^q$, the process $\{c^\top X_i^*\}_{i\in\Z}$ satisfies the assumptions of Theorem~3.1(b) in \citet{rlrv}.
	\item For any $\epsilon>0$, let $B_\epsilon = \{\theta \in \Theta: \norm{\theta-\theta^*}_2 \le \epsilon\}$.
    There exists $L>0$ such that $\sup_{j=1,\ldots,q}| g_{[j]}(\theta,x)-g_{[j]}(\theta^*,x)| \leq L \norm{\theta-\theta^*}_2$ for all $\theta\in B_\epsilon$ and $x\in\mathcal{X}$,
    where $g_{[j]}(\theta,x)$ is the $j$th entry of $g(\theta,x)$.
	\item For any $\epsilon>0$, there exists
    a large enough $N_b$ such that
        $\norm{\hat{\theta}_i-\theta^*}_2\leq \epsilon$, and
        it satisfies the rate that
        $\norm{\hat{\theta}_i-\theta^*}_2=O_p(N_i^{-1/2})$
        as $N_b\to\infty$, for $i=1,\ldots,b$.
	\item The kernel parameters satisfy $\psi=\xi=1/(1+2\lambda)$ for $\lambda\in\Z^+$.
\end{enumerate}
Then $\norm{\hat{\Sigma}_b(K)-\Sigma}_F\cin{p} 0$.
\end{proposition}
Algorithm~\ref{algo:ogmm} updates $\hat{\Sigma}_b(K)$ after computing $\hat{\theta}_b$, thus requiring that $\hat{\theta}_i$ is close enough to $\theta^*$ over $i=1,\ldots,b$ to guarantee the consistency of $\hat{\Sigma}_b(K)$.
In particular, it explains why we should compute $\hat{\theta}_1$ using reasonable method (e.g., $\GMM$) in line 3 of Algorithm~\ref{algo:ogmm} and $N_1=n_1\to\infty$ in Theorems \ref{thm:consistency} and \ref{thm:normality-par}: $\norm{\hat{\theta}_1-\theta^*}_2 =O_p(n_1^{-1/2})$ allows a consistent initial $\hat{\Sigma}_1(K)$ and $\hat{W}=\hat{\Sigma}_1(K)^{-1}$ for future updates,
then the induction steps as in the proof of Theorems \ref{thm:consistency} and \ref{thm:normality-par} establish the consistency of $\hat{\theta}_b$ and $\hat{\Sigma}_b(K)$ for $b=2,3,\ldots$.
Establishing the consistency of the LRV estimator under the
fixed-batch regime falls outside the primary scope of this article
and is left for future research.
}

\section{Consistency under fixed batch size}\label{sec:consistency_fix_batch}
This section discusses the $\OGMM$ consistency of $\hat{\theta}_b$ in Theorem~\ref{thm:tail_entry} under fixed batch size $n_i\equiv n_0<\infty$ for all $i$ and diverging batch numbers $b\to\infty$ so that $N_b=\sum_{i=1}^{b}n_i=bn_0\to\infty$.
Recall $G(\theta;D_i) = \sum_{j=1}^{n_0}g(\theta,x_{i,j})\in\R^{q}$ and $\bar{g}_b(\theta)=N_b^{-1}\sum_{i=1}^{b}G(\theta;D_i)$.
For each batch $D_i$, define the batch-average
$g_i(\theta):= n_0^{-1} G(\theta;D_i)$ and
$V_i(\theta) := n_0^{-1} \nabla G(\theta; D_i) \in\R^{q\times p}$,
and let the population version
$m(\theta) := \E\{g_i(\theta)\} = \E\{g(\theta,x_{i,j})\}$
and
$V(\theta) := \E\{V_i(\theta)\} = \E\{\nabla g(\theta,x_{i,j})\}$.
Note that $m(\theta^*)=0_q$.
Rewrite the $\OGMM$ recursion:
\begin{align*}
	\hat{V}_b & =
    \frac{1}{N_b} \sum_{i=1}^{b}\nabla G(\hat{\theta}_{\min(i,b-1)};D_i) = \frac1b \sum_{i=1}^b V_i(\hat{\theta}_{\min(i,b-1)}), \\
	\hat{U}_b & =
	\frac1b \sum_{i=1}^b \cbk{
         g_i(\hat{\theta}_{\min(i,b-1)}) + V_i(\hat{\theta}_{\min(i,b-1)}) (\hat{\theta}_{b-1}-\hat{\theta}_{\min(i,b-1)}
    }, \\
	\hat{\theta}_b &= \hat\theta_{b-1}-(\hat{V}_b^\top \hat W_b \hat{V}_b)^{-1}\hat{V}_b^\top \hat W_b \hat{U}_b
	= \hat{\theta}_{b-1} - \hat{H}_b^{-1} \hat{h}_b,
\end{align*}where $\hat{h}_b:=\hat{V}_b^\top \hat{W}_b \hat{U}_b$, $\hat{H}_b := \hat{V}_b^\top \hat{W}_b \hat{V}_b$ and $\hat{W}_b$ is the weighting matrix estimated by data $\{D_i\}_{i=1}^b$ with probability limit $W$.
In our case, $\hat{W}_b=\hat{\Sigma}_b^{-1}$, the inverse of an online nonparametric LRV estimator.
Define the population Gauss--Newton map
\begin{align*}
	T(\theta):=\theta-\{\underbrace{V(\theta)^\top W V(\theta)}_{=:H(\theta)}\}^{-1} \underbrace{V(\theta)^\top W m(\theta)}_{=:h(\theta)}
	= \theta-H(\theta)^{-1}h(\theta),
\end{align*}
By $T(\hat{\theta}_{b-1})=\hat{\theta}_{b-1}-H(\hat{\theta}_{b-1})^{-1}h(\hat{\theta}_{b-1})$, write
\begin{align}\label{eq:GN_map}
	\hat{\theta}_b &= \hat{\theta}_{b-1} - \hat{H}_b^{-1}\hat{h}_b
	= T(\hat{\theta}_{b-1}) + \underbrace{H(\hat\theta_{b-1})^{-1}h(\hat\theta_{b-1})-\hat{H}_b^{-1}\hat h_b}_{=:\xi_b}
	= T(\hat{\theta}_{b-1}) + \xi_b.
\end{align}
{The roadmap of proving the consistency in Theorem~\ref{thm:tail_entry} by \eqref{eq:GN_map} is:
\begin{enumerate}
	\item Impose global contraction of $T$; see Assumption~\ref{as:global-contraction}.
	\item Control the stale-iterate perturbation like $\hat{V}_b-V(\hat{\theta}_{b-1})$ by
	\begin{align*}
		P_b:=\frac1b\sum_{i=1}^{b-1}\norm{\hat\theta_i-\hat\theta_{b-1}}_2, \qquad
		Q_b:=\frac1b\sum_{i=1}^{b-1}\norm{\hat\theta_i-\hat\theta_{b-1}}_2^2;
	\end{align*}
	see Assumption~\ref{as:Lip_nabla_global} and Lemma~\ref{sec:proof-convergence-tail}.
	\item Bound $\xi_b$ by using an invariant-ball argument and take limsup to show the asymptotic radius $\limsup_{b\to\infty}\norm{\hat{\theta}_b-\theta^*}$ is arbitrarily small.
\end{enumerate}}

Denote $\lambda_{\min}(A)$, $\lambda_{\max}(A)$ and $\sigma_{\min}(A)$ the minimum and maximum eigenvalue and minimum singular value of matrix $A$, respectively.
To establish the consistency, besides Assumptions~\ref{asum:stationarity}--\ref{asum:moment-dominance}, we make the following stronger assumptions.
We define
\begin{align}
    \Delta_b^W &:= \norm{\hat W_b-W}_2, \label{eq:delta_W} \\
    \Delta_b^V &:= \sup_{\theta\in\Theta} \norm{\frac1b\sum_{i=1}^b V_i(\theta)-V(\theta)}_2
	= \sup_{\theta\in\Theta} \norm{\frac{1}{N_b}\sum_{i=1}^{b}\nabla G(\theta;D_i) - \E\{\nabla g(\theta,x_{i,j})\}}_2,    \label{eq:Delta_V} \\
	\Delta_b^U &:= \sup_{\theta\in\Theta} \norm{\frac1b\sum_{i=1}^b g_i(\theta)-m(\theta)}_2
	= \sup_{\theta\in\Theta} \norm{\frac{1}{N_b}\sum_{i=1}^{b} G(\theta;D_i) - \E\{ g(\theta,x_{i,j})\}}_2. \label{eq:Delta_U}
\end{align}

\begin{customas}{6*}[More on weighting matrix]\label{as:weight-new}
There exist $\underline{\lambda}_W>0$ and $\overline{\lambda}_W<\infty$ such that
\[
    0<\underline{\lambda}_W < \lambda_{\min}(W) < \lambda_{\max}(W)<\overline{\lambda}_W<\infty.
\]
\end{customas}

It is the same as Assumption~1(vi) in \citet{forneron2025convexity} and Assumption~8(A8.2) in \citet{chen2025slim}.
Like (D.2) in \citet{chen2025slim}, we have $\lambda_{\min}(\hat W_b)  \geq \lambda_{\min}(W)-\Delta_b^W$ and $\norm{\hat W_b}_2 \leq \overline{\lambda}_W + \Delta_b^W$.
Together with Assumption~\ref{asum:weight}, with probability approaching one, $\lambda_{\min}(\hat W_b) \geq \underline{\lambda}_W$ and $\norm{\hat{W}_b}_2 \leq \overline{\lambda}_W$.

\begin{customas}{8*}[Compactness and convexity]
\label{as:cvx}
$\Theta$ is a compact and convex set.

\end{customas}
The convexity assumption on $\Theta$ guarantees there are $\bar{\theta}_b\in\Theta$ lying between $\theta^*$ and $\hat{\theta}_b$, to allow the mean value theorem.

\begin{customas}{11*}[Global Lipschitz continuity of $\nabla g(\theta,x)$]\label{as:Lip_nabla_global}
There exists $L \geq 0$ such that for all $\theta_1,\theta_2\in\Theta$, and $x\in\mathcal{X}$,
$\norm{\nabla g(\theta_1,x) - \nabla g(\theta_2,x)}_2 \leq L \norm{\theta_1-\theta_2}_2$.
Moreover, $\sup_{\theta\in\Theta}\norm{V(\theta)}_2
=\sup_{\theta\in\Theta}\norm{\E\{\nabla g(\theta,x_{i,j})\}}_2
=:\overline{\sigma}_V<\infty$.
\end{customas}
Assumption~\ref{as:Lip_nabla_global} is a convenient global smoothness condition.
It is standard in recursive and stochastic approximation analyses of $\GMM$-type procedures; compare Assumption 4(A4.5) in \citet{chen2025slim} and Assumption 1(iii) in \citet{forneron2025convexity}.
Its role here is to control the stale-iterate effect created by evaluating different batches at different past iterates;
see Lemma~\ref{lemma:bound_diff_VUhat}.
In linear moment models, $V_i(\theta)$ does not depend on $\theta$, so one has $L=0$ and the stale-iterate terms disappear; see Example~\ref{eg:linear-models} below.
By this assumption, we also have
$\norm{V_i(\theta_1)-V_i(\theta_2)}_2 \leq L \norm{\theta_1-\theta_2}_2$
and $\norm{V(\theta_1)-V(\theta_2)}_2 \leq L \norm{\theta_1-\theta_2}_2$.

\begin{assumption}\label{as:full_rank_V}
There exists $\kappa_0>0$ such that $\inf_{\theta\in\Theta}\lambda_{\min}\!\big(V(\theta)^\top W V(\theta)\big) \geq \kappa_0>0$.
\end{assumption}

\begin{assumption}[Global contraction of the population Gauss--Newton map]\label{as:global-contraction}
There exists $\rho\in(0,1)$ such that
$\norm{T(\theta)-\theta^*}_2\leq \rho\norm{\theta-\theta^*}_2$
for all $\theta\in\Theta$.
\end{assumption}

\begin{assumption}[Regular problem]\label{as:regularProblem}
Fix $r_0>0$, assume that, for every target accuracy $\gamma>0$ and confidence level $\eta>0$,
there exist $\epsilon_\gamma>0$ and an integer $b_0=b_0(\gamma,\eta)$ such that $\epsilon_\gamma\downarrow 0$ as $\gamma\downarrow 0$ and
\begin{enumerate}
	\item $\kappa_0-\delta_{\times 0}(\epsilon_\gamma,r_0)>0$ and $\Phi_{\epsilon_\gamma}(r_0)\leq r_0$, where $\delta_{\times 0}(\epsilon_\gamma,r_0)$, $C_0'(\epsilon_\gamma,r_0)$, $C_0''(\epsilon_\gamma,r_0)$ and $\Phi_{\epsilon_\gamma}(r_0)$ are defined in \eqref{eq:one-step-C-delta}--\eqref{eq:one-step-Phi};

	\item $\bar{K}_{\epsilon_\gamma}(r_0)<1$ and $K_{0,\epsilon_\gamma}(r_0)\leq \gamma\{1-\bar{K}_{\epsilon_\gamma}(r_0)\}$, where $K_{0,\epsilon_\gamma}(r_0)$ and $\bar{K}_{\epsilon_\gamma}(r_0)$ are defined in \eqref{eq:one-step-K}--\eqref{eq:one-step-Psi}.

	\item $\pr\!\left(
	\sup_{1\le b\le b_0}\norm{\hat\theta_b-\theta^*}_2\le r_0,\
	\sup_{b\ge b_0+1} \max\{\Delta_b^V, \Delta_b^U, \Delta_b^W\} \le \epsilon_\gamma
	\right)\ge 1-\eta$, where $\Delta_b^W$, $\Delta_b^V$, and $\Delta_b^U$ are defined in \eqref{eq:delta_W}--\eqref{eq:Delta_U}, respectively.
\end{enumerate}
\end{assumption}

Theorem~\ref{thm:tail_entry}, stated in the main text,
ensures the weak consistency of $\hat{\theta}_b$, i.e.,
\[
    \hat{\theta}_b \cin{p} \theta^*
\]
even when the batch sizes are fixed
($n_1= n_2=\cdots <\infty$)
as long as the number of batches $b\rightarrow\infty$,
provided that
Assumptions~\ref{asum:stationarity}--\ref{asum:moment-dominance}, \ref{as:weight-new}, \ref{as:cvx}, \ref{as:Lip_nabla_global} and \ref{as:full_rank_V}--\ref{as:regularProblem} hold.
The proof can be found in \cref{sec:proof-tail_entry}.
Assumption~\ref{as:full_rank_V} is a global lower singular value condition on the population Jacobian, while Assumption~\ref{as:global-contraction} is a high-level stability condition on the population Gauss--Newton map.
A simple sufficient condition for Assumption~\ref{as:full_rank_V} is Assumption~\ref{as:weight-new} plus $\inf_{\theta\in\Theta} \sigma_{\min}\{V(\theta)\} > 0$,
because $\lambda_{\min}\{V(\theta)^\top W V(\theta)\} \ge \lambda_W \sigma_{\min}\{V(\theta)\}^2$.
Likewise, under Assumption~\ref{as:cvx}, $T$ is continuously differentiable on $\Theta$ and $\sup_{\theta\in\Theta}\norm{\partial T(\theta)/\partial \theta}_2\le \rho<1$,
then by the mean value theorem, for every $\theta\in\Theta$, there exists $\tilde{\theta}$ between $\theta$ and $\theta^*$ such that
\begin{align*}
\norm{T(\theta)-\theta^*}_2=
\norm{T(\theta)-T(\theta^*)}_2 \le
\sup_{\tilde{\theta}}\norm{\frac{\partial T(\tilde{\theta})}{\partial \theta}}_2 \norm{\theta-\theta^*}_2 \le
\rho \norm{\theta-\theta^*}_2.
\end{align*}
Thus Assumption~\ref{as:global-contraction} follows.
To check Assumptions~\ref{as:full_rank_V} and \ref{as:global-contraction} in practice, we may fix a finite grid
$\{\theta^{(m)}\}_{m=1}^M\subset\Theta$
and evaluate
$\min_{1\le m\le M} \lambda_{\min}\{V(\theta^{(m)})^\top W V(\theta^{(m)})\}$
and $\max_{1\le m\le M} \norm{DT(\theta^{(m)})}_2$ for model-specific evidence.
This way of separating conditioning from contraction is close in spirit to the global rank and injectivity conditions in \citet{forneron2025convexity} and also aligns with the perspective of not requiring consistent initialization emphasized in \citet{chen2025slim}.

Theorem~\ref{thm:tail_entry} is deliberately stated at a high level with Assumption~\ref{as:regularProblem}, {in which the empirical process errors must eventually be uniformly small, and before that happens the
finite initial path must not leave a radius where the perturbed contraction is stable, as required by item (c).}
A convenient sufficient condition is that there exists a deterministic sequence $a_b\downarrow 0$ such that, for every $\eta\in(0,1)$, there is a constant $C_\eta<\infty$ with
$\pr(\sup_{t\geq b} \max\{\Delta_t^V,\Delta_t^U,\Delta_t^W\} \leq C_\eta a_b)\geq 1-\eta$ for all sufficiently large $b$.
Then item (c) follows whenever $b_0$ is chosen so that $C_\eta a_{b_0+1}\leq \epsilon_\gamma$ and
$\pr\sbk{\sup_{1\le b\le b_0}\norm{\hat\theta_b-\theta^*}_2\le r_0}\ge 1-\eta$,
where $b_0$ should be large enough for empirical tail error control, and still small enough that $\{\hat{\theta}_b\}_{b=1}^{b_0}$ remain inside $B(\theta^*,r_0)$ with high probability, respectively.

For items (a)--(b) of Assumption~\ref{as:regularProblem}, the quantity $K_{0,\epsilon}(r)$ is the stochastic floor created by the tail accuracy level $\epsilon$.
The coefficient $K_{1,\epsilon}(r)$ measures the linear distortion of the contraction map $T$, while $K_{2,\epsilon}(r)$ and $K_{3,\epsilon}(r)$ are the quadratic and cubic stale-iterate penalties generated by $P_b$ and $Q_b$.
Thus item (b) simply says that, on the good event in item (c), the perturbed recursion remains contractive on the invariant ball $B(\theta^*,r_0)$.
As $\epsilon_\gamma\downarrow 0$, $K_{0,\epsilon_\gamma}(r_0) \to 0$, $\delta_{\times 0}(\epsilon_\gamma,r_0) \to \delta_{\times 0,\infty} (r_0):=\overline{\lambda}_W (2Lr_0)(2Lr_0+2\overline{\sigma}_V)$ and
\begin{align*}
\bar K_{\epsilon_\gamma}(r_0) &\to
\bar K_{0}(r_0)
:= \rho
+ \frac{4L\overline{\sigma}_V\overline{\lambda}_W}{\kappa_0-\delta_{\times 0,\infty}(r_0)}
\left(\frac{\overline{\sigma}_V^2\overline{\lambda}_W}{\kappa_0}+1\right) r_0
+ \frac{4L^2\overline{\lambda}_W}{\kappa_0-\delta_{\times 0,\infty}(r_0)}
\left(\frac{\overline{\sigma}_V^2\overline{\lambda}_W}{\kappa_0}+1\right) r_0^2.
\end{align*}
Therefore the admissible radius $r_0$ must first satisfy the feasibility conditions $\delta_{\times 0,\infty}(r_0)<\kappa_0$ and $\bar K_0(r_0)<1$.
Once $r_0$ is fixed in this feasible range, continuity implies that there exists a constant $c(r_0)>0$ such that $\epsilon_\gamma \lesssim c(r_0)\gamma$ for all sufficiently small $\gamma$ is sufficient for item (b).
In particular, a larger $r_0$ shrinks the slack $1-\bar K_0(r_0)$ and therefore forces a smaller admissible $\epsilon_\gamma$, which in turn requires a larger $b_0$ in item (c).
Hence there is a genuine trade-off:
a larger $r_0$ makes $\sup_{1\le b\le b_0}\norm{\hat\theta_b-\theta^*}_2\le r_0$ easier to satisfy, but it makes the contraction inequalities tighter and therefore requires a smaller $\epsilon_\gamma$ and a larger $b_0$.
Conversely, a smaller $r_0$ gives more slack in items (a)--(b), making
$\sup_{1\le b\le b_0}\norm{\hat\theta_b-\theta^*}_2\le r_0$ harder.

\begin{remark}[Typical orders for $b_0$]
Suppose the tail bound above holds with rate $a_b$.
Then $b_0$ can be chosen as the smallest integer such that $a_{b_0+1}\le \epsilon_\gamma/C_\eta$.
Typical examples are:
(i) If $a_b \asymp b^{-1/2}$, then $b_0 \asymp \epsilon_\gamma^{-2}$.
(ii) If $a_b \asymp \sqrt{\log b/b}$, then $b_0 \asymp \epsilon_\gamma^{-2}\log(\epsilon_\gamma^{-1})$.
Combining this with $\epsilon_\gamma \lesssim c(r_0)\gamma$ gives $b_0 \asymp \gamma^{-2}$ or $b_0 \asymp \gamma^{-2}\log(\gamma^{-1})$ up to the $r_0$-dependent constant $c(r_0)^{-2}$.
Thus a worse conditioning radius $r_0$ increases $b_0$ through the shrinkage of $c(r_0)$.
\end{remark}

We also provide an almost sure convergence version for reference.
This requires a stronger version of Assumption \ref{as:regularProblem}, as shown below:

\begin{customas}{B.3*}[Regular problem]\label{as:regularProblem_asVersion}
Assume $\Delta_b^W$, $\Delta_b^V$, $\Delta_b^U \cin{a.s.} 0$.
Assume that there exist constants $r_0>0$ and $\epsilon_0>0$ such that
\begin{enumerate}
	\item $\kappa_0-\delta_{\times 0}(\epsilon_0,r_0)>0$ and $\Phi_{\epsilon_0}(r_0)\le r_0$, where $\delta_{\times 0}(\epsilon_0,r_0)$, $C_0'(\epsilon_0,r_0)$, $C_0''(\epsilon_0,r_0)$ and $\Phi_{\epsilon_0}(r_0)$ are defined in \eqref{eq:one-step-C-delta}--\eqref{eq:one-step-Phi};

	\item with
    $\overline{\lambda}_W$ defined in Assumption \ref{as:weight-new},
    $L$ and $\bar{\sigma}_V$ in Assumption \ref{as:Lip_nabla_global}, and
    $\rho$ in Assumption \ref{as:global-contraction}, we have
	\begin{align*}
		\rho + \frac{4L\bar\sigma_V\overline{\lambda}_W}{\kappa_0-\delta_{\times 0}(\epsilon_0,r_0)}\left(\frac{\bar\sigma_V^2\overline{\lambda}_W}{\kappa_0}+1\right) r_0 +
		\frac{4L^2\overline{\lambda}_W}{\kappa_0-\delta_{\times 0}(\epsilon_0,r_0)}\left(\frac{\bar\sigma_V^2\overline{\lambda}_W}{\kappa_0}+1\right) r_0^2 < 1;
	\end{align*}

	\item with probability one, there exists $B_0<\infty$ such that $\sup_{1\le b\le B_0}\norm{\hat\theta_b-\theta^*}_2\le r_0$
    and $\sup_{b\geq B_0+1}\max\{\Delta_b^V,\Delta_b^U,\Delta_b^W\}\leq \epsilon_0$.
\end{enumerate}
\end{customas}

\setcounter{theorem}{1}
\begin{theorem}[Strong consistency of $\hat{\theta}_b$ with finite $n_b$]\label{thm:tail_entry_as}
Suppose Assumptions~\ref{asum:stationarity}--\ref{asum:moment-dominance}, \ref{as:weight-new}, \ref{as:cvx}, \ref{as:Lip_nabla_global}, \ref{as:full_rank_V}, \ref{as:global-contraction}, and \ref{as:regularProblem_asVersion} hold.
Then $\hat\theta_b\cin{a.s.}\theta^*$, $\hat{V}_b\cin{a.s.} V(\theta^*)$ and $\hat{U}_b\cin{a.s.} 0_q$.
\end{theorem}

The proof of Theorem \ref{thm:tail_entry_as}
can be found in \cref{sec:proof-tail_entry_as}.
We conclude this subsection with
an example to demonstrate how Assumptions~\ref{as:full_rank_V}--\ref{as:regularProblem} in Theorem~\ref{thm:tail_entry} can be verified.
\begin{example}[Linear moment models] \label{eg:linear-models}
Suppose Assumptions~\ref{asum:stationarity}--\ref{asum:moment-dominance}, \ref{as:weight-new}, \ref{as:cvx} and \ref{as:Lip_nabla_global} hold, and let $g(\theta,x)=A(x)\theta-c(x)$ so that $\nabla g(\theta,x)=A(x)$,
$V_i(\theta)=n_0^{-1}\sum_{j=1}^{n_0} A(x_{i,j})$,
$V(\theta) := A :=\E\{A(x_{i,j})\}$, and $m(\theta)=A\theta-c$ with $c:=\E\{c(x_{i,j})\}$.
Assume $\norm{A}_2=:\overline{\sigma}_V<\infty$ and $\lambda_{\min}(A^\top W A)=:\kappa_0>0$, then Assumptions~\ref{as:Lip_nabla_global} and \ref{as:full_rank_V} hold with $L=0$.
Since $m(\theta^*)=A\theta^*-c=0_q$ such that $A\theta^*=c$, we have
$T(\theta)=\theta-(A^\top W A)^{-1}A^\top W(A\theta-A\theta^*)=\theta^*$ and hence Assumption~\ref{as:global-contraction} holds with $\rho=0$.
We discuss Assumption~\ref{as:regularProblem} in Theorem~\ref{thm:tail_entry}:
\begin{itemize}
	\item[(a)] $\delta_{\times 0}(\epsilon_\gamma,r_0)=\delta_{\times 0}(\epsilon_\gamma)$,
	$C_0'(\epsilon_\gamma,r_0)=C_0'(\epsilon_\gamma)$, and
	$C_0''(\epsilon_\gamma,r_0)=C_0''(\epsilon_\gamma)$ are free of $r_0$, and $\delta_{\times 0}(\epsilon_\gamma), C_0'(\epsilon_\gamma), C_0''(\epsilon_\gamma), \Phi_{\epsilon_\gamma}(r_0)\downarrow 0$ as $\epsilon_\gamma\downarrow 0$.
	\item[(b)] $K_{2,\epsilon}(r_0)=K_{3,\epsilon}(r_0)=0$, while $K_{0,\epsilon}(r_0)=K_{0,\epsilon}$ and $K_{1,\epsilon}(r_0)=K_{1,\epsilon}$ are free of $r_0$ and $o(1)$ as $\epsilon\downarrow 0$,
	hence the condition simplifies to $K_{0,\epsilon_\gamma}\leq \gamma (1-K_{1,\epsilon_\gamma})$.
	So the admissible $\epsilon_\gamma \lesssim c_{\mathrm{lin}}\gamma$ for some constant $c_{\mathrm{lin}}=c_{\mathrm{lin}}(\overline{\lambda}_W,\overline{\sigma}_V,\kappa_0)>0$.
	\item[(c)] Hence, if $\sup_{t\ge b}\max\{\Delta_t^V,\Delta_t^U,\Delta_t^W\}=O_p(a_b)$, then item (c) is implied by $a_{b_0}\lesssim \gamma$ and thus $b_0 \asymp \gamma^{-2}$ when $a_b\asymp b^{-1/2}$ and $b_0 \asymp \gamma^{-2}\log(\gamma^{-1})$ when $a_b\asymp \sqrt{\log b/b}$.
	Thus $r_0$ mainly affects the entrance part of condition~(c), whereas the empirical tail requirement is driven by $\gamma$.
\end{itemize}
If Assumption \ref{as:regularProblem}(a)--(c) hold, then $\hat\theta_b\cin{p}\theta^*$, $\hat{V}_b \cin{p} A$, $\hat{U}_b \cin{p} 0_q$.
If $\Delta_b^W$, $\Delta_b^V$, $\Delta_b^U\cin{a.s.}0$, and Theorem~\ref{thm:tail_entry_as}(a)--(c) hold instead, then the same conclusions hold almost surely.
\end{example}
Example~\ref{eg:linear-models} covers the least squares and instrumental variables special cases in the main paper, the linear over identification experiments, the anomaly detection models built from linear instrumental variables moments, and the forecasting equations in the stochastic volatility application.
We do not include the quantile regression example here because its smoothed score depends on a bandwidth and requires separate verification.

\section{Proof of theorems under diverging batch size} \label{sec:proof-of-theorems}

\subsection{\texorpdfstring{Proof of Theorem \ref{thm:consistency}}{
		Proof of Theorem 1}} \label{sec:proof-consistency}

Theorem 3.1 in \citet{hall2005gmm} states that $\hat\theta_1 \cin{p} \theta^*$ under Assumptions \ref{asum:stationarity}--\ref{asum:moment-global} and \ref{asum:weight}--\ref{asum:moment-dominance}.
We do not separate the remaining proof by cases (a) and (b) because they only differ by the use of Lemma \ref{lem:convergence-uv}(a) and (b).

To prove by induction, we first consider $b=2$.
Recall that $\hat{W} \cin{p} W$ under Assumption \ref{asum:weight}, and
$\hat{V}_2 \cin{p} V$ and $\hat{U}_2 \cin{p} U = 0$ by Lemma \ref{lem:convergence-uv}.
Applying Slutsky’s theorem gives
\begin{align*}
\hat\theta_2
\!= \hat\theta_1 -(\hat{V}_2^\T \hat{W} \hat{V}_2)^{-1} \hat{V}_2^\T \hat{W} \hat{U}_2
\cin{p} \theta^* -(V^\T W V)^{-1} V^\T W U
= \theta^*.
\end{align*}
Now, suppose that $\norm{\hat\theta_i -\theta^*}_2 = o_p(1)$ and $\hat\theta_i \in B_\epsilon$ for $i = 1, \ldots, b-1$.
Then, we have $\hat{V}_b \cin{p} V$ and $\hat{U}_b \cin{p} U = 0$ by Lemma \ref{lem:convergence-uv} again.
By the same argument,
\begin{align*}
\hat\theta_b
\!= \hat\theta_{b-1} -(\hat{V}_b^\T \hat{W} \hat{V}_b)^{-1} \hat{V}_b^\T \hat{W} \hat{U}_b
\cin{p} \theta^* -(V^\T W V)^{-1} V^\T W U
= \theta^*.
\end{align*}

In the above argument, we assume that $\hat\theta_i \in B_\epsilon$ for $i=1,\cdots,b-1$ to simplify the proof.
Technically speaking, we can partition the probability space into $\Omega_{b-1}=\cup_{i=1}^{b-1} \{\hat\theta_i\notin B_\epsilon\} $ and its complement $\Omega_{b-1}^c$. On $\Omega_{b-1}^c$, all the above discussion can go through under the assumptions stated for a local neighborhood $B_\epsilon$.
Hence all statements that hold in probability are still valid since $P(\Omega_{b-1})\le \sum_{i=1}^{b-1} P(\hat\theta_i \notin B_\epsilon)\rightarrow 0$;
see \citet[Section 3.5]{newey1994asymptotic} for an example with $\hat\theta_1$ only.
In the following proofs, we shall not repeat this argument but instead assume $\hat\theta_{i} \in B_\epsilon$ for $i=1,\cdots,b-1$ to simplify the presentation.
\qed

\subsection{\texorpdfstring{Proof of Theorem \ref{thm:normality-par}}{Proof of Theorem 3}} \label{sec:proof-normality-par}

Recall that Assumptions \ref{asum:stationarity}--\ref{asum:lrv} are identical to Assumptions 3.1--3.5 and 3.7--3.11 in \citet{hall2005gmm},
and both cases (a) and (b) in Assumptions \ref{asum:gradient-lipschitz} and \ref{asum:gradient-uniform}(ii) imply Assumptions 3.12 and 3.13 in \citet{hall2005gmm}, respectively.
By Theorem 3.2 in \citet{hall2005gmm}, we have $\sqrt{N_1} (\hat\theta_1 -\theta^*) \cin{d} \Normal(0, H \Sigma H^\T)$.
However, an asymptotically efficient $\hat\theta_1$ is not necessary in the remaining proof.
Any $\sqrt{N_1}$-consistent $\hat\theta_1$ will suffice.

(Case a). To prove by induction, we first consider $b=2$.
By the mean value theorem,
\begin{align*}
	\bar{g}_2(\hat\theta_2) = \bar{g}_2(\theta^*) +\nabla \bar{g}_2(\bar\theta_2) (\hat\theta_2 -\theta^*),
\end{align*}
where $\bar\theta_b$ is a point between $\hat\theta_b$ and $\theta^*$.
Multiplying by $\sqrt{N_2} \hat{V}_2^\T \hat{W}$ and rearranging,
\begin{align}
	\sqrt{N_2} (\hat\theta_2 -\theta^*)
	={}& \{\hat{V}_2^\T \hat{W} \nabla \bar{g}_2(\bar\theta_2)\}^{-1} \hat{V}_2^\T \hat{W} \sqrt{N_2} \bar{g}_2(\hat\theta_2) -\{\hat{V}_2^\T \hat{W} \nabla \bar{g}_2(\bar\theta_2)\}^{-1} \hat{V}_2^\T \hat{W} \sqrt{N_2} \bar{g}_2(\theta^*) \nonumber \\
	={}& J_{2,1} +J_{2,2}. \label{eq:decomp-normality-2}
\end{align}
Using \eqref{eq:moment-error-a} and the condition that $\hat{V}_2^\T \hat{W} \{ \hat{U}_2 +\hat{V}_2 (\hat\theta_2 -\hat\theta_1) \} = 0$, the first term simplifies to
\begin{align*}
		J_{2,1}
		={}& \sqrt{N_2} \{\hat{V}_2^\T \hat{W} \nabla \bar{g}_2(\bar\theta_2)\}^{-1} \hat{V}_2^\T \hat{W} \left\{ \hat{U}_2 +\hat{V}_2 (\hat\theta_2 -\hat\theta_1) \right. \\
		&\left. \qquad\qquad\qquad\qquad\qquad\qquad\quad
		+O_p\bigg( \frac{\sqrt{n_1}+\sqrt{n_2}}{N_2} \norm{\hat\theta_1-\hat\theta_2}_2 +\norm{\hat\theta_1-\hat\theta_2}_2^2 \bigg) \right\} \\
		={}& \{\hat{V}_2^\T \hat{W} \nabla \bar{g}_2(\bar\theta_2)\}^{-1} \hat{V}_2^\T \hat{W} O_p\bigg( \frac{\sqrt{n_1}+\sqrt{n_2}}{\sqrt{N_2}} \norm{\hat\theta_1-\hat\theta_2}_2 \\
		&\qquad\qquad\qquad\qquad\qquad\qquad\quad
		+\sqrt{N_2} \norm{\hat\theta_1-\hat\theta_2}_2^2 \bigg).
\end{align*}
Check that
\begin{equation*}
	\hat\theta_2 = -(\hat{V}_2^\T \hat{W} \hat{V}_2)^{-1} \hat{V}_2^\T \hat{W} \hat{U}_2
	= \hat\theta_1 -(\hat{V}_2^\T \hat{W} \hat{V}_2)^{-1} \hat{V}_2^\T \hat{W} \frac{1}{N_2} \left\{ G(\hat\theta_1; D_1) +G(\hat\theta_1; D_2) \right\}.
\end{equation*}
Using the same trick in \eqref{eq:taylor-error-2nd-a} and Assumption \ref{asum:gradient-regularity}(c), we have
\begin{align}
	G(\hat\theta_1; D_1)
	&= G(\theta^*; D_1) +\nabla G(\theta^*; D_1) (\hat\theta_1 -\theta^*) +O_p\left( \sqrt{n_1} \norm{\hat\theta_1 -\theta^*}_2 +n_1 \norm{\hat\theta_1 -\theta^*}_2^2 \right) \nonumber \\
	&= G(\theta^*; D_1) +n_1 \E\{\nabla g(\theta^*, x_{i,j})\} (\hat\theta_1 -\theta^*) +O_p\left( \sqrt{n_1} \norm{\hat\theta_1 -\theta^*}_2 +n_1 \norm{\hat\theta_1 -\theta^*}_2^2 \right) \nonumber \\
	&= G(\theta^*; D_1) +O_p\left( n_1 \norm{\hat\theta_1 -\theta^*}_2 \right). \label{eq:taylor-error-1st}
\end{align}
Therefore, by the central limit theorem for $\bar{g}_2(\theta^*)$, and
$\norm{\hat\theta_1 -\theta^*}_2 = O_p(N_1^{-1/2})$,
\begin{align}
	\hat\theta_2 -\hat\theta_1
	&= -(\hat{V}_2^\T \hat{W} \hat{V}_2)^{-1} \hat{V}_2^\T \hat{W} \frac{1}{N_2} \left\{ G(\hat\theta_1; D_1) +G(\hat\theta_1; D_2) \right\} \nonumber \\
	&= -(\hat{V}_2^\T \hat{W} \hat{V}_2)^{-1} \hat{V}_2^\T \hat{W} \bar{g}_2(\theta^*) +O_p\left(\norm{\hat\theta_1 -\theta^*}_2\right) \nonumber \\
	&= O_p(N_2^{-1/2} +N_1^{-1/2})
	= O_p(N_1^{-1/2}). \label{eq:difference-par-2}
\end{align}
Using the same procedure that leads to \eqref{eq:convergence-v-a}, we have $\nabla \bar{g}_2(\bar\theta_2) \cin{p} V$.
Moreover, recall that $\hat{W} \cin{p} W$ under Assumption \ref{asum:weight} and
$\hat{V}_2 \cin{p} V$ by Lemma \ref{lem:convergence-uv}(a).
Therefore, $J_{2,1} = o_p(1)$ if and only if $\sqrt{N_2} = o(N_1) \Leftrightarrow n_2 = o(N_1^2) = o(n_1^2)$ by $N_1=n_1$.
For $J_{2,2}$, the central limit theorem for $\bar{g}_2(\theta^*)$ and Slutsky's theorem gives
$J_{2,2} \cin{d} \Normal(0, H \Sigma H^\T)$.
Combining the results with Slutsky's theorem, we have $\sqrt{N_2} (\hat\theta_2 -\theta^*) \cin{d} \Normal(0, H \Sigma H^\T)$.
Now, suppose that
$\norm{\hat\theta_i -\theta^*}_2 = O_p(N_i^{-1/2})$
for $i = 1, \ldots, b-1$.
Using the same procedure that leads to \eqref{eq:decomp-normality-2}, we have
\begin{align}
	\sqrt{N_b} (\hat\theta_b -\theta^*)
	={}& \{\hat{V}_b^\T \hat{W} \nabla \bar{g}_b(\bar\theta_b)\}^{-1} \hat{V}_b^\T \hat{W} \sqrt{N_b} \bar{g}_b(\hat\theta_b)  -\{\hat{V}_b^\T \hat{W} \nabla \bar{g}_b(\bar\theta_b)\}^{-1} \hat{V}_b^\T \hat{W} \sqrt{N_b} \bar{g}_b(\theta^*) \nonumber \\
	={}& J_{b,1} +J_{b,2}. \label{eq:decomp-normality-b}
\end{align}
By \eqref{eq:moment-error-a} and the condition that $\hat{V}_b^\T \hat{W} \{ \hat{U}_b +\hat{V}_b (\hat\theta_b -\hat\theta_{b-1}) \} = 0$,
\begin{align*}
	J_{b,1}
	={}& \{\hat{V}_b^\T \hat{W} \nabla \bar{g}_b(\bar\theta_b)\}^{-1} \hat{V}_b^\T \hat{W}
	\bigg\{ O_p\left( \frac{\sqrt{n_b}}{\sqrt{N_b}} \norm{\hat\theta_{b-1}-\hat\theta_b}_2 +\sum_{i=1}^{b-1} \frac{\sqrt{n_i}}{\sqrt{N_b}} \norm{\hat\theta_i-\hat\theta_b}_2 \right) \\
	& \qquad\qquad\qquad\qquad\qquad\qquad
	+O_p\left( \frac{n_b}{\sqrt{N_b}} \norm{\hat\theta_{b-1}-\hat\theta_b}_2^2 +\sum_{i=1}^{b-1} \frac{n_i}{\sqrt{N_b}} \norm{\hat\theta_i-\hat\theta_b}_2^2 \right) \bigg\}.
\end{align*}
We try to bound $\norm{\hat\theta_{b-1}-\hat\theta_b}_2$ first.
Recall that
\begin{align*}
	\hat{U}_b
	={}& \frac{1}{N_b} \left\{ N_{b-1} U_{b-1} +G(\hat\theta_{b-1}; D_b) \right\} \\
	={}& \frac{1}{N_b} \left\{ \sum_{i=1}^{b-1} G(\hat\theta_i; D_i) +G(\hat\theta_{b-1}; D_b) \right\}
	+\frac{1}{N_b} \sum_{i=1}^{b-1} \nabla G(\hat\theta_i; D_i)(\hat\theta_{b-1} -\hat\theta_i).
\end{align*}
Using the trick in \eqref{eq:taylor-error-2nd-a} and \eqref{eq:taylor-error-1st},
\begin{align*}
	\nabla G(\hat\theta_i; D_i)
	&= \nabla G(\theta^*; D_i) +\nabla G(\hat\theta_i; D_i) -\nabla G(\theta^*; D_i) \\
	&= n_i \E\{\nabla g(\theta^*, x_{i,j})\} +O_p\left(n_i^{1/2} +\norm{\hat\theta_i-\theta^*}_2\right)
	= O_p(n_i).
\end{align*}
By the triangle inequality and $\norm{\hat\theta_i -\theta^*}_2 = O_p(N_i^{-1/2})$ for $i = 1, \ldots, b-1$,
\begin{align*}
	& \norm{\hat{U}_b -\frac{1}{N_b} \left\{ \sum_{i=1}^{b-1} G(\hat\theta_i; D_i) +G(\hat\theta_{b-1}; D_b) \right\}}_2 \\
	={}& \norm{\frac{1}{N_b} \sum_{i=1}^{b-2} \nabla G(\hat\theta_i; D_i) (\hat\theta_{b-1} -\hat\theta_i)}_2 \\
	\le{}& \norm{\frac{1}{N_b} \sum_{i=1}^{b-2} \nabla G(\hat\theta_i; D_i) (\hat\theta_i -\theta^*)}_2
	+\norm{\frac{1}{N_b} \sum_{i=1}^{b-2} \nabla G(\hat\theta_i; D_i) (\hat\theta_{b-1} -\theta^*)}_2 \\
	={}& O_p\left( \sum_{i=1}^{b-2} \frac{n_i}{N_b} \norm{\hat\theta_i -\theta^*}_2 \right)
	= O_p\left( \sum_{i=1}^{b-2} \frac{n_i}{N_b \sqrt{N_i}} \right).
\end{align*}
Applying the trick in \eqref{eq:taylor-error-1st} again and the central limit theorem for $\bar{g}_b(\theta^*)$,
\begin{align*}
	\hat\theta_b -\hat\theta_{b-1}
	&= -(\hat{V}_b^\T \hat{W} \hat{V}_b)^{-1} \hat{V}_b^\T \hat{W} \hat{U}_b \\
	&= -(\hat{V}_b^\T \hat{W} \hat{V}_b)^{-1} \hat{V}_b^\T \hat{W} \frac{1}{N_b} \left\{ \sum_{i=1}^{b-1} G(\hat\theta_i; D_i) +G(\hat\theta_{b-1}; D_b) \right\}
	+O_p\left( \sum_{i=1}^{b-2} \frac{n_i}{N_b \sqrt{N_i}} \right) \\
	&= -(\hat{V}_b^\T \hat{W} \hat{V}_b)^{-1} \hat{V}_b^\T \hat{W} \bar{g}_b(\theta^*)
	+O_p\left( \frac{n_b}{N_b \sqrt{N_{b-1}}} +\sum_{i=1}^{b-1} \frac{n_i}{N_b \sqrt{N_i}} \right) \\
	&= O_p\left( \frac{1}{\sqrt{N_b}} +\frac{n_b}{N_b \sqrt{N_{b-1}}} +\sum_{i=1}^{b-1} \frac{n_i}{N_b \sqrt{N_i}} \right).
\end{align*}
By Lemma \ref{lem:bound-n},
\begin{align*}
\frac{1}{\sqrt{N_b}} +\frac{n_b}{N_b \sqrt{N_{b-1}}} +\sum_{i=1}^{b-1} \frac{n_i}{N_b \sqrt{N_i}}
\le \frac{1}{\sqrt{N_b}} +\frac{n_b}{N_b \sqrt{N_{b-1}}} +\frac{2\sqrt{N_{b-1}}}{N_b}
\le \frac{3}{\sqrt{N_b}} +\frac{n_b}{N_b \sqrt{N_{b-1}}}.
\end{align*}
Therefore,
\begin{equation} \label{eq:difference-par-b}
	\norm{\hat\theta_{b-1}-\hat\theta_b}_2
	= {O_p\left( \frac{1}{\sqrt{N_b}}+\frac{n_b}{N_b \sqrt{N_{b-1}}} \right)}
\end{equation}
Using $\norm{\hat\theta_i -\theta^*}_2 = O_p(N_i^{-1/2})$ for $i = 1, \ldots, b-1$, a direct consequence is that
\begin{align*}
\norm{\hat\theta_i-\hat\theta_b}_2
\le \norm{\hat\theta_i-\theta^*}_2 +\norm{\theta^*-\hat\theta_{b-1}}_2 +\norm{\hat\theta_{b-1}-\hat\theta_b}_2
= O_p(N_i^{-1/2}).
\end{align*}
We are ready to bound $J_{b,1}$.
By the triangle inequality, \eqref{eq:difference-par-b} and Lemma \ref{lem:bound-n},
\begin{align*}
	\frac{\sqrt{n_b}}{\sqrt{N_b}} \norm{\hat\theta_{b-1}-\hat\theta_b}_2 +\sum_{i=1}^{b-1} \frac{\sqrt{n_i}}{\sqrt{N_b}} \norm{\hat\theta_i-\hat\theta_b}_2
	&= O_p\left( \frac{1}{\sqrt{N_{b-1}}} +\sum_{i=1}^{b-1} \frac{\sqrt{n_i}}{\sqrt{N_i N_b}} \right) \\
	&= O_p\left\{ \frac{1}{\sqrt{N_{b-1}}} +\frac{1}{\sqrt{\min(n_1, \ldots, n_b)}} \right\}
	= o_p(1).
\end{align*}
Similarly, we can apply the triangle inequality to obtain
\begin{align*}
	& \frac{n_b}{\sqrt{N_b}} \norm{\hat\theta_{b-1}-\hat\theta_b}_2^2
	+\sum_{i=1}^{b-1} \frac{n_i}{\sqrt{N_b}} \norm{\hat\theta_i-\hat\theta_b}_2^2 \\
	\le&{} \sqrt{N_b} \norm{\hat\theta_{b-1}-\hat\theta_b}_2^2
	+2 \norm{\hat\theta_{b-1}-\hat\theta_b}_2 \sum_{i=1}^{b-1} \frac{n_i}{\sqrt{N_b}} \norm{\hat\theta_i-\hat\theta_{b-1}}_2
	+\sum_{i=1}^{b-1} \frac{n_i}{\sqrt{N_b}} \norm{\hat\theta_i-\hat\theta_{b-1}}_2^2.
\end{align*}
By \eqref{eq:difference-par-b} and Lemma \ref{lem:bound-n}, we have
\begin{align*}
\norm{\hat\theta_{b-1}-\hat\theta_b}_2 \sum_{i=1}^{b-1} \frac{n_i}{\sqrt{N_b}} \norm{\hat\theta_i-\hat\theta_{b-1}}_2
= o_p(1) O_p\left( \sum_{i=1}^{b-1} \frac{n_i}{\sqrt{N_i N_b}} \right)
= o_p(1) O_p\left( \frac{2 \sqrt{N_b}}{\sqrt{N_b}} \right)
= o_p(1)
\end{align*}
and
\begin{align*}
\sum_{i=1}^{b-1} \frac{n_i}{\sqrt{N_b}} \norm{\hat\theta_i-\hat\theta_{b-1}}_2^2
= O_p\left( \sum_{i=1}^{b-1} \frac{n_i}{N_i \sqrt{N_b}} \right)
= O_p\left( \frac{1 +\log N_b -\log n_1}{\sqrt{N_b}} \right)
= o_p(1).
\end{align*}
The conclusion then follows by the same argument for $b=2$ as we can verify that
\begin{align*}
J_{b,1} = O_p\left( \sqrt{N_b} \norm{\hat\theta_{b-1}-\hat\theta_b}_2^2 \right)
= O_p\left( \frac{n_b^2}{N_b^{3/2} N_{b-1}} \right)
= o_p(1)
\end{align*}
when $n_b = o(N_{b-1}^2)$.

(Case b). The proof is very similar to case (a) so we only highlight the difference.
When $b=2$, we still have \eqref{eq:decomp-normality-2}.
Applying \eqref{eq:moment-error-b} and the condition that $\hat{V}_2^\T \hat{W} \{ \hat{U}_2 +\hat{V}_2 (\hat\theta_2 -\hat\theta_1) \} = 0$ to $J_{2,1}$, we have
\begin{align*}
J_{2,1} = \{\hat{V}_2^\T \hat{W} \nabla \bar{g}_2(\bar\theta_2)\}^{-1} \hat{V}_2^\T \hat{W} O_p\left( \sqrt{N_2} \norm{\hat\theta_1-\hat\theta_2}_2^2 \right).
\end{align*}
The conclusion in \eqref{eq:difference-par-2} remain unchanged.
However, the argument for \eqref{eq:taylor-error-1st} is slightly different.
By \eqref{eq:taylor-error-2nd-b}, Assumptions \ref{asum:gradient-regularity}(c) and \ref{asum:gradient-uniform}(b)(i), we have
\begin{align*}
	\frac{1}{N_2} \sum_{i=1}^2 G(\hat\theta_1; D_i)
	&= \bar{g}_2(\theta^*) +\nabla \bar{g}_2(\theta^*) (\hat\theta_1 -\theta^*) +O_p\left( \norm{\hat\theta_1 -\theta^*}_2^2 \right) \\
	&= \bar{g}_2(\theta^*) +\E\{\nabla g(\theta^*, x_{i,j})\} (\hat\theta_1 -\theta^*)
	+O_p\left( \norm{\hat\theta_1 -\theta^*}_2 +\norm{\hat\theta_1 -\theta^*}_2^2 \right) \\
	&= \bar{g}_2(\theta^*) +O_p\left( \norm{\hat\theta_1 -\theta^*}_2 \right).
\end{align*}
Using the same procedure that leads to \eqref{eq:convergence-v-b}, we have $\nabla \bar{g}_2(\bar\theta_2) \cin{p} V$.
The remaining arguments are the same.
For $b>2$, check that
\begin{align*}
J_{b,1}
=\{\hat{V}_b^\T \hat{W} \nabla \bar{g}_b(\bar\theta_b)\}^{-1} \hat{V}_b^\T \hat{W}
O_p\left( \frac{n_b}{\sqrt{N_b}} \norm{\hat\theta_{b-1}-\hat\theta_b}_2^2 +\sum_{i=1}^{b-1} \frac{n_i}{\sqrt{N_b}} \norm{\hat\theta_i-\hat\theta_b}_2^2 \right).
\end{align*}
The other arguments are same as case (a) or those in the above for $b=2$.
\hfill$\qed$

\subsection{\texorpdfstring{Proof of Theorem \ref{thm:normality-moment}}{
		Proof of Theorem 4}} \label{sec:proof-normality-moment}

The proof of Theorem \ref{thm:normality-moment} is based on the results in the proof of Theorem \ref{thm:normality-par}.
Therefore, we do not separate the proof by cases (a) and (b) but refer to the relevant results instead.
By Theorem 3.3 in \citet{hall2005gmm}, we have $\hat{W}^{1/2} \sqrt{N_1} U_1 = \hat{W}^{1/2} \sqrt{N_1} \bar{g}_1(\hat\theta_1) \cin{d} \Normal(0, R W^{1/2} \Sigma (W^{1/2})^\T R^\T)$,
where the square root of a positive semi-definite matrix $A$ is defined as the unique positive semi-definite matrix $B$ such that $BB = B^\T B = A$.
For $b=2$, by the trick in \eqref{eq:taylor-error-2nd-a} and the result in \eqref{eq:difference-par-2},
\begin{align*}
	\hat{W}^{1/2} \sqrt{N_2} U_2
	&= \hat{W}^{1/2} \sqrt{N_2} \bar{g}_2(\hat\theta_2) +O_p\left( \frac{n_1}{\sqrt{N_2}} \norm{\hat\theta_1-\hat\theta_2}_2^2 \right) \\
	&= \hat{W}^{1/2} \sqrt{N_2} \bar{g}_2(\hat\theta_2) +o_p(1).
\end{align*}
It follows from the mean value theorem and \eqref{eq:decomp-normality-2} that
\begin{align*}
	& \hat{W}^{1/2} \sqrt{N_2} U_2 \\
	={}& \hat{W}^{1/2} \sqrt{N_2} \bar{g}_2(\theta^*) +\hat{W}^{1/2} \nabla \bar{g}_2(\bar\theta_2) \sqrt{N_2} (\hat\theta_2 -\theta^*) +o_p(1) \\
	={}& \left[ I_q -\hat{W}^{1/2} \nabla \bar{g}_2(\bar\theta_2) \{\hat{V}_2^\T \hat{W} \nabla \bar{g}_2(\bar\theta_2)\}^{-1} \hat{V}_2^\T (\hat{W}^{1/2})^\T \right] \hat{W}^{1/2} \sqrt{N_2} \bar{g}_2(\theta^*) +o_p(1).
\end{align*}
Recall that $\hat{W} \cin{p} W$, $\hat{V}_2 \cin{p} V$ and $\nabla \bar{g}_2(\bar\theta_2) \cin{p} V$ in the proof of Theorem \ref{thm:normality-par}.
Therefore, the central limit theorem for $\bar{g}_2(\theta^*)$ and Slutsky's theorem gives
$\hat{W}^{1/2} \sqrt{N_2} U_2 \cin{d} \Normal(0, R W^{1/2} \Sigma (W^{1/2})^\T R^\T)$.
For $b > 2$, by the trick in \eqref{eq:taylor-error-2nd-a} and the result in \eqref{eq:difference-par-b},
\begin{align*}
	\hat{W}^{1/2} \sqrt{N_b} U_b
	&= \hat{W}^{1/2} \sqrt{N_b} \bar{g}_b(\hat\theta_b) +O_p\left( \sum_{i=1}^{b-1} \frac{n_i}{\sqrt{N_b}} \norm{\hat\theta_i-\hat\theta_b}_2^2 \right)
	= \hat{W}^{1/2} \sqrt{N_b} \bar{g}_b(\hat\theta_b) +o_p(1).
\end{align*}
By the mean value theorem and \eqref{eq:decomp-normality-b},
\begin{align*}
	& \hat{W}^{1/2} \sqrt{N_b} U_b \\
	={}& \hat{W}^{1/2} \sqrt{N_b} \bar{g}_b(\theta^*) +\hat{W}^{1/2} \nabla \bar{g}_b(\bar\theta_b) \sqrt{N_b} (\hat\theta_b -\theta^*) +o_p(1) \\
	={}& \left[ I_q -\hat{W}^{1/2} \nabla \bar{g}_b(\bar\theta_b) \{\hat{V}_b^\T \hat{W} \nabla \bar{g}_b(\bar\theta_b)\}^{-1} \hat{V}_b^\T (\hat{W}^{1/2})^\T \right] \hat{W}^{1/2} \sqrt{N_b} \bar{g}_b(\theta^*) +o_p(1).
\end{align*}
The conclusion then follows by the same argument for $b=2$.
When $\hat\Sigma$ is positive semi-definite and converges in probability to $\Sigma$ in Assumption \ref{asum:lrv},
the asymptotic distribution of
$N_b U_b^\T \hat\Sigma^{-1} U_b$
is a direct consequence of the asymptotic normality of
$\hat{W}^{1/2} \sqrt{N_b} U_b$ with $\hat{W} = \hat\Sigma^{-1}$;
see the proof of Theorem 5.1 in \citet{hall2005gmm}.

\subsection{\texorpdfstring{Proof of Proposition~\ref{prop:plugin-lrv}}{Proof of Proposition A.1}}
We drop the $K$ in the LRV estimator such that $\hat{\Sigma}_b=\hat{\Sigma}_b(K)$ for simplicity.
Let
\begin{equation*}
\hat\Sigma_b^*
= \frac{1}{N_b} \sum_{i=1}^{N_b} \sum_{j=1}^{N_b} K_{N_b}(i, j)(X_i ^*- \bar X_{N_b}^*)(X_j^* - \bar X_{N_b}^*)^\T, \qquad \text{where} \qquad \bar X_{N_b}^* = \frac{1}{N_b}\sum_{i=1}^{N_b} X_i^*.
\end{equation*}
Also let $X_{i,[j]}=g_{[j]}(\hat{\theta}_{\beta(i)},x_i)$ and $X_{i,[j]}^*=g_{[j]}(\theta^*,x_i)$ be the $j$th entry of $g(\hat{\theta}_{\beta(i)},x_i)$ and $g(\theta^*,x_i)$, respectively, for $j=1,\ldots,q$.
For any two sequences of scalar $\{Y_i\}_{i=1}^{N_b}$ and $\{Z_i\}_{i=1}^{{N_b}}$ ($n\in\Z^+$) denote
\begin{align*}
    \hat{\sigma}_b^2 (\{Y_i\}_{i=1}^{N_b}) &= \frac{1}{N_b}\sum_{i=1}^{N_b}\sum_{j=1}^{N_b} K_{N_b}(i,j)(Y_i-\bar{Y}_{N_b})(Y_j-\bar{Y}_{N_b}), \\
    \hat{\sigma}_b^2 (\{Y_i\}_{i=1}^{N_b},\{Z_i\}_{i=1}^{N_b}) &= \frac{1}{N_b}\sum_{i=1}^{N_b}\sum_{j=1}^{N_b} K_{N_b}(i,j)(Y_i-\bar{Y}_{N_b})(Z_j-\bar{Z}_{N_b}),
\end{align*}where $\bar{Y}_{N_b}=N_b^{-1}\sum_{i=1}^{N_b} Y_i$ and $\bar{Z}_{N_b}=N_b^{-1}\sum_{i=1}^{N_b}Z_i$.
Decompose
\begin{align*}
    \hat{\Sigma}_b - \Sigma =
    (\hat{\Sigma}_b^* - \Sigma) + (\hat{\Sigma}_b - \hat{\Sigma}_b^*),
\end{align*}and study the two terms separately.

First, by item (a) and Theorem 3.1(b) of \citet{rlrv}, we have $c^\T \hat{\Sigma}_b^* c \cin{p} c^\T \Sigma c$ for any $c\in\R^q$, the LRV of $\{c^\T X_i^*\}_{i\in\Z^+}$.
By symmetry, the $(h,k)$-entry $(\hat{\Sigma}_{b}^*)_{h,k}$ of $\hat{\Sigma}_b^*$
\begin{align*}
    (\hat{\Sigma}_{b}^*)_{h,k} = \frac{1}{2} \cbk{
        \hat{\sigma}_b^2 (\{X_{i,[h]}^*+X_{i,[k]}^*\}_{i=1}^{N_b})
        - \hat{\sigma}_b^2 (\{X_{i,[h]}^*\}_{i=1}^{N_b})
        - \hat{\sigma}_b^2 (\{X_{i,[k]}^*\}_{i=1}^{N_b})
    },
\end{align*}
see also equation (A.2) of \citet{supp-rlrv}.
By letting $c=e_h+e_k$,
where $e_i\in\R^q$ is the vector of zeros except 1 in the $i$th entry, we have the consistency of the LRV of $\{X_{i,[h]}^*+X_{i,[k]}^*\}_{i\in\Z^+}$, $\{X_{i,[h]}^*\}_{i\in\Z^+}$ and $\{X_{i,[k]}^*\}_{i\in\Z^+}$ as $N_b\to\infty$, and $(\hat{\Sigma}_{b}^*)_{h,k}\cin{p}(\Sigma)_{h,k}$ as $N_b\to\infty$, the $(h,k)$-entry of $\Sigma$.
Since $q$ is fixed and finite, we have that
as $N_b\to \infty$
\begin{align*}
    \norm{\hat{\Sigma}_b^*-\Sigma}_F\cin{p}0.
\end{align*}

Second, we are going to show $\norm{\hat{\Sigma}_b - \hat{\Sigma}_b^*}_F\cin{p} 0$.
Let $(\hat{\Sigma}_b)_{h,k}$ be the $(h,k)$-entry of $\hat{\Sigma}_b$ and write the $(h,k)$-entry of $\hat{\Sigma}_b - \hat{\Sigma}_b^*$ as
\begin{align*}
    (\hat{\Sigma}_b - \hat{\Sigma}_b^*)_{h,k} &=
    \hat{\sigma}_b^2 (\{X_{i,[h]}\}_{i=1}^{N_b},\{X_{i,[k]}\}_{i=1}^{N_b})
    - \hat{\sigma}_b^2 (\{X_{i,[h]}^*\}_{i=1}^{N_b},\{X_{i,[k]}^*\}_{i=1}^{N_b}) \\
    &= \hat{\sigma}_b^2 (\{X_{i,[h]}-X_{i,[h]}^*\}_{i=1}^{N_b},\{X_{i,[k]}^*\}_{i=1}^{N_b})
    + \hat{\sigma}_b^2 (\{X_{i,[h]}^*\}_{i=1}^{N_b},\{X_{i,[k]}-X_{i,[k]}^*\}_{i=1}^{N_b}) \\
    & \qquad + \hat{\sigma}_b^2 (\{X_{i,[h]}-X_{i,[h]}^*\}_{i=1}^{N_b},\{X_{i,[k]}-X_{i,[k]}^*\}_{i=1}^{N_b}).
\end{align*}We focus on the first term of $(\hat{\Sigma}_b - \hat{\Sigma}_b^*)_{h,k}$:
\begin{align*}
    & \hat{\sigma}_b^2 (\{X_{i,[h]}-X_{i,[h]}^*\}_{i=1}^{N_b},\{X_{i,[k]}^*\}_{i=1}^{N_b}) \\
    & \qquad \leq \cbk{\max_{1\leq i\leq N_b} \sum_{j=1}^{N_b} \left|K_{N_b}(i,j)\right|}
    \mbk{\frac{1}{N_b}\sum_{i=1}^{N_b}\cbk{(X_{i,[h]}-X_{i,[h]}^*)-(\bar{X}_{N_b,h}-\bar{X}_{N_b,h}^*)}^2}^{1/2}\\
    & \qquad \qquad \times \cbk{\frac{1}{N_b}\sum_{i=1}^{N_b}\sbk{X_{i,[k]}^*-\bar{X}_{N_b,k}^*}^2}^{1/2},
\end{align*}where $\bar{X}_{N_b,k}=N_b^{-1}\sum_{i=1}^{N_b}X_{i,[k]}$ for $k=1,\ldots,q$ and it is similar for $\bar{X}_{N_b,k}^*$.
By the condition (a) in the proposition and Assumption~\ref{asum:lrv}, we have the law of large number that
\begin{align*}
    \frac{1}{N_b}\sum_{i=1}^{N_b}\sbk{X_{i,[k]}^*-\bar{X}_{N_b,[k]}^*}^2 \cin{p} \E\cbk{g_{[k]}^2(\theta,x_t')} < \infty.
\end{align*}
Besides,
\begin{align*}
    &\frac{1}{N_b}\sum_{i=1}^{N_b}\cbk{(X_{i,[h]}-X_{i,[h]}^*)-(\bar{X}_{N_b,[h]}-\bar{X}_{N_b,[h]}^*)}^2 \\
    &\qquad= \frac{1}{N_b}\sum_{i=1}^{N_b}\sbk{X_{i,[h]}-X_{i,[h]}^*}^2 - \sbk{\bar{X}_{N_b,[h]}-\bar{X}_{N_b,[h]}^*}^2 \\
    &\qquad \leq \frac{1}{N_b}\sum_{i=1}^{N_b}\sbk{X_{i,[h]}-X_{i,[h]}^*}^2
    = \frac{1}{N_b}\sum_{i=1}^{N_b}\cbk{g_{[h]}\sbk{\hat{\theta}_{\beta(i)},x_i} - g_{[h]}\sbk{\theta^*,x_i}}^2 \\
    &\qquad \leq \frac{L^2}{N_b}\sum_{i=1}^{N_b} \norm{\hat{\theta}_{\beta(i)} - \theta^*}_2^2 = \frac{L^2}{N_b}\sum_{i=1}^{b} n_i \norm{\hat{\theta}_i - \theta^*}_2^2,
\end{align*}
where the last inequality holds owing to the
conditions
(b)--(c) in the proposition.
Then by item (c) that $\norm{\hat{\theta}_i-\theta^*}_2=O_p(N_i^{-1/2})$, Lemma~\ref{lem:bound-n} that $\sum_{i=1}^b ({n_i}/{N_i}) \le 1 +\log ({N_b}/{n_1})$, \eqref{eq:supp-kernel} that
\begin{align*}
    \max_{1\leq i\leq N_b} \sum_{j=1}^{N_b} \left|K_{N_b}(i,j)\right|
    &= \max_{1\leq i\leq N_b} \sum_{j=1}^{N_b} \abs{1 - \frac{|i-j|^\lambda}{t_{N_b}^\lambda}} \I_{|i-j| \le s'_{\max(i,j)}} \\
    & \ \le \max_{1\leq i\leq N_b} \sum_{j=1}^{N_b} \abs{1-\frac{|i-j|^\lambda}{t_{N_b}^\lambda}} \I_{|i-j| \le \phi s_{N_b}} \\
    & \ \le 2 \sum_{k=0}^{\phi s_{N_b}} \abs{1-\frac{k^\lambda}{t_{N_b}^\lambda}}
    \le O(1) \cdot \int_{0}^{\phi s_{N_b}} \abs{1-\frac{k^\lambda}{t_{N_b}^\lambda}} \mathrm{d}k \\
    & \leq O\sbk{N_b^{\psi}+N_b^{\lambda(\psi-\xi)+\psi}},
\end{align*}see also equations (B.8)--(B.10) of \citet{supp-rlrv}, and condition (d) in the proposition,
we have
\begin{align*}
    & \cbk{\max_{1\leq i\leq N_b} \sum_{j=1}^{N_b} \left|K_{N_b}(i,j)\right|}
    \mbk{\frac{1}{N_b}\sum_{i=1}^{N_b}\cbk{(X_{i,[h]}-X_{i,[h]}^*)-(\bar{X}_{N_b,h}-\bar{X}_{N_b,h}^*)}^2}^{1/2} \\
    & \qquad = \cbk{\max_{1\leq i\leq N_b} \sum_{j=1}^{N_b} \left|K_{N_b}(i,j)\right|}
    \cbk{O_p\sbk{\frac{1}{N_b}\sum_{i=1}^{b}\frac{n_i}{N_i}}}^{1/2} \\
    & \qquad = O_p \cbk{
        \sbk{N_b^{\psi}+N_b^{\lambda(\psi-\xi)+\psi}} \cdot \frac{\sbk{1+\log(N_b/n_1)}^{1/2}}{N_b^{1/2}}
    } \\
    & \qquad = O_p\cbk{
        N_b^{\frac{1}{1+2\lambda}-\frac12} \sbk{1+\log\frac{N_b}{n_1}}^{1/2}
    }
    = o_p(1)
\end{align*}
since $\lambda \in \Z^+$ and $1/(1+2\lambda)-1/2<0$.
Therefore,
\begin{align*}
    \hat{\sigma}_b^2 (\{X_{i,[h]}-X_{i,[h]}^*\}_{i=1}^{N_b},\{X_{i,[k]}^*\}_{i=1}^{N_b}) = o_p(1),
\end{align*}and the same $o_p(1)$ holds for
$\hat{\sigma}_b^2 (\{X_{i,[h]}^*\}_{i=1}^{N_b},\{X_{i,[k]}-X_{i,[k]}^*\}_{i=1}^{N_b})$ and
$\hat{\sigma}_b^2 (\{X_{i,[h]}-X_{i,[h]}^*\}_{i=1}^{N_b},\{X_{i,[k]}-X_{i,[k]}^*\}_{i=1}^{N_b})$.
Hence, $(\hat{\Sigma}_b - \hat{\Sigma}_b^*)_{h,k}=o_p(1)$ and $\norm{\hat{\Sigma}_b - \hat{\Sigma}_b^*}_F\cin{p} 0$.
Combining with $\norm{\hat{\Sigma}_b^*-\Sigma}_F\cin{p}0$ gives $\norm{\hat{\Sigma}_b - \Sigma}_F\cin{p} 0$ as $N_b\to \infty$.

\section{Proof of theorems under fixed batch size} \label{sec:proof-of-theorems-fix}

\subsection{\texorpdfstring{Proof of Theorem~\ref{thm:tail_entry}}{Proof of Theorem 2}} \label{sec:proof-tail_entry}
For any fixed $\varepsilon,\eta>0$, choose $\gamma=\varepsilon/2$ and let $\epsilon_\gamma$ and $b_0=b_0(\gamma,\eta)$ be the corresponding quantities in the theorem.
Define the event
\begin{align*}
	\calE_{\eta,\gamma}:=\cbk{
		\sup_{1\le b\le b_0}\norm{\hat\theta_b-\theta^*}_2\le r_0,\
		\sup_{b\ge b_0+1} \max\{\Delta_b^V, \Delta_b^U, \Delta_b^W\} \le \epsilon_\gamma,
	}.
\end{align*}
By Assumption~\ref{as:regularProblem} item (c), we have $\pr(\calE_{\eta,\gamma})\ge 1-\eta$.
We first work on the event $\calE_{\eta,\gamma}$.
We claim that $\norm{\hat\theta_b-\theta^*}_2\le r_0$ for all $b\geq 1$.
It is true for $1\le b \le b_0$ by the definition of $\calE_{\eta,\gamma}$.
Then, on $\calE_{\eta,\gamma}$, by Lemma~\ref{lemma:one_step} and Assumption~\ref{as:regularProblem} item (a) that $\kappa_0-\delta_{\times 0}(\epsilon_\gamma,r_0)>0$ and $\Phi_{\epsilon_\gamma}(r_0)\leq r_0$, and $\sup_{b\ge b_0+1} \max\{\Delta_b^V, \Delta_b^U, \Delta_b^W\} \le \epsilon_\gamma$, all the items (a)--(e) in Lemma~\ref{lemma:one_step} hold for $b\geq b_0$, hence the result holds.

We now show convergence on $\calE_{\eta,\gamma}$.
Since $\norm{\hat\theta_b-\theta^*}_2\le r_0$ for all $b\ge b_0$, the sequence $(\norm{\hat\theta_b-\theta^*}_2)_{b\ge b_0}$ is bounded.
Let $\Delta_{\infty,\eta,\gamma} = \limsup_{b\to\infty} \norm{\hat\theta_b-\theta^*}_2 \leq r_0$ on $\calE_{\eta,\gamma}$.
Moreover, the Ces\`aro summation and the inequalities
\begin{align*}
	\norm{\hat\theta_i-\hat\theta_{b-1}}_2  \leq \norm{\hat\theta_i-\hat\theta^*}_2 + \norm{\hat\theta_{b-1}-\hat\theta^*}_2, \qquad
	\norm{\hat\theta_i-\hat\theta_{b-1}}_2^2  \leq 2\norm{\hat\theta_i-\hat\theta^*}_2^2 + 2\norm{\hat\theta_{b-1}-\hat\theta^*}_2^2,
\end{align*}imply
\begin{align*}
\limsup_{b\to\infty} P_b \le 2\Delta_{\infty,\eta,\gamma},
\qquad
\limsup_{b\to\infty} Q_b \le 4\Delta_{\infty,\eta,\gamma}^2
\qquad\text{on }\calE_{\eta,\gamma}.
\end{align*}
Recall \eqref{eq:GN_map} that we need to bound $\norm{\xi_b}_2 \le \norm{H(\hat\theta_{b-1})^{-1}-\hat{H}_b^{-1}}_2 \norm{h(\hat{\theta}_{b-1})}_2 + \norm{\hat{H}_b^{-1}}_2 \norm{h(\hat{\theta}_{b-1})-\hat{h}_b}_2$.
By \eqref{eq:deltaVU} in Lemma~\ref{lemma:bound_diff_VUhat},
\begin{align*}
\norm{\hat{V}_b-V(\hat\theta_{b-1})}_2 \le LP_b+\Delta_b^V \leq LP_b + \epsilon_\gamma, \qquad
\norm{\hat{U}_b-m(\hat\theta_{b-1})}_2 \le \frac{L}{2}Q_b+\Delta_b^U \leq \frac{L}{2}Q_b + \epsilon_\gamma
\end{align*}
on $\calE_\eta$ for all $b\geq b_0+1$.
And by \eqref{eq:hathm} in Lemma~\ref{lemma:bound_diff_VUhat},
\begin{align*}
	\norm{h(\hat{\theta}_{b-1})-\hat{h}_b}_2 & \leq \overline{\lambda}_W \overline{\sigma}_V (LP_b+\epsilon_\gamma)\norm{\hat{\theta}_{b-1}-\theta^*}_2
	+ \overline{\sigma}_V \epsilon_\gamma \sbk{LP_b+\epsilon_\gamma + \overline{\sigma}_V} \norm{\hat{\theta}_{b-1}-\theta^*}_2 \\
	& \quad + (\overline{\lambda}_W+\epsilon_\gamma) (LP_b+\epsilon_\gamma+\overline{\sigma}_V)\sbk{\frac{LQ_b}{2} + \epsilon_\gamma} := C_{b,\gamma}' \norm{\hat{\theta}_{b-1}-\theta^*}_2 + C_{b,\gamma}'',
\end{align*}
where $C_{b,\gamma}' = \overline{\sigma}_V\cbk{\overline{\lambda}_W(LP_b+\epsilon_\gamma) + \epsilon_\gamma (LP_b+\overline{\sigma}_V+\epsilon_\gamma)}$,
$C_{b,\gamma}'' = (\overline{\lambda}_W+\epsilon_\gamma)(LP_b + \overline{\sigma}_V + \epsilon_\gamma)(LQ_b/2+\epsilon_\gamma)$
on $\calE_{\eta,\gamma}$ for all $b\geq b_0+1$.
By \eqref{eq:hath} in Lemma~\ref{lemma:bound_diff_VUhat},
\begin{align*}
\norm{h(\hat{\theta}_{b-1})}_2 \leq \overline{\sigma}_V^2 \overline{\lambda}_W \norm{\hat{\theta}_{b-1}-\theta^*}_2,
\end{align*}on $\calE_{\eta,\gamma}$ for all $b\geq b_0+1$, and by Lemma~\ref{lemma:bound_diff_VUhat},
\begin{align*}
	& \norm{\hat H_b-H(\hat\theta_{b-1})}_2 = \norm{(\hat{V}_b-V(\hat\theta_{b-1}))^\top \hat W_b \hat{V}_b + V(\hat\theta_{b-1})^\top (\hat W_b-W)\hat{V}_b + V(\hat\theta_{b-1})^\top W (\hat{V}_b-V(\hat\theta_{b-1}))}_2 \\
	& \ \leq  (LP_b+\epsilon_\gamma) (\overline{\lambda}_W+\epsilon_\gamma) (LP_b+\epsilon_\gamma+\overline{\sigma}_V)
	+ \overline{\sigma}_V \epsilon_\gamma (LP_b+\epsilon_\gamma+\overline{\sigma}_V)
	+ \overline{\sigma}_V \overline{\lambda}_W (LP_b+\epsilon_\gamma) \\
	& \ = \overline{\lambda}_W (LP_b+\epsilon_\gamma)(LP_b+2\overline{\sigma}_V+\epsilon_\gamma) + \epsilon_\gamma (LP_b+\epsilon_\gamma+\overline{\sigma}_V)^2 := \delta_{\times b,\gamma}
\end{align*}on $\calE_{\eta,\gamma}$ for all $b\geq b_0+1$.
By Lemma~\ref{lemma:one_step} that $P_b\leq 2r_0$, we have $\delta_{\times b,\gamma} \leq \delta_{\times 0}(\epsilon_\gamma,r_0)$ in \eqref{eq:one-step-C-delta} and thus $\norm{\hat{H}_b^{-1}}_2 \leq (\kappa_0-\delta_{\times 0}(\epsilon_\gamma,r_0))^{-1}$ and
\begin{align*}
	\norm{H(\hat\theta_{b-1})^{-1}-\hat{H}_b^{-1}}_2 &
    \leq \norm{H(\hat\theta_{b-1})^{-1}}_2
    \norm{H(\hat\theta_{b-1})-\hat{H}_b}_2 \norm{\hat{H}_b^{-1}}_2
	\leq \frac{\delta_{\times b,\gamma}}{\kappa_0(\kappa_0-\delta_{\times 0}(\epsilon_\gamma,r_0))},
\end{align*}
on $\calE_\eta$ for all $b\ge b_0+1$.
Thus,
\begin{align*}
	\norm{\xi_b}_2 &\leq \norm{H(\hat\theta_{b-1})^{-1}-\hat{H}_b^{-1}}_2 \norm{h(\hat{\theta}_{b-1})}_2 + \norm{\hat{H}_b^{-1}}_2 \norm{h(\hat{\theta}_{b-1})-\hat{h}_b}_2 \\
	& \leq \frac{\delta_{\times b,\gamma}\overline{\sigma}_V^2 \overline{\lambda}_W}{\kappa_0(\kappa_0-\delta_{\times 0}(\epsilon_\gamma,r_0))} \norm{\hat{\theta}_{b-1}-\theta^*}_2
	+ \frac{1}{\kappa_0-\delta_{\times 0}(\epsilon_\gamma,r_0)} (C_{b,\gamma}' \norm{\hat{\theta}_{b-1}-\theta^*}_2 + C_{b,\gamma}'') \\
	& = \cbk{ \frac{\delta_{\times b,\gamma}\overline{\sigma}_V^2 \overline{\lambda}_W}{\kappa_0(\kappa_0-\delta_{\times 0}(\epsilon_\gamma,r_0))} + \frac{C_{b,\gamma}'}{\kappa_0-\delta_{\times 0}(\epsilon_\gamma,r_0)} } \norm{\hat{\theta}_{b-1}-\theta^*}_2
	+ \frac{C_{b,\gamma}''}{\kappa_0-\delta_{\times 0}(\epsilon_\gamma,r_0)}.
\end{align*}on $\calE_\eta$ for all $b\ge b_0+1$.
Therefore, by Assumption~\ref{as:global-contraction},
$\norm{\hat{\theta}_b - \theta^*}_2 \leq \rho \norm{\hat{\theta}_{b-1}-\theta^*}_2 + \norm{\xi_b}_2$,
and taking limsup on $\calE_{\eta,\gamma}$ gives
\begin{align*}
	\limsup_{b\to \infty} C_{b,\gamma}' & \leq \overline{\sigma}_V\cbk{\overline{\lambda}_W(2L\Delta_{\infty ,\eta,\gamma}+\epsilon_\gamma) + \epsilon_\gamma (2L\Delta_{\infty ,\eta,\gamma}+\overline{\sigma}_V+\epsilon_\gamma)}, \\
	\limsup_{b\to \infty} C_{b,\gamma}'' & \leq (2L\Delta_{\infty ,\eta,\gamma} + \overline{\sigma}_V + \epsilon_\gamma)(\overline{\lambda}_W+\epsilon_\gamma)(2L\Delta_{\infty ,\eta,\gamma}^2+\epsilon_\gamma), \\
	\limsup_{b\to \infty} \delta_{\times b,\gamma} & \leq \overline{\lambda}_W (2L\Delta_{\infty ,\eta,\gamma}+\epsilon_\gamma)(2L\Delta_{\infty ,\eta,\gamma}+2\overline{\sigma}_V+\epsilon_\gamma) + \epsilon_\gamma (2L\Delta_{\infty ,\eta,\gamma}+\epsilon_\gamma+\overline{\sigma}_V)^2.
\end{align*}Hence,
\begin{align*}
	\Delta_{\infty ,\eta,\gamma} & \leq \rho \Delta_{\infty ,\eta,\gamma} + \limsup_{b\to \infty } \norm{\xi_b}_2 \\
	& \leq K_{0,\epsilon_\gamma}(r_0) + K_{1,\epsilon_\gamma}(r_0) \Delta_{\infty ,\eta,\gamma} + K_{2,\epsilon_\gamma}(r_0) \Delta_{\infty ,\eta,\gamma}^2 + K_{3,\epsilon_\gamma}(r_0) \Delta_{\infty ,\eta,\gamma}^3
\end{align*}on $\calE_{\eta,\gamma}$, where $K_{j,\epsilon_\gamma}(r_0)$ for $j=0,1,2,3$ are defined in \eqref{eq:one-step-K}.
Let $\bar{K}_{\epsilon_\gamma}(r_0)$ be defined in \eqref{eq:one-step-Psi}, since $\Delta_{\infty,\eta,\gamma}\le r_0$, we obtain $\Delta_{\infty,\eta,\gamma}\le K_{0,\epsilon_\gamma}(r_0)+\bar{K}_{\epsilon_\gamma}(r_0)\Delta_{\infty,\eta,\gamma}$ on $\calE_{\eta,\gamma}$.
Hence, by Assumption~\ref{as:regularProblem} item (b), $\Delta_{\infty,\eta,\gamma} \le K_{0,\epsilon_\gamma}(r_0)/(1-\bar{K}_{\epsilon_\gamma}(r_0))\leq \gamma$ on $\calE_{\eta,\gamma}$.
Because $\gamma=\varepsilon/2$, we have
$\limsup_{b\to\infty}\norm{\hat\theta_b-\theta^*}_2\le \gamma = \varepsilon/2$
on $\calE_{\eta,\gamma}$.
Thus
\begin{align*}
\pr\sbk{\norm{\hat\theta_b-\theta^*}_2>\varepsilon}
\le \pr\sbk{\norm{\hat\theta_b-\theta^*}_2>\varepsilon,\calE_{\eta,\gamma}} + \pr(\calE_{\eta,\gamma}^c).
\end{align*}
Taking $\limsup$ gives
\begin{align*}
\limsup_{b\to\infty}\pr(\norm{\hat\theta_b-\theta^*}_2>\varepsilon) \le \pr(\calE_{\eta,\gamma}^c) \le \eta.
\end{align*}
Since $\eta>0$ is arbitrary, $\hat\theta_b\cin{p} \theta^*$.

We next prove $\hat{V}_b\cin{p} V(\theta^*)$.
For any fixed $\varepsilon,\eta>0$, choose $\gamma>0$ such that $3L\gamma+\epsilon_\gamma<\varepsilon$.
On $\calE_{\eta,\gamma}$, we already proved $\limsup_{b\to\infty}\norm{\hat\theta_b-\theta^*}_2\le \gamma.$
Hence $\limsup_{b\to\infty}P_b\le 2\gamma$ on $\calE_{\eta,\gamma}$, and by \eqref{eq:deltaVU} in Lemma~\ref{lemma:bound_diff_VUhat} and Lipschitz continuity of $V(\cdot)$,
\begin{align*}
	\limsup_{b\to\infty}\norm{\hat{V}_b-V(\theta^*)}_2 &
    \le \limsup_{b\to\infty}\norm{\hat{V}_b-V(\hat\theta_{b-1})}_2
    + \limsup_{b\to\infty}\norm{V(\hat\theta_{b-1})-V(\theta^*)}_2 \\
	& \le (2L\gamma+\epsilon_\gamma) + L\gamma < \varepsilon
\end{align*}on $\calE_{\eta,\gamma}$.
Applying the same argument as above yields $\hat{V}_b\cin{p} V(\theta^*)$.

We now prove $\hat{U}_b\cin{p} 0_q$.
For any fixed $\varepsilon,\eta>0$, choose $\gamma>0$ such that $\overline{\sigma}_V \gamma + 2L\gamma^2+\epsilon_\gamma<\varepsilon.$
On $\calE_{\eta,\gamma}$, we already proved $\limsup_{b\to\infty}\norm{\hat\theta_b-\theta^*}_2\le \gamma$ and $\limsup_{b\to\infty}Q_b\le 4\gamma^2$.
Hence \eqref{eq:m_thetahat}--\eqref{eq:deltaVU} in Lemma~\ref{lemma:bound_diff_VUhat}  gives
\begin{align*}
	\limsup_{b\to\infty}\norm{\hat{U}_b}_2 &\le \limsup_{b\to\infty}\norm{\hat{U}_b-m(\hat\theta_{b-1})}_2 + \norm{m(\hat\theta_{b-1})}_2 \\
	& \le (2L\gamma^2 + \epsilon_\gamma) + \overline{\sigma}_V \gamma \le \varepsilon
\end{align*}on $\calE_{\eta,\gamma}$.
The same argument gives $\hat{U}_b\cin{p} 0_q$.
\hfill$\qed$

\subsection{\texorpdfstring{Proof of Theorem~\ref{thm:tail_entry_as}}{Proof of Theorem B.2}} \label{sec:proof-tail_entry_as}
Define the event
\begin{align*}
	\Omega_0 & := \cbk{
		\exists B_0<\infty \text{ such that } \sup_{1\le b\le B_0}\norm{\hat\theta_b-\theta^*}_2\le r_0 \text{ and } \sup_{b\geq B_0+1}\max\{\Delta_b^V,\Delta_b^U,\Delta_b^W\}\leq \epsilon_0
	} \\
	& \qquad \cap \cbk{
		\Delta_b^V\to 0,\ \Delta_b^U\to 0,\ \Delta_b^W\to 0 \text{ as } b\to\infty
	}.
\end{align*}
By condition (c) and $\Delta_b^W$, $\Delta_b^V$, $\Delta_b^U \cin{a.s.} 0$, we have $\pr(\Omega_0)=1.$

We work on $\omega\in\Omega_0$ and let $B_0=B_0(\omega)$.
By condition~(a) and Lemma~\ref{lemma:one_step}, $\norm{\hat{\theta}_b(\omega)-\theta^*}_2\leq r_0$ for all $b\geq 1$.
Besides, all the items (a)--(e) in Lemma~\ref{lemma:one_step} hold for $b\geq b_0$.
Let $\Delta_{\infty}:=\limsup_{b\to\infty}\norm{\hat\theta_b-\theta^*}_2\leq r_0$, then similar to the proof of Theorem~\ref{thm:tail_entry}, $\limsup_{b\to\infty}P_b \le 2\Delta_\infty$ and $\limsup_{b\to\infty}Q_b \le 4\Delta_\infty^2$.

By the similar argument of Theorem~\ref{thm:tail_entry}, we bound $\norm{\xi_b}_2\le\norm{H(\hat\theta_{b-1})^{-1}-\hat{H}_b^{-1}}_2 \norm{h(\hat{\theta}_{b-1})}_2 + \norm{\hat{H}_b^{-1}}_2 \norm{h(\hat{\theta}_{b-1})-\hat{h}_b}_2$.
By \eqref{eq:deltaVU} in Lemma~\ref{lemma:bound_diff_VUhat},
$\norm{\hat{V}_b-V(\hat\theta_{b-1})}_2\le LP_b+\Delta_b^V$ and
$\norm{\hat{U}_b-m(\hat\theta_{b-1})}_2 \le L Q_b/2+\Delta_b^U$ for all $b\geq 1$.
And by \eqref{eq:hathm} in Lemma~\ref{lemma:bound_diff_VUhat},
\begin{align*}
	\norm{h(\hat{\theta}_{b-1})-\hat{h}_b}_2 & \leq \overline{\lambda}_W \overline{\sigma}_V (LP_b+\Delta_b^V)\norm{\hat{\theta}_{b-1}-\theta^*}_2
	+ \overline{\sigma}_V \Delta_b^W \sbk{LP_b+\Delta_b^V + \overline{\sigma}_V} \norm{\hat{\theta}_{b-1}-\theta^*}_2 \\
	& \quad + (\overline{\lambda}_W+\Delta_b^W) (LP_b+\Delta_b^V+\overline{\sigma}_V)\sbk{\frac{LQ_b}{2} + \Delta_b^U}
	:= C_b' \norm{\hat{\theta}_{b-1}-\theta^*}_2 + C_b'',
\end{align*}
where $C_b' = \overline{\sigma}_V\cbk{\overline{\lambda}_W(LP_b+\Delta_b^V) + \Delta_b^W (LP_b+\overline{\sigma}_V+\Delta_b^V)}$,
$C_b'' = (\overline{\lambda}_W+\Delta_b^W)(LP_b + \overline{\sigma}_V + \Delta_b^V)(LQ_b/2+\Delta_b^U)$ for all $b\geq b_0+1$.
By \eqref{eq:hath} in Lemma~\ref{lemma:bound_diff_VUhat}, $\norm{h(\hat{\theta}_{b-1})}_2 \leq \overline{\sigma}_V^2 \overline{\lambda}_W \norm{\hat{\theta}_{b-1}-\theta^*}_2$ for all $b\geq b_0+1$, and by Lemma~\ref{lemma:bound_diff_VUhat},
\begin{align*}
	& \norm{\hat H_b-H(\hat\theta_{b-1})}_2 = \norm{(\hat{V}_b-V(\hat\theta_{b-1}))^\top \hat W_b \hat{V}_b + V(\hat\theta_{b-1})^\top (\hat W_b-W)\hat{V}_b + V(\hat\theta_{b-1})^\top W (\hat{V}_b-V(\hat\theta_{b-1}))}_2 \\
	& \leq  (LP_b+\Delta_b^V) (\overline{\lambda}_W+\Delta_b^W) (LP_b+\Delta_b^V+\overline{\sigma}_V)
	+ \overline{\sigma}_V \Delta_b^W (LP_b+\Delta_b^V+\overline{\sigma}_V)
	+ \overline{\sigma}_V \overline{\lambda}_W (LP_b+\Delta_b^V) =: \delta_{\times b}
\end{align*}for all $b\geq b_0+1$.
By Lemma~\ref{lemma:one_step} that $P_b\leq 2r_0$, and $\max\{\Delta_b^V,\Delta_b^U,\Delta_b^W\}\leq \epsilon_0$, we have $\delta_{\times b} \leq \delta_{\times 0}(\epsilon_0, r_0)$ defined in \eqref{eq:one-step-C-delta}, and thus $\norm{\hat{H}_b^{-1}}_2\leq (\kappa_0-\delta_{\times 0}(\epsilon_0, r_0))^{-1}$ and
\begin{align*}
	\norm{H(\hat\theta_{b-1})^{-1}-\hat{H}_b^{-1}}_2 & \leq \norm{H(\hat\theta_{b-1})^{-1}}_2 \norm{H(\hat\theta_{b-1})-\hat{H}_b}_2 \norm{\hat{H}_b^{-1}}_2
	\leq \frac{\delta_{\times b}}{\kappa_0(\kappa_0-\delta_{\times 0}(\epsilon_0, r_0))},
\end{align*}
for all $b\ge b_0+1$.
Thus,
\begin{align*}
	\norm{\xi_b}_2 &\leq \norm{H(\hat\theta_{b-1})^{-1}-\hat{H}_b^{-1}}_2 \norm{h(\hat{\theta}_{b-1})}_2 + \norm{\hat{H}_b^{-1}}_2 \norm{h(\hat{\theta}_{b-1})-\hat{h}_b}_2 \\
	& = \cbk{ \frac{\delta_{\times b}\overline{\sigma}_V^2 \overline{\lambda}_W}{\kappa_0(\kappa_0-\delta_{\times 0}(\epsilon_0, r_0))} + \frac{C_b'}{\kappa_0-\delta_{\times 0}(\epsilon_0, r_0)} } \norm{\hat{\theta}_{b-1}-\theta^*}_2
	+ \frac{C_b''}{\kappa_0-\delta_{\times 0}(\epsilon_0, r_0)}.
\end{align*}
Taking limsup gives
\begin{align*}
	\limsup_{b\to \infty} C_b' & \leq 2L\overline{\sigma}_V\overline{\lambda}_W\Delta_{\infty}, \\
	\limsup_{b\to \infty} C_b'' & \leq 2L\overline{\lambda}_W(2L\Delta_{\infty}^3 + \overline{\sigma}_V \Delta_{\infty}^2) \\
	\limsup_{b\to \infty} \delta_{\times b} & \leq 4L\overline{\lambda}_W\Delta_{\infty}(L\Delta_{\infty}+\overline{\sigma}_V).
\end{align*}Hence, by Assumption~\ref{as:global-contraction}, we have
$\Delta_{\infty} \leq \rho \Delta_{\infty} + K_{2,0}\Delta_{\infty}^2 + K_{3,0}\Delta_{\infty}^3$, where
\begin{align*}
		K_{2,0} &:= \frac{4L\overline{\sigma}_V\overline{\lambda}_W}{\kappa_0-\delta_{\times 0}(\epsilon_0, r_0)}
		\sbk{\frac{\overline{\sigma}_V^2\overline{\lambda}_W}{\kappa_0}+1},  \quad
		K_{3,0} := \frac{4L^2\overline{\lambda}_W}{\kappa_0-\delta_{\times 0}(\epsilon_0, r_0)}
		\sbk{\frac{\overline{\sigma}_V^2\overline{\lambda}_W}{\kappa_0}+1}.
\end{align*}
Since $\Delta_\infty\le r_0$, we have $\Delta_\infty \le \cbk{\rho+K_{2,0}r_0+K_{3,0}r_0^2}\Delta_\infty$ and by condition (b), $\Delta_{\infty}=0$.
Thus $\hat{\theta}_b(\omega)\to\theta^*$.
Since $\omega\in\Omega_0$ was arbitrary and $\pr(\Omega_0)=1$, we conclude that $\hat\theta_b \cin{a.s.} \theta^*$.

It remains to prove the limits for $\hat{V}_b$ and $\hat{U}_b$.
Because $\hat\theta_b \cin{a.s.} \theta^*$, we have
\begin{align*}
\frac1b\sum_{i=1}^{b-1}\norm{\hat\theta_i-\theta^*}_2\cin{a.s.} 0,
\qquad
\frac1b\sum_{i=1}^{b-1}\norm{\hat\theta_i-\theta^*}_2^2\cin{a.s.} 0,
\end{align*}
and therefore $P_b \cin{a.s.} 0$ and $Q_b \cin{a.s.} 0$ a.s.
By Lemma~\ref{lemma:bound_diff_VUhat},
$\norm{\hat{V}_b-V(\hat\theta_{b-1})}_2 \cin{a.s.} 0$ and $\norm{\hat{U}_b-m(\hat\theta_{b-1})}_2 \cin{a.s.} 0$.
Since $V(\cdot)$ and $m(\cdot)$ are continuous and $\hat\theta_{b-1} \cin{a.s.} \theta^*$, $V(\hat\theta_{b-1}) \cin{a.s.} V(\theta^*)$ and $m(\hat\theta_{b-1}) \cin{a.s.} m(\theta^*)=0_q$, we have $\hat{V}_b \cin{a.s.} V(\theta^*)$ and $\hat{U}_b \cin{a.s.} 0_q$.
\hfill\qed

\section{Lemmas and their proofs} \label{sec:proof-of-lemmas}
\subsection{Lemmas for theorems under diverging batch size}\label{sec:proof-convergence-uv}

\begin{lemma}[Convergence of $\hat{U}_b$ and $\hat{V}_b$] \label{lem:convergence-uv}
	Suppose Assumptions \ref{asum:stationarity}--\ref{asum:moment-dominance} hold.
	\begin{enumerate}
		\item Under Assumptions \ref{asum:gradient-lipschitz}(a) and \ref{asum:gradient-uniform}(a)(i),
		$\hat{V}_b \cin{p} V = \E\{\nabla g(\theta^*, x_{i,j})\}$ and
		$\hat{U}_b \cin{p} U = \E\{g(\theta^*, x_{i,j})\} = 0$ as
		$\min(n_1, \ldots, n_b) \to \infty$.
		\item Under Assumptions \ref{asum:gradient-lipschitz}(b) and \ref{asum:gradient-uniform}(b)(i),
		$\hat{V}_b \cin{p} V$ and
		$\hat{U}_b \cin{p} U = \E\{g(\theta^*, x_{i,j})\} = 0$ as
		$n_1 \to \infty$.
	\end{enumerate}
\end{lemma}

\begin{proof}
	We can prove this lemma under the condition that $\norm{\hat\theta_i -\theta^*}_2 = o_p(1)$ and $\hat\theta_i \in B_\epsilon$ for $i = 1, \ldots, b-1$;
	see the proof of Theorem \ref{thm:consistency} based on induction with the same assumptions.

	(Case a). First, we follow Theorem 4.1 in \citet{newey1994asymptotic} to show that $\hat{V}_b \cin{p} V$.
	By the triangle inequality, Assumptions \ref{asum:gradient-lipschitz}(a) and \ref{asum:gradient-uniform}(a)(i), we have
	\begin{align}
		& \norm{\hat{V}_b - V}_F \nonumber \\
		={}& \norm{\frac{1}{N_b} \left\{ \sum_{i=1}^{b-1} \nabla G(\hat\theta_i; D_i) +\nabla G(\hat\theta_{b-1}; D_b) \right\} -\E\{\nabla g(\theta^*, x_{i,j})\}}_F \nonumber \\
		\le{}& \norm{\frac{1}{N_b} \left\{ \sum_{i=1}^{b-1} \nabla G(\hat\theta_i; D_i) +\nabla G(\hat\theta_{b-1}; D_b) \right\} -\sum_{i=1}^{b-1} \frac{n_i}{N_b} \E\{\nabla g(\hat\theta_i, x_{i,j})\} -\frac{n_b}{N_b} \E\{\nabla g(\hat\theta_{b-1}, x_{i,j})\}}_F  \nonumber \\
		& +\norm{\sum_{i=1}^{b-1} \frac{n_i}{N_b} \E\{\nabla g(\hat\theta_i, x_{i,j})\} +\frac{n_b}{N_b} \E\{\nabla g(\hat\theta_{b-1}, x_{i,j})\} -\sum_{i=1}^b \frac{n_i}{N_b} \E\{\nabla g(\theta^*, x_{i,j})\}}_F \nonumber \\
		\le{}& \sum_{i=1}^b \frac{n_i}{N_b} \sup_{\theta \in B_\epsilon} \norm{ \frac{1}{n_i} \nabla G(\theta; D_i) -\E\{\nabla g(\theta, x_{i,j})\} }_F
		+L \left( \frac{n_b}{N_b} \norm{\hat\theta_{b-1} -\theta^*}_2 +\sum_{i=1}^{b-1} \frac{n_i}{N_b} \norm{\hat\theta_i -\theta^*}_2 \right) \nonumber \\
		\cin{p}{}& 0. \label{eq:convergence-v-a}
	\end{align}
	Next, we bound the Taylor approximation error using Assumptions \ref{asum:gradient-lipschitz}(a) and \ref{asum:gradient-uniform}(a)(ii).
	We will point out the difference under Assumption \ref{asum:gradient-uniform}(a)(i) later.
	To illustrate,
	\begin{align*}
		& \norm{\frac{1}{n_i} \nabla G(\bar\theta_i; D_i) -\frac{1}{n_i} \nabla G(\hat\theta_i; D_i)}_F \\
		\le{}& \norm{ \frac{1}{n_i} \nabla G(\bar\theta_i; D_i) -\E\{\nabla g(\bar\theta_i, x_{i,j})\}}_F
		+\norm{ \frac{1}{n_i} \nabla G(\hat\theta_i; D_i) -\E\{\nabla g(\hat\theta_i, x_{i,j})\}}_F \\
		& +\norm{ \E\{\nabla g(\bar\theta_i, x_{i,j})\} -\E\{\nabla g(\hat\theta_i, x_{i,j})\}}_F \\
		\le{}& 2 \sup_{\theta \in B_\epsilon} \norm{ \frac{1}{n_i} \nabla G(\theta; D_i) -\E\{\nabla g(\theta, x_{i,j})\}}_F
		+L \norm{\bar\theta_i -\hat\theta_i}_2 \\
		={}& O_p(n_i^{-1/2} +\norm{\hat\theta_i -\theta^*}_2),
	\end{align*}
	where $\bar\theta_i$ is a point between $\hat\theta_i$ and $\theta^*$.
	By the mean value theorem,
	\begin{align*}
	G(\theta^*; D_i)
	= G(\hat\theta_i; D_i) +\nabla G(\hat\theta_i; D_i) (\theta^*-\hat\theta_i)
	+\left\{ \nabla G(\bar\theta_i; D_i) -\nabla G(\hat\theta_i; D_i) \right\} (\theta^*-\hat\theta_i),
	\end{align*}
	where the sub-multiplicativity of Frobenius norm gives
	\begin{align}
		\norm{\left\{ \nabla G(\bar\theta_i; D_i) -\nabla G(\hat\theta_i; D_i) \right\} (\theta^*-\hat\theta_i)}_2
		&\le \norm{ \nabla G(\bar\theta_i; D_i) -\nabla G(\hat\theta_i; D_i) }_F \norm{\theta^*-\hat\theta_i}_2 \nonumber \\
		&= O_p\left( n_i^{1/2} \norm{\hat\theta_i -\theta^*}_2 +n_i\norm{\hat\theta_i -\theta^*}_2^2 \right). \label{eq:taylor-error-2nd-a}
	\end{align}
	Repeating this trick, we have
	\begin{align}
		\begin{split}
			\bar{g}_b(\theta^*)
			={}& \frac{1}{N_b} \sum_{i=1}^{b-1} \left\{ G(\hat\theta_i; D_i) +\nabla G(\hat\theta_i; D_i) (\theta^*-\hat\theta_i) \right\}
			+\frac{1}{N_b} G(\theta^*; D_b) \nonumber \\
			& +O_p\left\{ \sum_{i=1}^{b-1} \bigg( \frac{n_i^{1/2}}{N_b} \norm{\hat\theta_i-\theta^*}_2 +\frac{n_i}{N_b} \norm{\hat\theta_i-\theta^*}_2^2 \bigg) \right\} \nonumber \\
			={}& \frac{N_{b-1}}{N_b} \{U_{b-1} +V_{b-1} (\theta^*-\hat\theta_{b-1})\}
			+\frac{1}{N_b} G(\theta^*; D_b) \\
			& +O_p\left\{ \sum_{i=1}^{b-1} \bigg( \frac{n_i^{1/2}}{N_b} \norm{\hat\theta_i-\theta^*}_2 +\frac{n_i}{N_b} \norm{\hat\theta_i-\theta^*}_2^2 \bigg) \right\} \nonumber
		\end{split} \\
		\begin{split}
			={}& \hat{U}_b +\hat{V}_b (\theta^*-\hat\theta_{b-1})
			+O_p\left( \frac{n_b^{1/2}}{N_b} \norm{\hat\theta_{b-1}-\theta^*}_2 +\sum_{i=1}^{b-1} \frac{n_i^{1/2}}{N_b} \norm{\hat\theta_i-\theta^*}_2 \right) \\
			& +O_p\left( \frac{n_b}{N_b} \norm{\hat\theta_{b-1}-\theta^*}_2^2 +\sum_{i=1}^{b-1} \frac{n_i}{N_b} \norm{\hat\theta_i-\theta^*}_2^2 \right). \label{eq:moment-error-a}
		\end{split}
	\end{align}
	If we use Assumption \ref{asum:gradient-uniform}(a)(i) instead, \eqref{eq:moment-error-a} becomes
	\begin{align*}
		\bar{g}_b(\theta^*)
		={}& \hat{U}_b +\hat{V}_b (\theta^*-\hat\theta_{b-1})
		+o_p\left( \frac{n_b}{N_b} \norm{\hat\theta_{b-1}-\theta^*}_2 +\sum_{i=1}^{b-1} \frac{n_i}{N_b} \norm{\hat\theta_i-\theta^*}_2 \right) \\
		& +O_p\left( \frac{n_b}{N_b} \norm{\hat\theta_{b-1}-\theta^*}_2^2 +\sum_{i=1}^{b-1} \frac{n_i}{N_b} \norm{\hat\theta_i-\theta^*}_2^2 \right).
	\end{align*}
	Since $\bar{g}_b(\theta^*) \cin{p} 0$ by the uniform law of large numbers (see, e.g., Lemma 2.4 in \citet{newey1994asymptotic}),
	$\hat{V}_b \cin{p} V$, and
	$\norm{\hat\theta_i -\theta^*}_2 = o_p(1)$ for $i = 1, \ldots, b-1$,
		applying Slutsky's theorem gives $\hat{U}_b = \bar{g}_b(\theta^*) -\hat{V}_b(\theta^*-\hat\theta_{b-1}) +o_p(1) \cin{p} 0$.

	(Case b). Again, we try to show that $\hat{V}_b \cin{p} V$ first.
	By the triangle inequality, Assumptions \ref{asum:gradient-lipschitz}(b) and \ref{asum:gradient-uniform}(b)(i),
	\begin{align}
		& \norm{\hat{V}_b - V}_F \nonumber \\
		={}& \norm{ \frac{1}{N_b} \left\{ \sum_{i=1}^{b-1} \nabla G(\hat\theta_i; D_i) +\nabla G(\hat\theta_{b-1}; D_b) \right\} -\E\{\nabla g(\theta^*, x_{i,j})\}}_F \nonumber \\
		\le{}& \norm{ \frac{1}{N_b} \left\{ \sum_{i=1}^{b-1} \nabla G(\hat\theta_i; D_i) +\nabla G(\hat\theta_{b-1}; D_b) \right\} -\frac{1}{N_b} \sum_{i=1}^b \nabla G(\theta^*; D_i) }_F \nonumber \\
		& +\norm{ \frac{1}{N_b} \sum_{i=1}^b \nabla G(\theta^*; D_i) -\E\{\nabla g(\theta^*, x_{i,j})\} }_F \nonumber \\
		\le{}& L \left( \frac{n_b}{N_b} \norm{\hat\theta_{b-1}-\theta^*}_2 +\sum_{i=1}^{b-1} \frac{n_i}{N_b} \norm{\hat\theta_i-\theta^*}_2 \right)
		+\norm{ \frac{1}{N_b} \sum_{i=1}^b \nabla G(\theta^*; D_i) -\E\{\nabla g(\theta^*, x_{i,j})\} }_F \nonumber \\
		\cin{p}{}& 0. \label{eq:convergence-v-b}
	\end{align}
	Using Assumption \ref{asum:gradient-lipschitz}(b),
	\begin{align}
		\norm{\nabla G(\bar\theta_i; D_i) -\nabla G(\hat\theta_i; D_i)}_F
		&\le \norm{\nabla G(\bar\theta_i; D_i) -\nabla G(\theta^*; D_i)}_F +\norm{\nabla G(\hat\theta_i; D_i) -\nabla G(\theta^*; D_i)}_F \nonumber \\
		&= L \cdot n_i (\norm{\bar\theta_i -\theta^*}_2 +\norm{\hat\theta_i -\theta^*}_2)
		= O_p\left( n_i \norm{\hat\theta_i -\theta^*}_2 \right), \label{eq:taylor-error-2nd-b}
	\end{align}
	where $\bar\theta_i$ is a point between $\hat\theta_i$ and $\theta^*$.
	Repeating the mean value theorem, we have
	\begin{align}
		\begin{split}
			\bar{g}_b(\theta^*)
			={}& \frac{1}{N_b} \sum_{i=1}^{b-1} \left\{ G(\hat\theta_i; D_i) +\nabla G(\hat\theta_i; D_i) (\theta^*-\hat\theta_i) \right\}
			+\frac{1}{N_b} G(\theta^*; D_b)
			+O_p\left( \sum_{i=1}^{b-1} \frac{n_i}{N_b} \norm{\hat\theta_i-\theta^*}_2^2 \right) \nonumber \\
			={}& \frac{N_{b-1}}{N_b} \{U_{b-1} +V_{b-1} (\theta^*-\hat\theta_{b-1})\}
			+\frac{1}{N_b} G(\theta^*; D_b)
			+O_p\left( \sum_{i=1}^{b-1} \frac{n_i}{N_b} \norm{\hat\theta_i-\theta^*}_2^2 \right) \nonumber
		\end{split} \\
		\begin{split}
			={}& \hat{U}_b +\hat{V}_b (\theta^*-\hat\theta_{b-1})
			+O_p\left( \frac{n_b}{N_b} \norm{\hat\theta_{b-1}-\theta^*}_2^2 +\sum_{i=1}^{b-1} \frac{n_i}{N_b} \norm{\hat\theta_i-\theta^*}_2^2 \right). \label{eq:moment-error-b}
		\end{split}
	\end{align}
	The conclusion then follows by the same argument in case (a).
\end{proof}

\begin{lemma}[Bound for sums of sample sizes] \label{lem:bound-n}
	Recall that $n_i$ is the sample size of $D_i$ and $N_i = \sum_{j=1}^i n_j$.
	We have
	\begin{align*}
	\sum_{i=1}^b \frac{n_i}{N_i} \le 1 +\log \frac{N_b}{n_1}, \quad
	\sum_{i=1}^b \frac{n_i}{\sqrt{N_i}} \le 2 \sqrt{N_b} \quad \text{and} \quad
	\sum_{i=1}^b \frac{\sqrt{n_i}}{\sqrt{N_i}} \le \frac{2\sqrt{N_b}}{\sqrt{\min(n_1, \ldots, n_b)}}.
	\end{align*}
\end{lemma}

\begin{proof}
	The first two inequalities are restated from \citet{luo2022qif} for ease of reference.
	Their proof can be found in the supplementary material of \citet{luo2022qif}.
	The last inequality follows from
	\begin{align*}
	\sum_{i=1}^b \frac{\sqrt{n_i}}{\sqrt{N_i}}
	\le \frac{1}{\sqrt{\min(n_1, \ldots, n_b)}} \sum_{i=1}^b \frac{n_i}{\sqrt{N_i}}
	\le \frac{2\sqrt{N_b}}{\sqrt{\min(n_1, \ldots, n_b)}}.
	\end{align*}
\end{proof}

\subsection{Lemmas for theorems under fixed batch size}\label{sec:proof-convergence-tail}

Let
\begin{align*}
    P_b:=\frac1b\sum_{i=1}^{b-1}\norm{\hat\theta_i-\hat\theta_{b-1}}_2, \qquad
    Q_b:=\frac1b\sum_{i=1}^{b-1}\norm{\hat\theta_i-\hat\theta_{b-1}}_2^2.
\end{align*}

\begin{lemma}\label{lemma:bound_diff_VUhat}
	Suppose Assumptions~\ref{asum:stationarity}--\ref{asum:moment-dominance}, \ref{as:weight-new}, \ref{as:cvx} and \ref{as:Lip_nabla_global} hold.
	Then
	\begin{align}
		& \norm{m(\hat\theta_{b-1})}_2 \leq \overline{\sigma}_V \norm{\hat\theta_{b-1}-\theta^*}_2, \label{eq:m_thetahat} \\
		& \norm{\hat{V}_b - V(\hat\theta_{b-1})}_2 \leq L P_b + \Delta_b^V,\qquad
		\norm{\hat{U}_b - m(\hat{\theta}_{b-1})}_2 \leq \frac{L Q_b}{2} + \Delta_b^U, \label{eq:deltaVU} \\
		& \norm{\hat{V}_b}_2 \leq LP_b + \Delta_b^V + \overline{\sigma}_V, \label{eq:deltaV} \\
		& \norm{h(\hat{\theta}_{b-1})-\hat{h}_b}_2 \leq \overline{\lambda}_W\overline{\sigma}_V (LP_b+\Delta_b^V)\norm{\hat{\theta}_{b-1}-\theta^*}_2
		+ \overline{\sigma}_V \Delta_b^W \sbk{LP_b+\Delta_b^V + \overline{\sigma}_V}\norm{\hat{\theta}_{b-1}-\theta^*}_2 \nonumber \\
		& \qquad\qquad\qquad\quad + (\overline{\lambda}_W+\Delta_b^W) (LP_b+\Delta_b^V+\overline{\sigma}_V)\sbk{\frac{LQ_b}{2} + \Delta_b^U}, \label{eq:hathm} \\
		& \norm{h(\hat{\theta}_{b-1})}_2 \leq \overline{\sigma}_V^2 \overline{\lambda}_W \norm{\hat{\theta}_{b-1}-\theta^*}_2. \label{eq:hath}
	\end{align}
\end{lemma}

\begin{proof}
By $m(\hat\theta_{b-1})-m(\theta^*)=\int_0^1 V\big(\theta^*+t(\hat\theta_{b-1}-\theta^*)\big)(\hat\theta_{b-1}-\theta^*)\,dt$, $m(\theta^*)=0_q$, and Assumption~\ref{as:Lip_nabla_global}, we have $\norm{m(\hat\theta_{b-1})}_2 \leq \overline{\sigma}_V \norm{\hat\theta_{b-1}-\theta^*}_2$, we thus show \eqref{eq:m_thetahat}.

Then by Assumption~\ref{as:Lip_nabla_global},
\begin{align*}
	\norm{\hat{V}_b - V(\hat\theta_{b-1})}_2 &= \norm{
		\frac1b\sum_{i=1}^{b-1}\cbk{V_i(\hat\theta_i)-V_i(\hat\theta_{b-1})}
		+\frac1b\sum_{i=1}^b \cbk{V_i(\hat\theta_{b-1})-V(\hat\theta_{b-1})}
	}_2\nonumber \\
	&\leq \frac1b\sum_{i=1}^{b-1}\norm{V_i(\hat\theta_i)-V_i(\hat\theta_{b-1})}_2
	+\sup_{\theta\in\Theta} \norm{\frac1b\sum_{i=1}^b V_i(\theta)-V(\theta)}_2 \nonumber \\
	& \leq LP_b + \Delta_b^V.
\end{align*}This argument is similar to (C.4) in the proof of Lemma~1 (Case b).
Likewise,
\begin{align*}
	\hat{U}_b - m(\hat{\theta}_{b-1}) &= \frac1b\sum_{i=1}^{b-1} \cbk{ g_i(\hat\theta_i)+V_i(\hat\theta_i)(\hat\theta_{b-1}-\hat\theta_i) } + \frac1b g_b(\hat\theta_{b-1}) - m(\hat{\theta}_{b-1}) \\
	&= \frac1b\sum_{i=1}^{b-1}
	\underbrace{\cbk{g_i(\hat\theta_i)-g_i(\hat\theta_{b-1})+V_i(\hat\theta_i)(\hat\theta_{b-1}-\hat\theta_i)}}_{=:R_{i,b}}
	+ \frac1b\sum_{i=1}^b \{g_i(\hat\theta_{b-1})-m(\hat\theta_{b-1})\}.
\end{align*}

By the first-order Taylor expansion,
\begin{align*}
	\norm{R_{i,b}}_2 &= \norm{-R_{i,b}}_2 = \norm{
		g_i(\hat{\theta}_{b-1}) - g_i(\hat{\theta}_i) - V_i(\hat{\theta}_i)(\hat{\theta}_{b-1}-\hat{\theta}_i)
	}_2 \\
	&= \norm{
		\int_{0}^1 \cbk{V_i(\hat{\theta}_i+t(\hat{\theta}_{b-1}-\hat{\theta}_i))
		-V_i(\hat{\theta}_i)} (\hat{\theta}_{b-1}-\hat{\theta}_i) dt
	}_2 \\
	&\leq \cbk{\frac{1}{n_0}\sum_{j=1}^{n_0}L(x_{i,j})} \int_{0}^1 t \norm{\hat{\theta}_{b-1}-\hat{\theta}_i}_2^2 dt
	\leq \frac{L}{2}\norm{\hat{\theta}_{b-1}-\hat{\theta}_i}_2^2.
\end{align*}
Hence
\begin{align*}
	\norm{\hat{U}_b - m(\hat{\theta}_{b-1})}_2 & \leq
	\frac{1}{b}\sum_{i=1}^{b-1}\norm{R_{i,b}}_2 + \sup_{\theta\in\Theta}\norm{\frac1b\sum_{i=1}^b g_i(\theta)-m(\theta)}_2
	\leq \frac{L Q_b}{2} + \Delta_b^U.
\end{align*}
We thus prove \eqref{eq:deltaVU}.

Then, it is straightforward that
\begin{align*}
	\norm{\hat{V}_b}_2\leq \norm{\hat{V}_b-V(\hat{\theta}_{b-1})}_2 + \norm{V(\hat{\theta}_{b-1})}_2 \leq LP_b + \Delta_b^V + \overline{\sigma}_V,
\end{align*}which is \eqref{eq:deltaV}.

Next, by \eqref{eq:m_thetahat}--\eqref{eq:deltaV} and $\norm{\hat{W}_b}_2\leq\overline{\lambda}_W + \Delta_b^W$, we have
\begin{align*}
	& \norm{h(\hat{\theta}_{b-1})-\hat{h}_b}_2
	= \norm{V(\hat{\theta}_{b-1})^\top W m(\hat{\theta}_{b-1}) - \hat{V}_b^\top \hat{W}_b \hat{U}_b}_2 \\
	&\ =\norm{\{V(\hat\theta_{b-1})-\hat{V}_b\}^\top W m(\hat\theta_{b-1})
	+\hat{V}_b^\top (W-\hat W_b)m(\hat\theta_{b-1})
	+\hat{V}_b^\top \hat W_b \{m(\hat\theta_{b-1})-\hat{U}_b\}}_2 \\
	&\ \leq \norm{V(\hat\theta_{b-1})-\hat{V}_b}_2 \norm{W}_2 \norm{m(\hat{\theta}_{b-1})}_2
	+ \norm{\hat{V}_b}_2 \norm{W-\hat{W}_b}_2 \norm{m(\hat{\theta}_{b-1})}_2
	+ \norm{\hat{V}_b}_2 \norm{\hat{W}_b}_2 \norm{m(\hat\theta_{b-1})-\hat{U}_b}_2 \\
	&\ \leq \overline{\lambda}_W\overline{\sigma}_V (LP_b+\Delta_b^V)\norm{\hat{\theta}_{b-1}-\theta^*}_2
	+ \overline{\sigma}_V \Delta_b^W \sbk{LP_b+\Delta_b^V + \overline{\sigma}_V}\norm{\hat{\theta}_{b-1}-\theta^*}_2 \\
	& \quad + (\overline{\lambda}_W+\Delta_b^W) (LP_b+\Delta_b^V+\overline{\sigma}_V)\sbk{\frac{LQ_b}{2} + \Delta_b^U},
\end{align*}which is \eqref{eq:hathm}.
At last,
\begin{align*}
	\norm{h(\hat{\theta}_{b-1})}_2 \leq \norm{V(\hat{\theta}_{b-1})}_2 \norm{W}_2 \norm{m(\hat{\theta}_{b-1})}_2 \leq \overline{\sigma}_V^2 \overline{\lambda}_W \norm{\hat{\theta}_{b-1}-\theta^*}_2,
\end{align*}we thus prove \eqref{eq:hath}.
\end{proof}

\begin{lemma}[One-step invariant-ball estimate]
\label{lemma:one_step}
Suppose Assumptions~\ref{asum:stationarity}--\ref{asum:moment-dominance}, \ref{as:weight-new}, \ref{as:cvx}, \ref{as:Lip_nabla_global}, \ref{as:full_rank_V} and \ref{as:global-contraction} hold.
Fix constants $r_0>0$ and $\epsilon>0$.
Assume that there exists $b_0$ such that
\begin{align*}
\label{eq:versionB_good_event_step}
\sup_{b\leq b_0}\norm{\hat\theta_b-\theta^*}_2\le r_0,\qquad
\max\{\Delta_{b_0+1}^V, \Delta_{b_0+1}^U, \Delta_{b_0+1}^W\} \leq \epsilon.
\end{align*}We further let
\begin{equation}\label{eq:one-step-C-delta}
	\begin{aligned}
		\delta_{\times 0} = \delta_{\times 0}(\epsilon,r_0) &:= \overline{\lambda}_W (2Lr_0+\epsilon)(2Lr_0+2\overline{\sigma}_V+\epsilon) + \epsilon (2Lr_0+\epsilon+\overline{\sigma}_V)^2, \\
		C_0' = C_0'(\epsilon,r_0) &:= \overline{\sigma}_V\cbk{\overline{\lambda}_W(2Lr_0+\epsilon) + \epsilon (\overline{\sigma}_V+2Lr_0+\epsilon)}, \\
		C_0'' = C_0''(\epsilon,r_0) &:= (\overline{\sigma}_V + 2Lr_0 + \epsilon)(\overline{\lambda}_W+\epsilon)(2Lr_0^2+\epsilon),
	\end{aligned}
\end{equation}
Assume $\kappa_0-\delta_{\times 0}>0$, then the following hold.
\begin{enumerate}
\item $P_{b_0+1}\le 2r_0$ and $Q_{b_0+1}\le 4r_0^2$.

\item $\norm{\hat{V}_{b_0+1}-V(\hat\theta_{b_0})}_2 \le 2Lr_0+\epsilon$ and $\norm{\hat{U}_{b_0+1}-m(\hat\theta_{b_0})}_2 \le 2Lr_0^2+\epsilon.$

\item $\norm{\hat{V}_{b_0+1}}_2 \le \overline{\sigma}_V+2Lr_0+\epsilon$.

\item $\norm{\hat H_{b_0+1}-H(\hat\theta_{b_0})}_2\leq \delta_{\times 0}$, $\hat H_{b_0+1}$ is invertible and $\norm{\hat{H}_{b_0+1}^{-1}}_2\leq (\kappa_0-\delta_{\times 0})^{-1}$.

\item
\begin{align*}
	\norm{\hat{\theta}_{b_0+1}-\theta^*}_2 \leq \cbk{\rho+\frac{\delta_{\times 0}\overline{\sigma}_V^2 \overline{\lambda}_W}{\kappa_0(\kappa_0-\delta_{\times 0})} + \frac{C_0'}{\kappa_0-\delta_{\times 0}}} \norm{\hat{\theta}_{b_0}-\theta^*}_2+ \frac{C_0''}{\kappa_0-\delta_{\times 0}}
\end{align*}

If, in addition, $\Phi_{\epsilon}(r_0) \leq r_0$, where
\begin{align}\label{eq:one-step-Phi}
	\Phi_\epsilon(r)
	:= \left\{\rho + \frac{C_0'}{\kappa_0-\delta_{\times 0}}
	+ \frac{\delta_{\times 0}\bar\sigma_V^2\overline{\lambda}_W}{\kappa_0(\kappa_0-\delta_{\times 0})}\right\} r
	+ \frac{C_0''}{\kappa_0-\delta_{\times 0}},
\end{align}then
$\norm{\hat{\theta}_{b_0+1}-\theta^*}_2 \leq r_0$.
\end{enumerate}
\end{lemma}
By Lemma~\ref{lemma:one_step}, the condition $\Phi_\epsilon(r_0)\le r_0$ implies one-step invariance of the ball $B(\theta^*,r_0)$.
For later use, define
\begin{equation}\label{eq:one-step-K}
	\begin{aligned}
		K_{0,\epsilon}(r) &:= \frac{\epsilon\{\epsilon^2+(\bar\sigma_V+\overline{\lambda}_W)\epsilon+\overline{\lambda}_W\bar\sigma_V\}}{\kappa_0-\delta_{\times 0}(\epsilon,r)},\\
		K_{1,\epsilon}(r) &:= \rho + \frac{\epsilon}{\kappa_0-\delta_{\times 0}(\epsilon,r)}
		\bigg\{
		\frac{\bar\sigma_V^2\overline{\lambda}_W}{\kappa_0}\big[\overline{\lambda}_W(2\bar\sigma_V+\epsilon)+(\epsilon+\bar\sigma_V)^2\big]
		+\bar\sigma_V(\overline{\lambda}_W+\bar\sigma_V+\epsilon)+2L(\overline{\lambda}_W+\epsilon)
		\bigg\},\\
		K_{2,\epsilon}(r)&:= \frac{2L(\overline{\lambda}_W+\epsilon)}{\kappa_0-\delta_{\times 0}(\epsilon,r)}
		\bigg\{
		\frac{2\bar\sigma_V^2\overline{\lambda}_W}{\kappa_0}(\epsilon+\bar\sigma_V)+2\bar\sigma_V+\epsilon
		\bigg\},\\
		K_{3,\epsilon}(r) &:= \frac{4L^2(\overline{\lambda}_W+\epsilon)}{\kappa_0-\delta_{\times 0}(\epsilon,r)}
		\bigg\{
		\frac{\bar\sigma_V^2\overline{\lambda}_W}{\kappa_0}+1
		\bigg\}.
	\end{aligned}
\end{equation}
Let
\begin{equation}\label{eq:one-step-Psi}
	\begin{aligned}
	\bar K_\epsilon(r)&:=K_{1,\epsilon}(r)+K_{2,\epsilon}(r)\cdot r+K_{3,\epsilon}(r)\cdot r^2.
	\end{aligned}
\end{equation}

\begin{proof}
	For (a), by the inequalities $\norm{\hat\theta_i-\hat\theta_{b_0}}_2  \leq \norm{\hat\theta_i-\hat\theta^*}_2 + \norm{\hat\theta_{b_0}-\hat\theta^*}_2$
	and	$\norm{\hat\theta_i-\hat\theta_{b_0}}_2^2  \leq 2\norm{\hat\theta_i-\hat\theta^*}_2^2 + 2\norm{\hat\theta_{b_0}-\hat\theta^*}_2^2$, we have
	\begin{align*}
		P_{b_0+1} = \frac1{b_0+1} \sum_{i=1}^{b_0} \norm{\hat{\theta}_i - \hat{\theta}_{b_0}}_2 \leq 2r_0, \qquad
		Q_{b_0+1} = \frac1{b_0+1} \sum_{i=1}^{b_0} \norm{\hat{\theta}_i - \hat{\theta}_{b_0}}_2^2 \leq 4r_0^2.
	\end{align*}

	For (b), by Lemma~\ref{lemma:bound_diff_VUhat},
	\begin{align*}
		\norm{\hat{V}_{b_0+1}-V(\hat{\theta}_{b_0})}_2 & \leq L P_{b_0+1} + \Delta_{b_0+1}^V \leq 2Lr_0+\epsilon, \\
		\norm{\hat{U}_{b_0+1}-m(\hat{\theta}_{b_0})}_2 & \leq \frac{LQ_{b_0+1}}{2} + \Delta_{b_0+1}^U \leq 2Lr_0^2 + \epsilon.
	\end{align*}

	For (c), it is straightforward that
	$\norm{\hat{V}_{b_0+1}}_2 \leq \norm{\hat{V}_{b_0+1}-V(\hat{\theta}_{b_0})}_2 + \norm{V(\hat{\theta}_{b_0})}_2
	\leq 2Lr_0+\epsilon+\overline{\sigma}_V$.

	For (d),by $\norm{\hat W_b}_2 \leq \overline{\lambda}_W + \epsilon$, note that
	\begin{align*}
		& \norm{\hat H_{b_0+1}-H(\hat\theta_{b_0})}_2 \\
		&= \norm{(\hat{V}_{b_0+1}-V(\hat\theta_{b_0}))^\top \hat W_{b_0+1} \hat{V}_{b_0+1} + V(\hat\theta_{b_0})^\top (\hat W_{b_0+1}-W)\hat{V}_{b_0+1} + V(\hat\theta_{b_0})^\top W (\hat{V}_{b_0+1}-V(\hat\theta_{b_0}))}_2 \\
		& \leq (2Lr_0+\epsilon)(\overline{\lambda}_W+\epsilon)(2Lr_0+\epsilon+\overline{\sigma}_V)
		+ \overline{\sigma}_V\epsilon(2Lr_0+\epsilon+\overline{\sigma}_V) + \overline{\sigma}_V\overline{\lambda}_W(2Lr_0+\epsilon) \\
		& = \overline{\lambda}_W (2Lr_0+\epsilon)(2Lr_0+2\overline{\sigma}_V+\epsilon) + \epsilon (2Lr_0+\epsilon+\overline{\sigma}_V)^2 := \delta_{\times 0}.
	\end{align*}
	Since $\lambda_{\min}(\hat{H}_{b_0+1})\geq \lambda_{\min}(H(\hat{\theta}_{b_0})) - \norm{\hat H_{b_0+1}-H(\hat\theta_{b_0})}_2 \geq \kappa_0-\delta_{\times 0} > 0$ by Assumption~\ref{as:full_rank_V}, $\hat{H}_{b_0+1}$ is invertible and $\norm{\hat{H}_{b_0+1}^{-1}}_2\leq (\kappa_0-\delta_{\times 0})^{-1}$.

	At last, for (e), by $T(\hat{\theta}_{b_0})=\hat{\theta}_{b_0}-H(\hat{\theta}_{b_0})^{-1}h(\hat{\theta}_{b_0})$, write
	\begin{align*}
		\hat{\theta}_{b_0+1} &= \hat{\theta}_{b_0} - \hat{H}_{b_0+1}^{-1}\hat{h}_{b_0+1}
		= T(\hat{\theta}_{b_0}) + \underbrace{H(\hat\theta_{b_0})^{-1}h(\hat\theta_{b_0})-\hat H_{b_0+1}^{-1}\hat h_{b_0+1}}_{=:\xi_{b_0+1}}
		= T(\hat{\theta}_{b_0}) + \xi_{b_0+1}.
	\end{align*}We want to bound $\norm{\xi_{b_0+1}}_2$.
	Note that
	\begin{align*}
		\norm{\xi_{b_0+1}}_2 \leq \norm{H(\hat\theta_{b_0})^{-1}-\hat{H}_{b_0+1}^{-1}}_2 \norm{h(\hat{\theta}_{b_0})}_2 + \norm{\hat{H}_{b_0+1}^{-1}}_2 \norm{h(\hat{\theta}_{b_0})-\hat{h}_{b_0+1}}_2,
	\end{align*}and we bound these terms one by one.
	By \eqref{eq:hathm} and \eqref{eq:hath},
	\begin{align*}
		\norm{h(\hat{\theta}_{b_0})-\hat{h}_{b_0+1}}_2 & \leq \overline{\lambda}_W\overline{\sigma}_V(2Lr_0+\epsilon)\norm{\hat{\theta}_{b_0}-\theta^*}_2
		+ \overline{\sigma}_V\epsilon(2Lr_0+\epsilon+\overline{\sigma}_V)\norm{\hat{\theta}_{b_0}-\theta^*}_2 \\
		& \ + (\overline{\lambda}_W+\epsilon)(2Lr_0+\epsilon+\overline{\sigma}_V)(2Lr_0^2+\epsilon) \\
		& = C_0'\norm{\hat{\theta}_{b_0}-\theta^*}_2 +C_0'', \\
		\norm{H(\hat\theta_{b_0})^{-1}-\hat{H}_{b_0+1}^{-1}}_2 & \leq \norm{H(\hat\theta_{b_0})^{-1}}_2 \norm{H(\hat\theta_{b_0})-\hat{H}_{b_0+1}}_2 \norm{\hat{H}_{b_0+1}^{-1}}_2 \leq \frac{\delta_{\times 0}}{\kappa_0(\kappa_0-\delta_{\times 0})}, \\
		\norm{h(\hat{\theta}_{b_0})}_2 & \leq \overline{\sigma}_V^2 \overline{\lambda}_W \norm{\hat{\theta}_{b_0}-\theta^*}_2,
	\end{align*}thus
	\begin{align*}
		\norm{\xi_{b_0+1}}_2 & \leq \frac{\delta_{\times 0}\overline{\sigma}_V^2 \overline{\lambda}_W }{\kappa_0(\kappa_0-\delta_{\times 0})} \norm{\hat{\theta}_{b_0}-\theta^*}_2 + \frac{1}{\kappa_0-\delta_{\times 0}}\sbk{C_0'\norm{\hat{\theta}_{b_0}-\theta^*}_2 +C_0''}
	\end{align*}and by $\norm{\hat{\theta}_{b_0}-\theta^*}_2\leq r_0$ and Assumption~\ref{as:global-contraction},
	\begin{align*}
		\norm{\hat{\theta}_{b_0+1}-\theta^*}_2 &\leq \rho \norm{\hat{\theta}_{b_0}-\theta^*}_2 + \norm{\xi_{b_0+1}}_2 \\
		&\leq \cbk{\rho+\frac{\delta_{\times 0}\overline{\sigma}_V^2 \overline{\lambda}_W }{\kappa_0(\kappa_0-\delta_{\times 0})} + \frac{C_0'}{\kappa_0-\delta_{\times 0}}} \norm{\hat{\theta}_{b_0}-\theta^*}_2 + \frac{C_0''}{\kappa_0-\delta_{\times 0}} \leq \Phi_{\epsilon}(r_0),
	\end{align*}where $\Phi_{\epsilon}(\cdot)$ is defined in \eqref{eq:one-step-Phi}.
	If $\Phi_\epsilon(r_0)\leq r_0$ holds, then $\norm{\hat{\theta}_{b_0+1}-\theta^*}_2\leq r_0$.
\end{proof}

\section{Additional results} \label{sec:additional-results}

\subsection{Online instrumental variables regression in higher dimension} \label{sec:additional-iv-p}
We consider the same linear instrumental variables regression models \ref{enum:iv-het} and \ref{enum:iv-arma} in \cref{sec:online-instrumental-variables-regression}.
We change the dimension $p$ of $\theta^*$ and $q=2p$ as follows.
For the independent model \ref{enum:iv-het}, let $p=10,20,30,40,50$, $\theta^*=(1,\ldots,1)\in\R^{p}$ and $N_{\max}=100b$ for $b=1,\ldots,200$.
For the dependent model \ref{enum:iv-arma}, let $p=2,4,8,16$ with $\theta^*$ satisfying $\theta_1^*=0.5$, $\theta_p^*=0.2$ and $\theta_i^*=0$ otherwise,
the dimension $p'$ of $\phi^*$ be $p'=2,3,5,10$ with $\phi_1^*=0.4$, $\phi_{p'}=0.2$ and $\phi_i^*=0$ otherwise, and $N_{b}=10,000p=100b$ for $b=1,\ldots,100p$.
We also consider the case when the batch sizes are 200 and 500, but the results are similar so dropped here.

Tables~\ref{tab:51(I)2_Het_bsize=100} and \ref{tab:51(II)2_IV_ARMA_bsize=100_armode=3} show the MSEs, coverage rates and computational time in seconds.
$\OGMM$ yields the smallest MSE across all dimensions, indicating superior estimation accuracy, while $\SGMM$ gives a huge positive value, especially when using Welford's estimators.
The gap widens as $p$ grows, showing that $\OGMM$ remains stable and accurate even in high dimensions.
For the dependent model \ref{enum:iv-arma}, slightly higher MSE of $\GMM$ compared to $\OGMM$ is due to the less statistically efficient Bartlett-kernel estimator for inference and the weighting matrix.
Remarkably, by the unreported result, $\SGMM$ ($\kappa=0.5$) using \citeauthor{rlrv}'s \citeyearpar{rlrv} LRV estimator returns error induced by the inversion of nearly-singular matrices when $p=16$.
Computationally, under fixed $N_b$, Table~\ref{tab:51(I)2_Het_bsize=100} shows that $\OGMM$ is substantially faster than the offline 2SLS, and together with Table~\ref{tab:51(II)2_IV_ARMA_bsize=100_armode=3}, its time grows very slowly with $p$, reflecting its online updating nature and suitability and scalability for high-dimensional streaming data.

\begin{table}[htbp]
\centering
\setlength{\tabcolsep}{3pt}
\small
\caption{Online instrumental variables estimation of $\theta_1^*$ in the independent model \ref{enum:iv-het} with dimension $p$ and $N_b=20,000=100b$ for $b=1,\ldots,200$.
}
\label{tab:51(I)2_Het_bsize=100}
\begin{tabular}{lccccc}
\toprule
\textbf{Method} & \textbf{$p=10$} & \textbf{$p=20$} & \textbf{$p=30$} & \textbf{$p=40$} & \textbf{$p=50$} \\
\midrule
\multicolumn{6}{c}{Log MSE (coverage rate with $\alpha=0.05$)} \\
\midrule
2SLS              & $-5.459$ ($0.950)$ & $-6.424$ ($0.955)$ & $-6.899$ ($0.952)$ & $-7.236$ ($0.954)$ & $-7.491$ ($0.955$) \\
$\SGMM$ ($\kappa=0.5$) & $ 8.178$ ($0.516)$ & $18.966$ ($0.005)$ & $33.257$ ($0.000)$ & $59.413$ ($0.000)$ & $128.020$ ($0.000$) \\
$\SGMM$ ($\kappa=0.9$) & $-1.997$ ($0.786)$ & $ 8.760$ ($0.249)$ & $16.555$ ($0.007)$ & $29.581$ ($0.002)$ & $ 75.383$ ($0.000$) \\
$\OGMM$ (Welford)      & $-5.827$ ($0.941)$ & $-6.696$ ($0.944)$ & $-7.083$ ($0.948)$ & $-7.399$ ($0.945)$ & $ -7.707$ ($0.938$) \\
$\OGMM$ (LASER)        & $-5.825$ ($0.944)$ & $-6.697$ ($0.945)$ & $-7.072$ ($0.943)$ & $-7.380$ ($0.941)$ & $ -7.721$ ($0.943$) \\
\midrule
\multicolumn{6}{c}{Computational time (seconds)} \\
\midrule
2SLS              & 16.224 & 46.661 & 95.467 & 165.675 & 245.443 \\
$\SGMM$ ($\kappa=0.5$) &  0.246 &  1.320 &  4.189 &  11.462 &  20.905 \\
$\SGMM$ ($\kappa=0.9$) &  0.234 &  1.239 &  3.744 &  10.147 &  19.170 \\
$\OGMM$ (Welford)      &  0.035 &  0.070 &  0.129 &   0.222 &   0.370 \\
$\OGMM$ (LASER)        &  0.048 &  0.119 &  0.252 &   0.503 &   0.766 \\
\bottomrule
\end{tabular}
\end{table}

\begin{table}[htbp]
\centering
\setlength{\tabcolsep}{3pt}
\small
\caption{Online instrumental variables estimation in the dependent model \ref{enum:iv-arma} with dimension $p$ and $N_{\max}=10,000p=100b$ for $b=1,\ldots,100p$.
For $\SGMM$ ($\kappa=0.5$, Welford), the reported $\mathrm{Inf}$ MSE at $p=16$ is due to numerical overflow caused by extremely large estimation errors.
}
\label{tab:51(II)2_IV_ARMA_bsize=100_armode=3}
\begin{tabular}{lcccc}
\toprule
\textbf{Method} & \textbf{$p=2$} & \textbf{$p=4$} & \textbf{$p=8$} & \textbf{$p=16$} \\
\midrule
\multicolumn{5}{c}{Log MSE (coverage rate with $\alpha=0.05$)} \\
\midrule
$\GMM$                         & $-4.894$ ($0.869$) & $-3.251$ ($0.835$) & $-1.309$ ($0.573$) & $-3.847$ ($0.650$) \\
$\SGMM$ ($\kappa=0.5$)         & $ 2.088$ ($0.894$) & $ 5.154$ ($0.396$) & $ 8.886$ ($0.108$) & $15.702$ ($0.046$) \\
$\SGMM$ ($\kappa=0.5$, Welford)& $ 1.912$ ($0.938$) & $38.557$ ($0.940$) & $176.259$ ($0.744$) & Inf ($0.695$) \\
$\OGMM$ ($LASER)$                & $-6.341$ ($0.937$) & $-5.036$ ($0.942$) & $-3.549$ ($0.948$) & $-6.189$ ($0.963$) \\
Implicit $\OGMM$ (LASER)       & $-6.341$ ($0.937$) & $-5.036$ ($0.942$) & $-3.544$ ($0.948$) & $-6.299$ ($0.967$) \\
\midrule
\multicolumn{5}{c}{Computational time (seconds)} \\
\midrule
$\GMM$                         & $10.720$  & $100.483$ & $1237.698$ & $6746.638$ \\
$\SGMM$ ($\kappa=0.5$)         & $ 0.023$  & $ 0.083$  & $ 0.524$   & $4.931$ \\
$\SGMM$ ($\kappa=0.5$, Welford)& $ 0.121$  & $ 0.290$  & $ 0.865$   & $4.750$ \\
$\OGMM$ (LASER)                & $ 0.020$  & $ 0.049$  & $ 0.148$   & $0.578$ \\
Implicit $\OGMM$ (LASER)       & $ 0.024$  & $ 0.055$  & $ 0.156$   & $0.611$ \\
\bottomrule
\end{tabular}
\end{table}

\subsection{Online Sargan--Hansen test} \label{sec:additional-overident}

Figure \ref{fig:overident-supp} reports the additional results of online Sargan--Hansen test at $5\%$ nominal level.
It verifies the findings in \cref{sec:online-sargan-hansen-test} for other values of $\theta_2^*$.

\begin{figure}[!t]
	\centering
	\includegraphics[width=\textwidth]{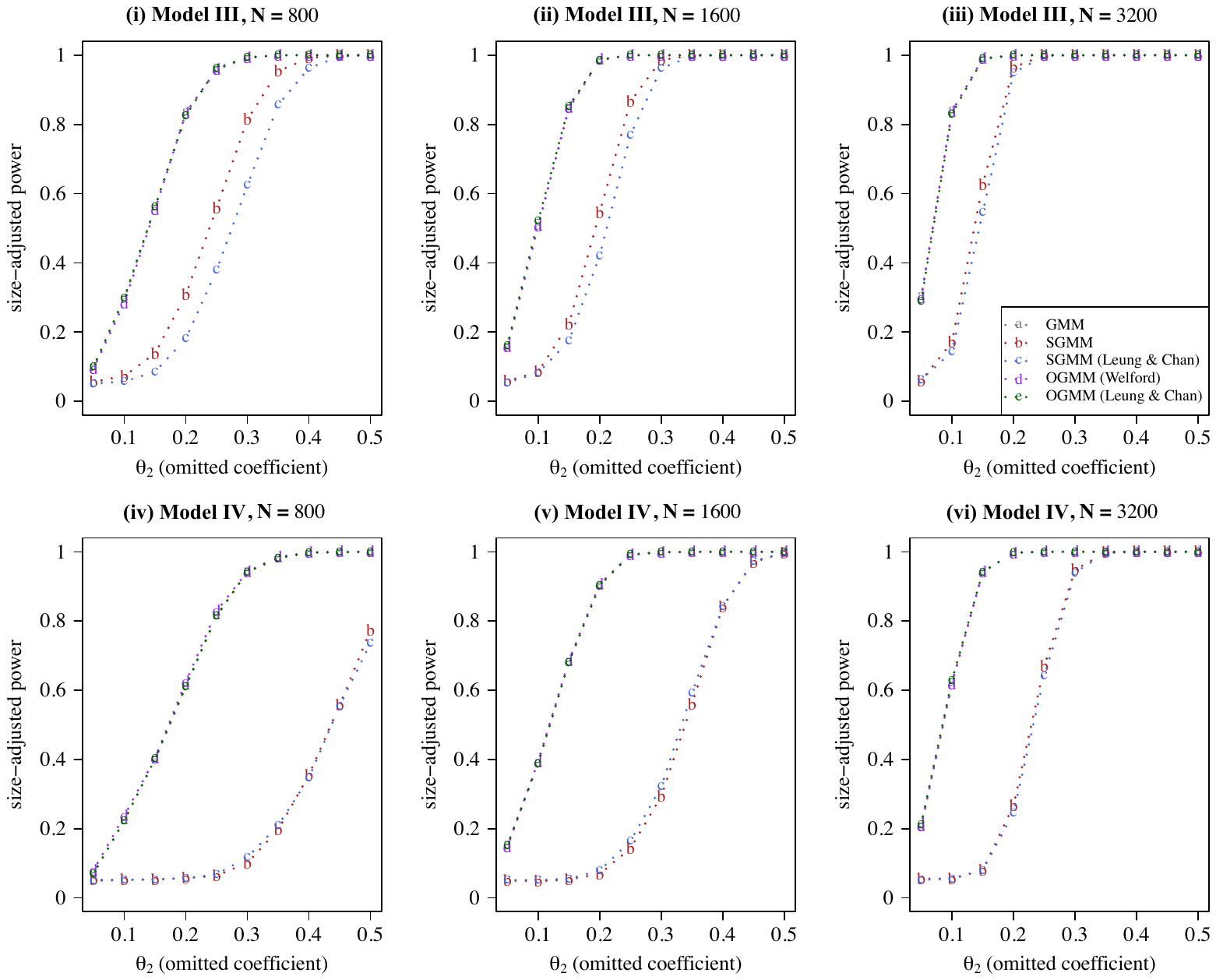}
	\caption{Online Sargan--Hansen test at $5\%$ nominal level under the independent model \ref{enum:overident-ind} (upper panel) and dependent model \ref{enum:overident-dep} (lower panel).
		The case $\theta_2^* = 0$ is omitted because Figure \ref{fig:overident} has already shown the size at all sample sizes.}
	\label{fig:overident-supp}
\end{figure}

\subsection{Online quantile regression} \label{sec:additional-quant}

First, Figure \ref{fig:reg-quant-supp} reports the additional results of online quantile regression.
It shows that implicit updates do not improve the statistical efficiency.

\begin{figure}[!t]
	\centering
	\includegraphics[width=\textwidth]{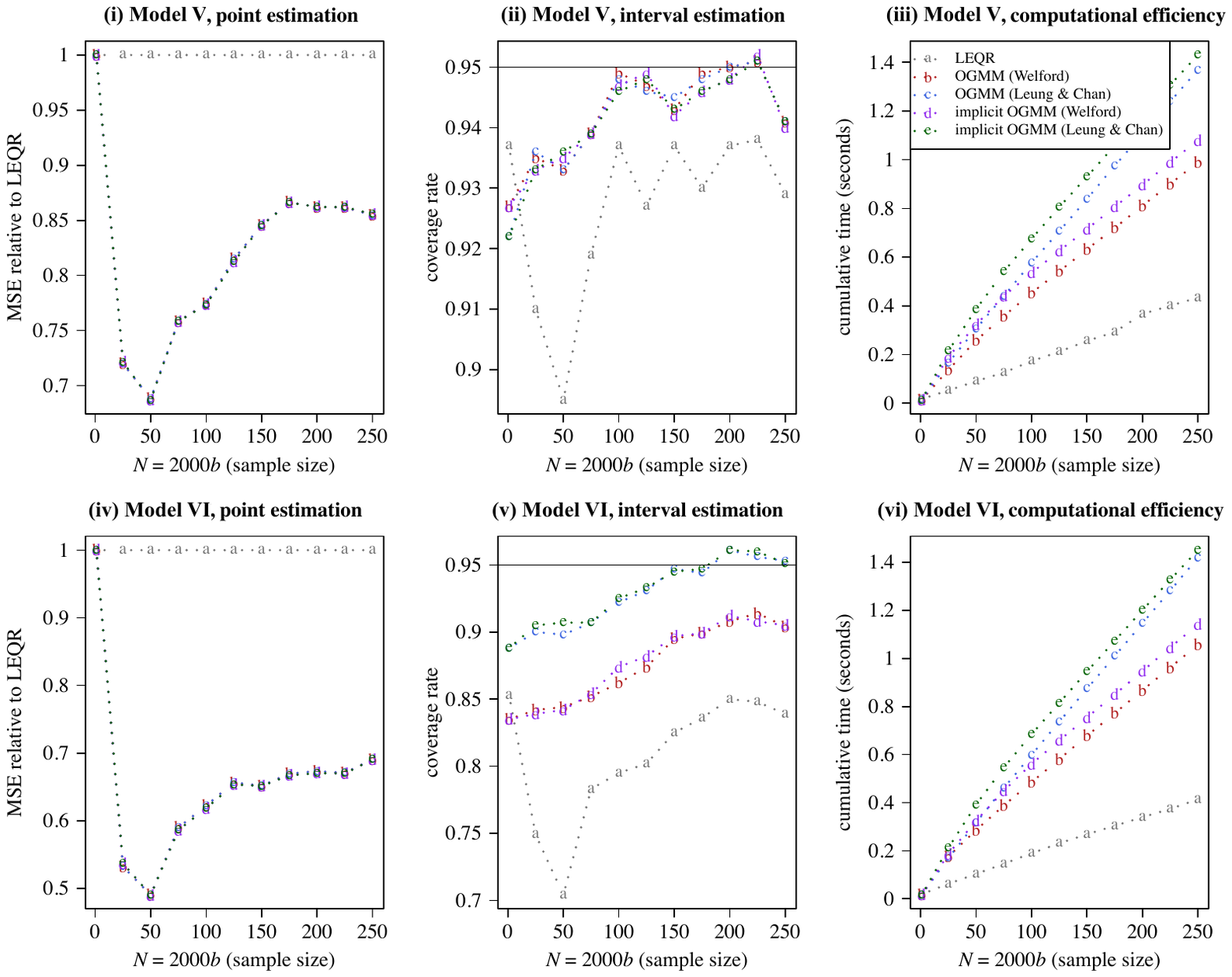}
	\caption{Online quantile regression under the independent model \ref{enum:quant-ind} (upper panel) and dependent model \ref{enum:quant-dep} (lower panel).
        }
	\label{fig:reg-quant-supp}
\end{figure}

{
Second, we report the additional results of 
online quantile regression on the point estimation
and interval estimation of the first element $\theta_{\tau,1}^*$ of $\theta_\tau^*$.
In this case, we also add an additional gradient-based method $\CQER$ for conditional linear quantile regression models under iid data by a convolution smoothed approach by \citet{he2023smoothed} and their R-package \texttt{conquer}, whose inference is based on multiplier bootstrap and percentile-type CI.
It is not an online method, so full-sample re-computations are needed when new data batch arrives.
This method is not included in the main text for $J^\T \theta_\tau^*$ as the package does not support it.
Figure~\ref{fig:53_reg_quant_single} gives the similar results for the methods used in \cref{sec:online-quantile-regression}.
Although the new method $\CQER$ provides good point estimating performance, for the dependent model \ref{enum:overident-dep}, the coverage rate is lower due to the independent weights for multiplier bootstrap inference.
Moreover, due to its offline nature, $\CQER$ is the most time-consuming method for large-scale streaming data.
}

\begin{figure}[!t]
	\centering
	\includegraphics[width=\textwidth]{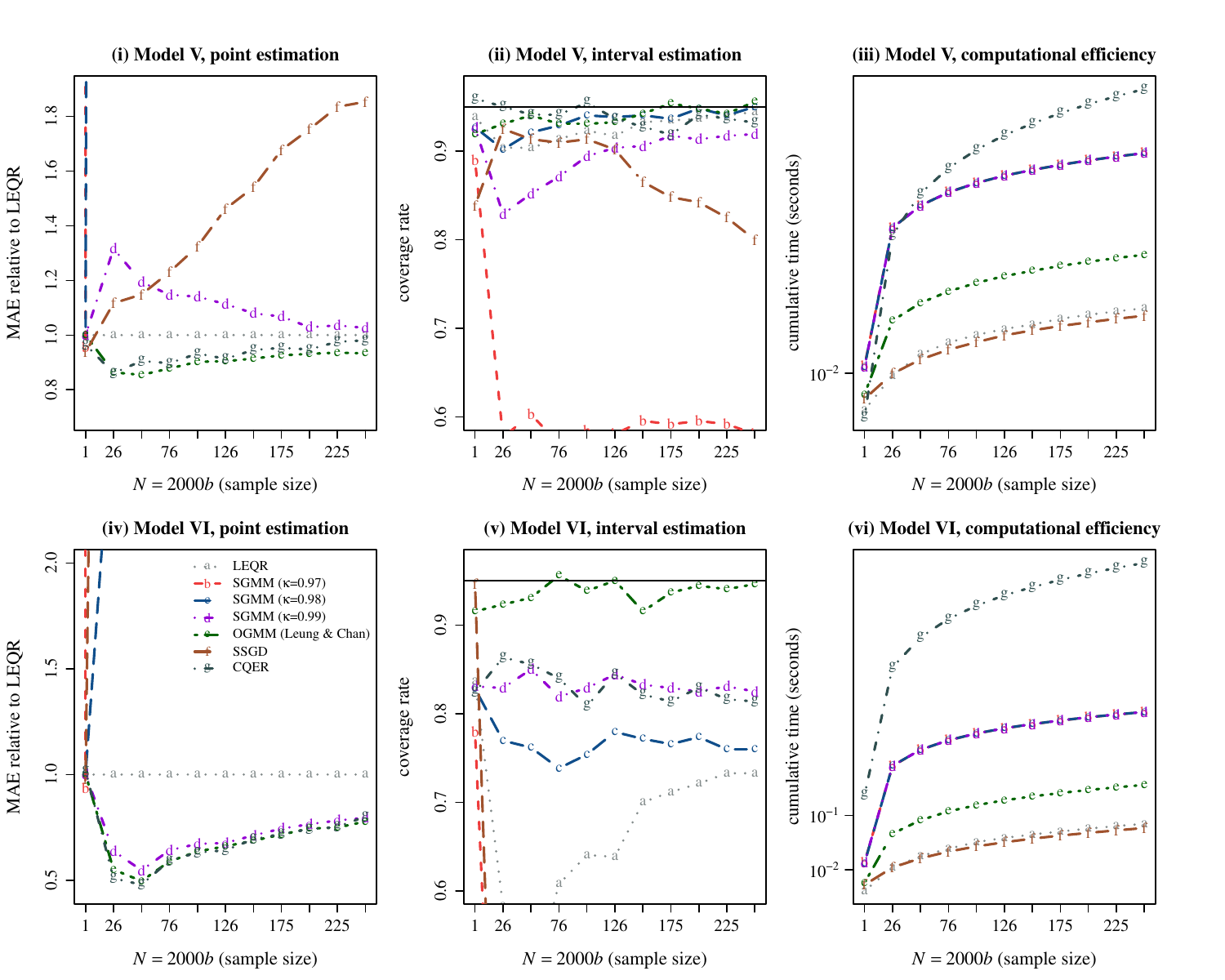}
	\caption{Online quantile regression under the independent model \ref{enum:quant-ind} (upper panel) and dependent model \ref{enum:quant-dep} (lower panel) for $ \theta_{\tau,1}^*$.
    In the upper panel, the values of $\SGMM$ ($\kappa=0.97$) and $\SGMM$ ($\kappa=0.98$) are too large to display, which are $1.00\times 10^4$ and $169.05$ for $b=101$, respectively.
    In the lower panel, the values of $\SGMM$ ($\kappa=0.97$), $\SGMM$ ($\kappa=0.98$) and $\SSGD$ are too large to display, which are $1.66\times 10^7$, $9.04$ and $24.80$ for $b=101$, respectively.
    }
	\label{fig:53_reg_quant_single}
\end{figure}

\subsection{Online anomaly detection} \label{sec:additional-stable}

Figure \ref{fig:stable-supp} reports the additional results of online anomaly detection at $5\%$ nominal level.
The findings are same as those described in \cref{sec:online-anomaly-detection}.

\begin{figure}[!t]
	\centering
	\includegraphics[width=\textwidth]{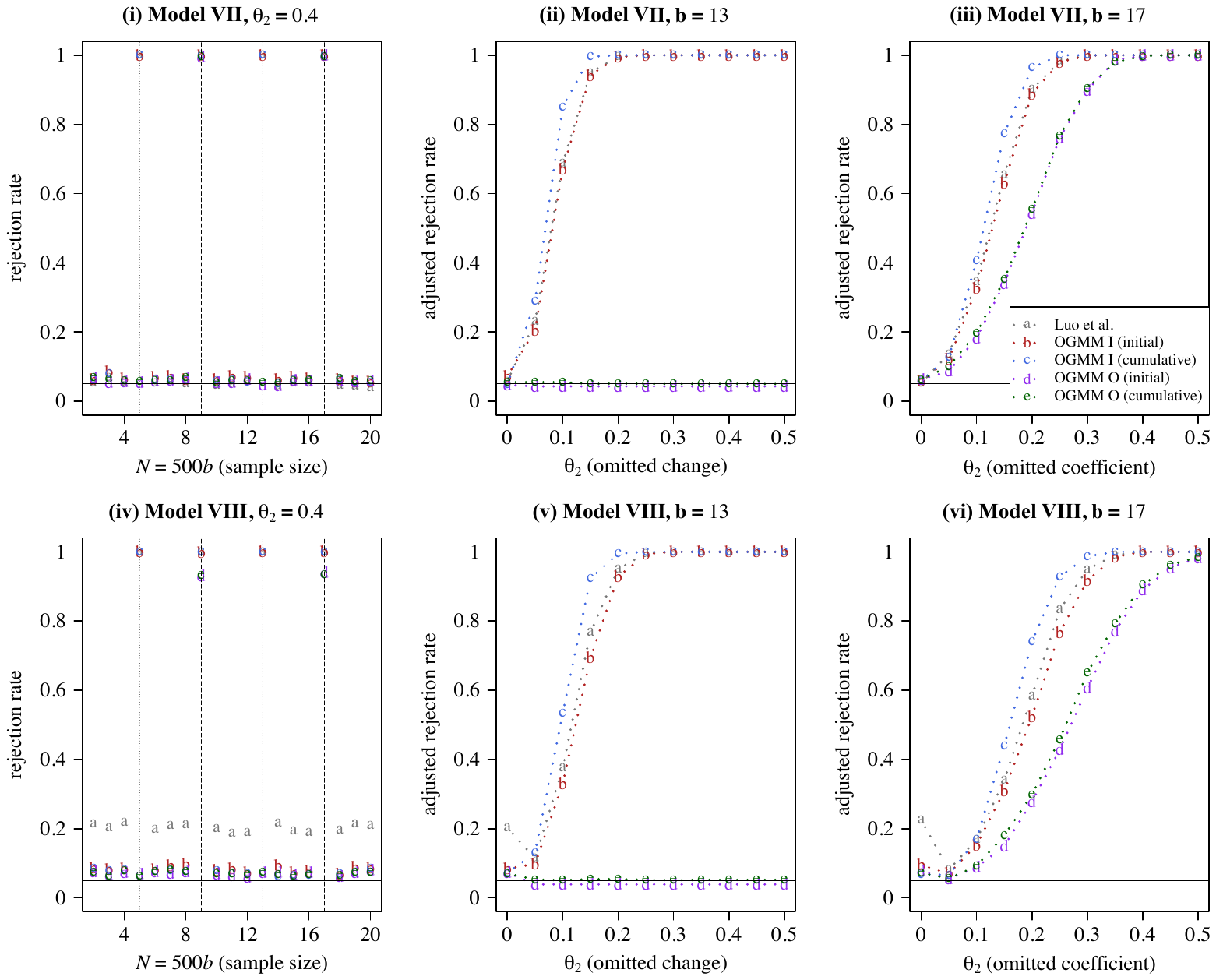}
	\caption{Online anomaly detection at $5\%$ nominal level under the independent model \ref{enum:stable-ind} (upper panel) and dependent model \ref{enum:stable-dep} (lower panel).
		The caption of Figure \ref{fig:stable} also applies here.}
	\label{fig:stable-supp}
\end{figure}

\section{Implementation details} \label{sec:implementation-details}
\subsection{Stochastic volatility modeling} \label{sec:implementation-sv}

Apart from removing observations in the weekend, we remove abnormal exchange rates at 2001-09-12 03:12:00 (due to negativity) and 2001-06-08 06:32:00 (due to an extreme jump that leads to a log 5-minute return of over $700\%$).

\subsection{Inertial sensor calibration} \label{sec:implementation-sensor}

To use $\OGMM$ with wavelet moments, we extend the pyramid algorithm for maximal overlap discrete wavelet transform using the Haar filter to an online setting in Algorithm \ref{algo:omodwt}.
The key idea is to store the necessary data in queues so that they can be retrieved and later deleted in $O(1)$ time.
As a result, Algorithm \ref{algo:omodwt} gives identical wavelets compared with its offline counterpart.
For an incoming batch of size $n_b$, the time and space complexities are $O(n_b q)$ and $O(n_b +2^q)$, respectively.

\begin{algorithm}[!t]
	\caption{Online maximal overlap discrete wavelet transform} \label{algo:omodwt}
	\SetAlgoVlined
	\DontPrintSemicolon
	\SetNlSty{texttt}{[}{]}
	\small
	\textbf{initialization}: \;
	Receive $q$ and $\{y_t\}_{t=1}^{n}$, where $n > 2^q$ \;
	Initialize empty queues $\vec{z}_0, \ldots, \vec{z}_{q-1}$ \;
	\For{$j = 0$ \KwTo $q-1$}{
		Set $\tau = 2^j$ \;
		\For{$t = 2^{j+1}$ \KwTo $n$}{
			Set $w_{t, j} = 0.5 y_t -0.5 y_{t-\tau}$ \tcp{Haar filter}
			Set $v_t = 0.5 y_t +0.5 y_{t-\tau}$ \tcp{pyramid algorithm}
		}
		Store $y_{n-\tau+1}, \ldots, y_n$ in $\vec{z}_j$ \tcp{for online updates after initialization}
		Set $\{y_t\}_{t=1}^{n} = \{v_t\}_{t=1}^{n}$ \;
	}
	\Begin{
		Receive $y_{n+1}$ \;
		\For{$j = 0$ \KwTo $q-1$}{
			Set $\tau = 2^j$ \;
			Pop $y_{n-\tau+1}$ from $\vec{z}_j$ \;
			Set $w_{n+1, j} = 0.5 y_{n+1} -0.5 y_{n-\tau+1}$ \;
			Push $y_{n+1}$ into $\vec{z}_j$ \;
			Set $y_{n+1} = 0.5 y_{n+1} +0.5 y_{n-\tau+1}$ \;
		}
		Set $n = n +1$ \;
	}
\end{algorithm}

In $\OGMM$ estimation, we slightly modify the arguments for long-run variance estimation compared with other examples in our paper.
Specially, we use \citeauthor{rlrv}'s \citeyearpar{rlrv} recursive long-run variance estimator with $\lambda=1$, $\phi=1$, \texttt{pilot=0} and \texttt{warm=TRUE}, which is the default setting in the R-package \texttt{rlaser}.
We also ignore $\{ w_{t, j} \}_{t=1}^{2^q-1}$ for $j = 1, \ldots, q$ since $w_{t, q}$ is only defined for $t \ge 2^q$.
For structural stability testing with \eqref{eq:stable-unrestricted}, we need to perform unrestricted $\GMM$ estimation.
This is done using \texttt{optim} in the R-package \texttt{stats} with \texttt{method="L-BFGS-B"}, \texttt{lower = c(rep(c(0.0001, 1e-13), 2), 1e-13)} and \texttt{upper = c(rep(c(0.9999, 1e-5), 2), 1e-5)}, where the order of parameter is $\theta^* = (\rho_1, \sigma_1^2, \rho_2, \sigma_2^2, \sigma_3^2)^\T$.

Finally, we remark that the R-package \texttt{gmwm} is slightly modified in our paper; see \url{https://github.com/hemanlmf/gmwm/tree/gmwm2-improvement}.
We add the possibility to specify the number of moment condition $q$ as the original version fixes $q = \lfloor \log_2 N \rfloor$, where $N$ is the sample size.
We also include an argument to change the optimization algorithm, although the default conjugate gradients method is still used at the end.

\bibliographystyle{rss}
{\small
\setstretch{1.0}
\bibliography{references}
}

\clearpage\end{document}